\pdfoutput=1
\documentclass[sigplan,10pt]{acmart}

\startPage{1}

\usepackage{booktabs}   
\usepackage[T1]{fontenc}
\usepackage{flushend}
\usepackage[scaled=0.75]{beramono}

\usepackage{graphicx}
\usepackage{amsmath,amssymb,stmaryrd}
\usepackage[inference]{semantic}
\usepackage{nanoml}
\usepackage{xspace}
\usepackage{paralist}
\usepackage{mathpartir}
\usepackage{mathtools}
\usepackage{enumitem}
\usepackage{mathrsfs}
\usepackage{algorithm,algorithmicx}
\usepackage[noend]{algpseudocode}
\usepackage{siunitx}
\usepackage{amsthm}
\theoremstyle{definition}

\newtheorem*{example*}{Example}
\newtheorem{theorem}{Theorem}
\newtheorem{lemma}{Lemma}

\makeatletter
\protected\def\ccell#1#{%
  \kern-\fboxsep
  \@ccell{#1}%
}
\def\@ccell#1#2#3{%
  \colorbox#1{#2}{#3}%
  \kern-\fboxsep
}
\makeatother

\usepackage[all=normal,mathspacing,mathdisplays,floats,paragraphs,wordspacing,charwidths]{savetrees}

\usepackage{multirow}
\usepackage{hhline}

\newcommand{\head}[1]{\emph{#1}}
\newif\iflong
\longtrue

\usepackage{calc}
\usepackage{ifthen}

\newcommand{\li}[2][XX]{%
   \ifthenelse{\equal{#1}{XX}}%
      {\ensuremath{L(#2)}}%
            {\ensuremath{L^{#1}(#2)}}}

\newlength{\rWidth}

\newcommand{\fun}[3][XX]{%
   \ifthenelse{\equal{#1}{XX}}%
      {\ensuremath{#2 {\rightarrow} #3}}%
      { \settowidth{\rWidth}{\ensuremath{#1}}%
        \ensuremath{ #2\hspace{.1em} {\xrightarrow{\hspace{\rWidth}}\hspace{-1.1\rWidth}}^%
         {#1}
         \hspace{0.3\rWidth}\hspace{.1em} #3}}}

\newcommand{\Rule}[4][]{\ensuremath{\inferrule*[lab={\footnotesize{(#2)}},#1]{#3}{#4}}}
\def\MathparLineskip{\lineskip=0.29cm}

\newcommand{\p}[1]{\ensuremath{(#1)}}
\newcommand{\sq}[1]{\ensuremath{[#1]}}
\newcommand{\set}[1]{\ensuremath{\{#1\}}}
\newcommand{\tuple}[1]{\ensuremath{\langle#1\rangle}}
\newcommand{\interp}[1]{\ensuremath{\llbracket #1 \rrbracket}}
\newcommand{\interps}[1]{\ensuremath{\llparenthesis #1 \rrparenthesis}}

\newcommand{\defeq}{\overset{\underset{\text{def}}{}}{=}}

\newcommand{\bbB}{\ensuremath{\mathbb{B}}}
\newcommand{\bbN}{\ensuremath{\mathbb{N}}}
\newcommand{\bbZ}{\ensuremath{\mathbb{Z}}}
\newcommand{\calI}{\ensuremath{\mathcal{I}}}
\newcommand{\calJ}{\ensuremath{\mathcal{J}}}
\newcommand{\calT}{\ensuremath{\mathcal{T}}}
\newcommand{\scrB}{\ensuremath{\mathscr{B}}}

\newcommand{\subst}[3]{\ensuremath{\sq{#1/#2}#3}}
\newcommand{\sharing}{\ensuremath{\mathrel{\curlyveedownarrow}}}
\newcommand{\potv}[3]{\ensuremath{\Phi_{#1}\p{#2 : #3}}}
\newcommand{\potc}[2]{\ensuremath{\Phi_{#1}\p{#2}}}
\newcommand{\pot}[1]{\ensuremath{\Phi\p{#1}}}
\newcommand{\condv}[3]{\ensuremath{\Psi_{#1}\p{#2:#3}}}
\newcommand{\condc}[2]{\ensuremath{\Psi_{#1}\p{#2}}}

\newcommand{\hole}[1]{\ensuremath{\mathring{#1}}}
\newcommand{\varsof}[1]{\ensuremath{\mathsf{fv}(#1)}}
\newcommand{\condcpara}[3]{\ensuremath{\Psi_{#1}^{#2}\p{#3}}}
\newcommand{\fold}[1]{\ensuremath{\mathsf{fold}\p{#1}}}

\newcommand{\etrue}{\ensuremath{\mathsf{true}}}
\newcommand{\efalse}{\ensuremath{\mathsf{false}}}
\newcommand{\econd}[3]{\ensuremath{\mathsf{if}\p{#1,#2,#3}}}
\newcommand{\elet}[3]{\ensuremath{\mathsf{let}\p{#1,#2.#3}}}
\newcommand{\enil}{\ensuremath{\mathsf{nil}}}
\newcommand{\econs}[2]{\ensuremath{\mathsf{cons}\p{#1,#2}}}
\newcommand{\ematl}[5]{\ensuremath{\mathsf{matl}\p{#1,#2,#3.#4.#5}}}
\newcommand{\eabs}[2]{\ensuremath{\lambda\p{#1.#2}}}
\newcommand{\efix}[3]{\ensuremath{\mathsf{fix}\p{#1.#2.#3}}}
\newcommand{\eapp}[2]{\ensuremath{\mathsf{app}\p{#1,#2}}}
\newcommand{\eimp}{\ensuremath{\mathsf{impossible}}\xspace}
\newcommand{\econsumename}{\ensuremath{\mathsf{tick}}\xspace}
\newcommand{\econsume}[2]{\ensuremath{\econsumename\p{#1,#2}}}
\newcommand{\ehole}{\ensuremath{\circ}}
\newcommand{\elets}[2]{\ensuremath{\mathsf{lets}\p{#1.#2}}}

\newcommand{\tprod}[2]{\ensuremath{#1 \times #2}}

\newcommand{\jval}[1]{\ensuremath{#1 \in \mathsf{Val}}}
\newcommand{\jstep}[4]{\ensuremath{\tuple{#1,#3} \mapsto \tuple{#2,#4}}}
\newcommand{\jsteps}[4]{\ensuremath{\tuple{#1,#3} \mapsto^* \tuple{#2,#4}}}
\newcommand{\jstepn}[4]{\ensuremath{\tuple{#1,#3} \mapsto^n \tuple{#2,#4}}}
\newcommand{\jsort}[3]{\ensuremath{#1 \vdash #2 \in #3}}
\newcommand{\jwftype}[2]{\ensuremath{#1 \vdash #2~\mathsf{type}}}
\newcommand{\jwfctxt}[1]{\ensuremath{\vdash #1~\mathsf{context}}}
\newcommand{\jsubty}[3]{\ensuremath{#1 \vdash #2 <: #3}}
\newcommand{\jsharing}[4]{\ensuremath{#1 \vdash #2 \sharing #3 \mid #4}}
\newcommand{\jctxsharing}[3]{\ensuremath{\vdash #1 \sharing #2 \mid #3}}
\newcommand{\jprop}[2]{\ensuremath{#1 \models #2}}
\newcommand{\jatyping}[3]{\ensuremath{#1 \vdash #2 : #3}}
\newcommand{\jstyping}[3]{\ensuremath{#1 \vdash #2 \dblcolon #3}}
\newcommand{\jctxtyping}[3][\cdot]{\ensuremath{#1 \vdash #2 \dblcolon #3}}

\newcommand{\jtunfoldnil}[3]{\ensuremath{#1 \vdash #2 \triangleleft^\mathsf{nil} #3}}
\newcommand{\jtunfoldcons}[3]{\ensuremath{#1 \vdash #2 \triangleleft^\mathsf{cons} #3}}

\newcommand{\jfill}[4]{\ensuremath{#1 \vdash #2 \dblcolon #3 \rightsquigarrow #4}}
\newcommand{\jsynth}[3]{\jfill{#1}{\ehole}{#2}{#3}}
\newcommand{\jafill}[4]{\ensuremath{#1 \vdash #2 \dblcolon #3 \overset{a}{\rightsquigarrow} #4}}
\newcommand{\jasynth}[3]{\jafill{#1}{\ehole}{#2}{#3}}
\newcommand{\jproppara}[3]{\ensuremath{#1 \models_{#3} #2}}

\newcommand{\many}[1]{\overrightarrow{#1}}
\newcommand{\vbind}[2]{\ensuremath{#1:#2}}

\newcommand{\tbool}{\ensuremath{\mathsf{bool}}}

\newcommand{\tlist}[1]{\ensuremath{L\p{#1}}}
\newcommand{\tpot}[2]{\ensuremath{{#1}^{#2}}}
\newcommand{\tsubset}[2]{\ensuremath{\set{ #1 \mid #2 }}}
\newcommand{\tarrow}[3]{\ensuremath{#1{\,:\,}#2 \rightarrow #3}}
\newcommand{\tarrowm}[4]{\ensuremath{{#4}\cdot\p{#1{\,:\,}#2 \rightarrow #3}}}
\newcommand{\trefined}[3]{\tpot{\tsubset{#1}{#2}}{#3}}

\newcommand{\bindvar}[2]{\vbind{#1}{#2}}

\newcommand{\tscalar}[1]{\ensuremath{#1~\mathsf{scalar}}}
\newcommand{\tbot}{\ensuremath{\mathsf{?}}}

\newcommand{\melems}[1]{\mathsf{elems}~#1}

\newcommand{\sol}{\ensuremath{\mathcal{L}}}

\newcommand{\tclist}[1]{\ensuremath{CL\p{#1}}}
\newcommand{\tilist}[1]{\ensuremath{SL\p{#1}}}

\newcommand{\cegissol}{\ensuremath{\mathcal{C}}}
\newcommand{\cegisex}{\ensuremath{\mathcal{E}}}

\newcommand{\simplify}{\ensuremath{\mathbb{C}}}

\algrenewcommand\algorithmicrequire{\textbf{Input:}}
\algrenewcommand\algorithmicensure{\textbf{Output:}}

\begin{document}

\title{Resource-Guided Program Synthesis}
\iflong
\subtitle{Extended Version}
\fi


\author{Tristan Knoth}
\affiliation{
  \institution{University of California, San Diego}
}
\email{tknoth@ucsd.edu}

\author{Di Wang}
\affiliation{
  \institution{Carnegie Mellon University}
}
\email{diw3@cs.cmu.edu}

\author{Nadia Polikarpova}
\affiliation{
  \institution{University of California, San Diego}
}
\email{npolikarpova@ucsd.edu}

\author{Jan Hoffmann}
\affiliation{
  \institution{Carnegie Mellon University}
}
\email{jhoffmann@cmu.edu}

\newcommand{\custompar}[1]{\parskip 0pt \textbf{\textit{#1}}}
\newcommand{\Implies}{\Rightarrow}
\renewcommand{\And}{\wedge}
\newcommand{\Or}{\vee}
\newcommand{\Subt}{<:}
\newcommand{\Consi}{\mathrel{\mathrm{\wedge:}}}
\newcommand{\HasT}{\;\mathrel{\mathrm{::}}\;}
\newcommand{\App}[2]{{#1}\ {#2}}

\newcommand{\etal}{\textit{et al.}\@\xspace}
\newcommand{\eg}{\textit{e.g.}\@\xspace}
\newcommand{\ie}{\textit{i.e.}\@\xspace}
\newcommand{\wrt}{\textit{wrt.}\@\xspace}

\newcommand{\mtimes}{\cdot}
\newcommand{\spl}{\mid \mid}

\newcommand{\todo}[1]{\textcolor{ACMRed}{(TODO: {#1})}}

\newcommand{\tname}[1]{\textsc{#1}\xspace}
\newcommand{\tool}{\tname{ReSyn}}
\newcommand{\synquid}{\tname{Synquid}}
\newcommand{\raml}{RaML}
\newcommand{\typesys}{$\mathrm{Re}^2$\xspace}
\newcommand{\typesysa}{$\mathrm{Re}^2_A$\xspace}

\newcommand{\Omit}[1]{}
\newcommand{\jan}[1]{\textbf{\textcolor{blue}{\ Jan says: #1}}}
\newcommand{\nadia}[1]{\textbf{\textcolor{purple}{\ #1}}}
\newcommand{\tristan}[1]{\textbf{\textcolor{orange}{\ #1}}}

\newcommand{\numMB}{16\xspace}
\newcommand{\numBench}{43\xspace}
\newcommand{\slowdown}{$2.5\times$\xspace}
\newcommand{\nincslowdown}{$2\times$\xspace}

\begin{abstract}
This article presents \emph{resource-guided synthesis},
a technique for synthesizing recursive programs
that satisfy both a functional specification and a symbolic resource bound. 
The technique is type-directed
and rests upon a novel \emph{type system}
that combines polymorphic refinement types with potential annotations
of automatic amortized resource analysis.
The type system enables efficient constraint-based type checking
and can express precise refinement-based resource bounds.
The proof of type soundness shows that synthesized
programs are correct by construction. 
By tightly integrating program exploration and type checking, 
the synthesizer can leverage the user-provided resource bound to guide the search,
eagerly rejecting incomplete programs that consume too many resources.
An implementation in the resource-guided synthesizer \tool 
is used to evaluate the technique on a range of recursive data structure manipulations.
The experiments show that \tool synthesizes programs 
that are asymptotically more efficient than those generated by a resource-agnostic synthesizer.
Moreover, synthesis with \tool is faster than a naive combination of synthesis and resource analysis.
\tool is also able to generate 
implementations that have a constant resource
consumption for fixed input sizes,
which can be used to mitigate side-channel attacks.
\end{abstract}


\copyrightyear{2019}
\acmYear{2019}
\setcopyright{acmlicensed}
\acmConference[PLDI '19]{Proceedings of the 40th ACM SIGPLAN Conference on Programming Language Design and Implementation}{June 22--26, 2019}{Phoenix, AZ, USA}
\acmBooktitle{Proceedings of the 40th ACM SIGPLAN Conference on Programming Language Design and Implementation (PLDI '19), June 22--26, 2019, Phoenix, AZ, USA}
\acmPrice{15.00}
\acmDOI{10.1145/3314221.3314602}
\acmISBN{978-1-4503-6712-7/19/06}


\maketitle

\section{Introduction}\label{sec:intro}

In recent years, \emph{program synthesis} has emerged as a promising
technique for automating low-level aspects of programming~\cite{GulwaniHS12,Solar-Lezama13,TorlakB14}.
Synthesis technology enables users to create programs by describing
desired behavior with
input-output examples~\cite{OseraZd15,FeserChDi15,Smith-Albarghouthi:PLDI16,FengMGDC17,FengM0DR17,FengMBD18,WangDS18,WangCB17}, 
natural language~\cite{Yaghmazadeh-al:OOPSLA17}, 
and partial or complete formal specifications~\cite{SrivastavaGF10,KneussKuKuSu13,PolikarpovaKS16,InalaPQLS17,QiuS17}.
If the input is a formal specification, synthesis algorithms can not only create
a program but also a proof that the program meets the given specification
~\cite{SrivastavaGF10,KneussKuKuSu13,PolikarpovaKS16,QiuS17}.

One of the greatest challenges in software development is to write
programs that are not only correct but also efficient with respect to
memory usage, execution time, or domain specific
resource metrics.
For this reason, automatically optimizing program performance has long been a goal of synthesis,
and several existing techniques tackle this problem 
for \emph{low-level straight-line code}~\cite{Schkufza0A13,Phothilimthana14,Sharma15,PhothilimthanaT16,Bornholt16}
or add efficient synchronization to concurrent programs~\cite{CernyCHRS11,GuptaHRST15,CernyCHRRST15,FerlesGDS18}.
However, the developed techniques are not applicable to recent advances in the 
synthesis of \emph{high-level} looping or recursive programs manipulating custom data structures
~\cite{KneussKuKuSu13,OseraZd15,FeserChDi15,PolikarpovaKS16,InalaPQLS17,QiuS17}.
These techniques lack the means to analyze and 
understand the resource usage of the synthesized programs.
Consequently, they cannot take into account the program's efficiency
and simply return the first program that arises during the search and satisfies the functional specification.

\emph{In this work, we study the problem of synthesizing high-level
  recursive programs given \emph{both} a functional specification of a
  program and a bound on its resource usage.}  A naive solution would
be to first generate a program using conventional program synthesis
and then use existing automatic static resource analyses~\cite{RAML12,TiML,CicekBGGH16} to
check whether its resource usage satisfies the bound.  Note, however,
that for recursive programs, both synthesis and resource analysis 
are undecidable in theory and expensive in practice.  
Instead, in this paper we propose
\emph{resource-guided synthesis}: an approach that tightly integrates
program synthesis and resource analysis, and uses the resource bound
to guide the synthesis process, generating programs that are efficient
by construction.
%

\vspace{-.8ex}
\paragraph{Type-Driven Synthesis}

\emph{In a nutshell, the idea of this work is to combine
  \emph{type-driven program synthesis}, pioneered in the work on
  \synquid~\cite{PolikarpovaKS16}, with type-based automatic amortized
  resource analysis
  (AARA)~\cite{Jost03,Jost10,HoffmannAH10,HoffmannW15} as implemented
  in Resource Aware ML (RaML)~\cite{ramlWeb}}. 
%
Type-driven synthesis and AARA are a perfect match because they are
both based on decidable, constraint-based type systems that can be
easily checked with off-the-shelf constraint solvers.

In \synquid, program specifications
are written as \emph{refinement types}~\cite{VazouRoJh13,KnowlesF09}. The key to
efficient synthesis is \emph{round-trip type checking},
which uses an SMT solver to aggressively prune the search space by rejecting partial
programs that do not meet the specification (see
\autoref{sec:background:synquid}).  Until now, types have only
been used in the context of synthesis to specify functional
properties.

AARA is a type-based technique for automatically deriving symbolic
resource bounds for functional programs. 
The idea is to add resource annotations to data types,
in order to specify a potential function that
maps values of that type to non-negative numbers. 
The type system
ensures that the initial potential is sufficient to cover the cost of
the evaluation. By a priori fixing the shape of the potential
functions, type inference can be reduced to linear programming (see
\autoref{sec:background:aara}).


\vspace{-.8ex}
\paragraph{The\ \ \typesys Type System}

\emph{The first contribution of this paper is a new type system, which we
dub \typesys{}---for \emph{re}finements and \emph{re}sources---that
combines polymorphic refinement types with AARA (\autoref{sec:typesys}).} 
\typesys{} is a conservative extension of \synquid's refinement type system and
RaML's affine type system with linear potential annotations.
As a result, \typesys{} can express logical
assertions that are required for effectively specifying program
synthesis problems. In addition, the type system features annotations
of numeric sort in the same refinement language to express potential
functions. Using such annotations, programmers can express precise resource
bounds that go beyond the template potential functions of RaML. 

The features that distinguish \typesys{} from other refinement-based type
systems for resource analysis~\cite{TiML,CicekBGGH16,Radicek18} are 
\begin{inparaenum}[(1)] 
\item the combination of logical and quantitative refinements
and
\item the use of AARA, 
which simplifies resource constraints
and naturally applies to non-monotone resources like memory that can become
available during the execution.
\end{inparaenum}
These features also pose nontrivial technical challenges:
the interaction between substructural and dependent types is known to be tricky~\cite{Krishnaswami15,LagoG11},
while polymorphism and higher-order functions are challenging for AARA
(one solution is proposed in~\cite{Jost10}, but their treatment of polymorphism is not fully formalized).

In addition to the design of \typesys, we prove the soundness of the
type system with respect to a small-step cost semantics.  
In the
formal development, we focus on a simple call-by-value functional
language with Booleans and lists,
where type refinements are restricted to 
linear inequalities over lengths of lists.
However, we structure the formal development to emphasize that \typesys
can be extended with user-defined data types, more expressive
refinements, or non-linear potential annotations.
The proof strategy itself is a contribution of this paper. The type
soundness of the logical refinement part of the system is
inspired by TiML~\cite{TiML}.
The main novelty is the
soundness proof of the potential annotations using a small-step cost
semantics instead of RaML's big-step evaluation semantics.



\vspace{-.5ex}
\paragraph{Type-Driven Synthesis with \typesys}

\emph{The second contribution of this paper is a resource-guided
  synthesis algorithm based on \typesys}.
In \autoref{sec:synthesis},
we first develop a system of \emph{synthesis rules} that prescribe how to derive well-typed programs
from \typesys types,
and prove its soundness \wrt the \typesys type system.
We then show how to algorithmically derive programs
using a combination of backtracking search
and constraint solving.
In particular this requires solving a new form of constraints we call \emph{resource constraints},
which are constrained linear inequalities over unknown numeric refinement terms.
To solve resource constraints, we develop a custom solver 
based on counter-example guided inductive synthesis~\cite{Solar-LezamaTBSS06}
and SMT~\cite{deMoura-Bjorner:TACAS08}.

\vspace{-.5ex}
\paragraph{The \tool Synthesizer}

\emph{The third contribution of this paper is the implementation and
  experimental evaluation of the first resource-aware synthesizer for recursive programs.}
We implemented our synthesis algorithm in a tool called \tool, 
which takes as input 
\begin{inparaenum}[(1)]
\item a \emph{goal type} that specifies the logical refinements and resource requirements of the program, and
\item types of \emph{components} (\ie library functions that the program may call).
\end{inparaenum}
\tool then synthesizes a program that provably meets the specification
(assuming the soundness of components).
%

To evaluate the scalability of the synthesis algorithm and
the quality of the synthesized programs,
we compare \tool with baseline \synquid
on a variety of general-purpose data structure operations,
such as 
eliminating duplicates from a list or
computing common elements between two lists.
The evaluation (\autoref{sec:eval}) shows that \tool is able to synthesize programs 
that are asymptotically more efficient than those generated by \synquid.
Moreover, the tool scales better than a naive combination of synthesis and resource analysis.



\section{Background and Overview}\label{sec:background}

This section provides the necessary background on type-driven program
synthesis (\autoref{sec:background:synquid}) and automatic resource
analysis (\autoref{sec:background:aara}).  We then describe and
motivate their combination in \typesys{} and showcase novel features
of the type system (\autoref{sec:background:re2}).  Finally, we
demonstrate how \typesys{} can be used for resource-guided synthesis
(\autoref{sec:background:resyn}).

\subsection{Type-Driven Program Synthesis}\label{sec:background:synquid}

Type-driven program synthesis~\cite{PolikarpovaKS16} is a technique for automatically generating
functional programs from their high-level specifications
expressed as \emph{refinement types}~\cite{KnowlesF09,RondonKaJh08}.
For example, a programmer might describe a function
that computes the common elements between two lists
using the following type signature:
\vspace{-.8ex}
\begin{nanoml}
common::l1:List a -> l2:List a 
  -> {_v:List a|elems _v == elems l1 *set elems l2}  
\end{nanoml}
\vspace{-.8ex}
Here, the return type of \T{common} is \emph{refined} with the predicate
\T{elems _v == elems l1 *set elems l2},
which restricts the set of elements of the output list $\nu$%
\footnote{Hereafter the bound variable of the refinement is always called $\nu$ and the binding is omitted.}
to be the intersection of the sets of elements of the two arguments.
Here \T{elems} is a user-defined logic-level function, also called \emph{measure}~\cite{KawaguchiRJ09,VazouRoJh13}.
In addition to the \emph{synthesis goal} above,
the synthesizer takes as input a \emph{component library}: 
signatures of data constructors and functions it can use. 
In our example, the library includes the list constructors \T{Nil} and \T{Cons}
and the function
\vspace{-.8ex}
\begin{nanoml}
member::x:a -> l:List a -> {Bool|_v = (x in elems l)}
\end{nanoml}
\vspace{-.8ex}
which determines whether a given value is in the list.
Given this goal and components, 
the type-driven synthesizer \synquid~\cite{PolikarpovaKS16} produces an implementation of \T{common} in \autoref{fig:inefficient-common}.

\begin{figure}[t]
  \centering
  \small
  \begin{nanoml}[numbers=left, numbersep=0pt]
  common = \l1.\l2.match l1 with Nil -> Nil
    Cons x xs -> if !(member x l2)
        then common xs l2
        else Cons x (common xs l2)
  \end{nanoml}  
  \caption{Synthesized program that computes common elements between two lists}\label{fig:inefficient-common}  
\end{figure}

\paragraph{The Synthesis Mechanism}

Type-driven synthesis works by systematically exploring the space of programs
that can be built from the component library
and validating candidate programs against the goal type
using a variant of liquid type inference~\cite{RondonKaJh08}.
To validate a program against a refinement type, 
liquid type inference generates a system of \emph{subtyping constraints}
over refinement types.
The subtyping constraints are then reduced to implications between refinement predicates.
For example, checking \T{common xs l2} in line 3 of \autoref{fig:inefficient-common}  
against the goal type reduces to validating the following implication:
\small
\begin{multline*}
(\melems{l_1} = \{x\} \cup \melems{xs}) \wedge 
(x \notin \melems{l_2}) \wedge\\
(\melems{\nu} = \melems{xs} \cap \melems{l_2})
 \implies \melems{\nu} = \melems{l_1} \cap \melems{l_2}
\end{multline*}
\normalsize
Since this formula belongs to a decidable theory of uninterpreted functions and arrays,
its validity can be checked by an SMT solver~\cite{deMoura-Bjorner:TACAS08}.
In general, the generated implications may contain unknown predicates.
In this case, type inference reduces to a system of \emph{constrained horn clauses}~\cite{Bjorner-al:15},
which can be solved via predicate abstraction.

\Omit{
\paragraph{Round-Trip Type Checking}

To make synthesis efficient,
\synquid type-checks each candidate incrementally as it is being constructed,
trying to discard an ill-typed program prefix as early as possible.
For example, while enumerating candidates for the function \T{common},
the synthesizer constructs the following prefix
(the missing program part is marked with \T{??}):
\begin{nanoml}
common = \l1.\l2.match l1 with Nil -> l2
                        Cons x xs -> ??
\end{nanoml}  
This prefix can be safely discarded---%
together with all of its numerous possible completions---%
since no version of the \T{Cons} case
can make the \T{Nil} case satisfy the goal type.
To reject program prefixes early,
\synquid extends liquid type inference to \emph{round-trip type checking},
which aggressively propagates type information top-down from the goal
and generates subtyping constraints as close as possible to the leaves. 

\paragraph{Condition Abduction}

The second key ingredient of efficient synthesis
is the \emph{condition abduction} mechanism,
which enables \synquid to generate conditionals---like the one inside \T{common}--%
without blindly enumerating all possible guards.
Instead, \synquid first generates the branch \T{common xs l2}
and uses the constraint solver to infer the weakest path condition (different from $\bot$)
that would make this branch satisfy the goal type;
if such a formula exists---%
in our example, $x \notin \melems{l_2}$---%
it is used to construct a goal type for synthesizing the guard.
\jan{I'm concerned about the length of this section. I think the two previous paragraphs
don't add a lot here. Could they maybe be moved to the synthesis section?
We could focus on how this works in the new type system.}
}

\paragraph{Synthesis and Program Efficiency}

The program in \autoref{fig:inefficient-common} is correct,
but not particularly efficient:
it runs roughly in time $n{\cdot}m$, 
where $m$ is the length of \T{l1} and $n$ is the length of \T{l2},
since it calls the \T{member} function (a linear scan) for every element of \T{l1}.
The programmer might realize that 
keeping the input lists \emph{sorted}
would enable computing common elements in linear time
by scanning the two lists in parallel.
To communicate this intent to the synthesizer,
they can define the type of (strictly) sorted lists
by augmenting a traditional list definition with a simple refinement:
\vspace{-.8ex}
\begin{nanoml}
data SList a where SNil::SList a
  SCons::x:a -> xs:SList {a|x < _v} -> SList a
\end{nanoml}
\vspace{-.8ex}
This definition says that a sorted list is either empty,
or is constructed from a head element \T{x} 
and a tail list \T{xs},
as long as \T{xs} is sorted and all its elements are larger than \T{x}.%
\footnote{Following \synquid, 
our language imposes an implicit constraint on all type variables
to support equality and ordering. Hence, they cannot be instantiated with arrow types.
This could be lifted by adding type classes.}
Given an updated synthesis goal 
(where \T{selems} is a version of \T{elems} for \T{SList})
\vspace{-.8ex}
\begin{nanoml}
common'::l1:SList a -> l2:SList a 
  -> {_v:List a|elems _v == selems l1 *set selems l2}  
\end{nanoml}
\vspace{-.8ex}
and a component library that includes
\T{List}, \T{SList}, and $<$ (but not \T{member}!),
\synquid can synthesize an efficient program shown in in \autoref{fig:efficient-common}.

\begin{figure}[t]  
  \small
  \centering
  \begin{nanoml}[numbers=left, numbersep=0pt]
  common' = \l1.\l2.match l1 with SNil -> Nil
    SCons x xs -> match l2 with SNil -> Nil
      SCons y ys -> 
        if x < y then common' xs l2
                 else if y < x then common' l1 ys
                               else Cons x (common' xs ys)
  \end{nanoml}
  \caption{A more efficient version of the program in \autoref{fig:inefficient-common} for sorted lists}\label{fig:efficient-common}
\end{figure}

However, if the programmer leaves the function \T{member} in the library,
\synquid will synthesize the inefficient implementation in \autoref{fig:inefficient-common}.
In general, \synquid explores candidate programs in the order of size
and returns the first one that satisfies the goal refinement type.
This can lead to suboptimal solutions,
especially as the component library grows larger
and allows for many functionally correct programs.
To avoid inefficient solutions, the synthesizer has
to be aware of the resource usage of the candidate programs.

\subsection{Automatic Amortized Resource Analysis}\label{sec:background:aara}

To reason about the resource usage of programs we take inspiration from 
\emph{automatic amortized resource analysis (AARA)}~\cite{Jost03,Jost10,HoffmannAH10,HoffmannW15}. 
AARA is a state-of-the-art technique for automatically deriving symbolic resource bounds on functional programs,
and is implemented for a subset of OCaml in Resource Aware ML
(RaML)~\cite{HoffmannW15,ramlWeb}.
For example, RaML 
is able to automatically derive the worst-case bound $2m + n{\cdot}m$ on the number of recursive calls
for the function \T{common} and $m+n$ for \T{common'}
~\footnote{In this section we assume for simplicity 
that the resource of interest is the number of recursive calls.
Both AARA and our type system support user-defined cost metrics (see \autoref{sec:typesys} for details).}.


\paragraph{Potential Annotations}
AARA is inspired by the \emph{potential method} for
manually analyzing the worst-case cost of a sequence of operations
\cite{tarjan85}. It uses annotated types to introduce
potential functions that map program states to non-negative
numbers. To derive a bound, we have to statically ensure that the
potential at every program state is sufficient to cover the cost of
the next transition and the potential of the following state. In this
way, we ensure that the initial potential is an upper bound on the
total cost. 

The key to making this approach effective is to closely integrate the
potential functions with data structures~\cite{Jost03,Jost10}. For
instance, in RaML the type \li[1]{\T{int}} stands for a list that contains
one unit of potential for every element. This type
defines the potential function
$\phi(\ell{:}\li[1]{\T{int}}) = 1\cdot|\ell|$. 
The potential can be used to pay for a recursive call
(or, in general, cover resource usage) 
or to assign potential to other data structures.

\paragraph{Bound Inference}
Potential annotations can be derived automatically by starting 
with a symbolic type
derivation that contains fresh variables for the potential annotations
of each type,
and applying syntax directed type rules that impose
local constraints on the annotations.
The integration of data structures and potential ensures
that these constraints are linear even for polynomial potential annotations.


\Omit{
\paragraph{Cost Metrics}

Before we discuss resource analysis, we have to define a cost metric
that defines the resource usage of programs.
AARA is not limited to a specific resource but works for user-defined
cost metrics.  This includes cost metrics that model non-monotone
resources like memory that can become available during an execution.

For simplicity, we assume a fixed cost metric that counts the number
of function applications in an evaluation in this article. Recall the
functions defined in \autoref{fig:inefficient-common} and
\autoref{fig:efficient-common}. With this cost metric, the resource
usage of the evaluation of the expression \T{common' l1 l2} is
$\max(m,n)$, where $m$ is the length of \T{l1} and $n$ is the length
of $l2$. Similarly, the resource usage of the expression \T{common l1
  l2} is $2m + n{\cdot}m$.

In the metatheory of \typesys{}, we represent cost metrics with
\T{tick} annotations that define a constant resource cost at a
particular program point. In this way, resource accounting is limited
to one syntactic form. Such ticks could be available at the user-level
to specify resource usage. However, in the context of synthesis, it
makes more sense to think of them as being part of an intermediate
representation that reflects a given cost metric.
%

\paragraph{Potential Annotations}
AARA is inspired by the \emph{potential method} for
manually analyzing the worst-case cost of a sequence of operations
\cite{tarjan85}. It uses annotated types to introduce
potential functions that map program states to non-negative
numbers. To derive a bound, we have to statically ensure that the
potential at every program state is sufficient to cover the cost of
the next transition and the potential of the following state. In this
way, we ensure that the initial potential is an upper bound on the
total cost. 

The key to making this approach effective is to closely integrate the
potential functions with data structures~\cite{Jost03,Jost10}. For
instance, the type \li[1]{\T{int}} expresses that we have one
potential unit for every element of a list. In other words, the type
defines the potential function
$\phi(\ell{:}\li[1]{\T{int}}) = 1\cdot|\ell|$. The intuition is that
we have one potential unit per element of the list that can be used to
cover the resource usage or to assign potential to other data
structures.


\paragraph{Type Inference}
To automatically derive such a type annotations, we start with a symbolic type
derivation that contains fresh variables for the potential annotations
of each type. We then apply syntax directed type rules that impose
local constraints on the annotations. 
\Omit{
For example, we would generate
the constraint $v = p$ if \T{compress} has the symbolic type as given
before and there is a call \T{compress(x)} where $\li[v]{\T{int}}$ is
the symbolic type of the variable $x$ in the type derivation. In
practice, the constraint set we derive for \T{compress} is larger (41
constraints in Resource Aware ML) but it is equivalent to the two
constraints in the given type schema. 
}
To derive a bound for a function we apply an off-the-shelf
linear-programming (LP) solver with an objective function that
requires to minimize the potential of the argument.

\paragraph{Polynomial Potential}

Maybe surprisingly, this simple idea scales well to different language
features and potential functions. It has been shown how this technique
can be applied to derive multivariate polynomials
bounds~\cite{HoffmannH10,HoffmannAH10} and lower
bounds~\cite{NgoDFH16}, and to derive bounds that depend on user-defined
inductive types~\cite{Jost10,HoffmannW15}. 
AARA is implemented for a subset of OCaml in Resource Aware ML
(RaML)~\cite{HoffmannW15,ramlWeb}. For example, RaML automatically
derives the worst-case bound $2m + n{\cdot}m$
for the function \T{common}.
}

\subsection{Bounding Resources with \typesys{}}
\label{sec:background:re2}

To reason about resource usage in type-driven synthesis, 
we integrate AARA's potential annotations and refinement types
into a novel type system that we call \typesys. 
In \typesys, a refinement type can be annotated with a \emph{potential term} $\phi$ of numeric sort,
which is drawn from the same logic as refinements.
%
Intuitively, the type $\tpot{R}{\phi}$
denotes values of refinement type $R$ with $\phi$ units of potential.
In the rest of this section we illustrate features of \typesys
on a series of examples,
and delay formal treatment to \autoref{sec:typesys}.

With potential annotations, users can specify 
that \T{common'} must run in time at most $m + n$, 
by giving it the following type signature:
\begin{nanoml}
common'::l1:SList a^1 -> l2:SList a^1 
  -> {_v:List a|elems _v == selems l1 *set selems l2}  
\end{nanoml}
%
This type assigns one unit of potential to every element of the arguments \T{l1} and \T{l2},
and hence only allows making one recursive call per element of each list.
Whenever resource annotations are omitted, the potential is implicitly zero:
for example, the elements of the result carry no potential.

Our type checker uses the following reasoning to argue that this potential is sufficient to cover the efficient implementation in \autoref{fig:efficient-common}.
%
Consider the recursive call in line 4, which has a cost of one.
Pattern-matching \T{l1} against \T{SCons x xs} transfers the potential from \T{l1} to the binders,
resulting in types $\T{x}: \tpot{a}{1}$ and $\T{xs}: \T{SList}\ \p{\tpot{\tsubset{a}{\T{x} < \nu}}{1}}$.
The unit of potential associated with \T{x}
can now be used to pay for the recursive call.
Moreover, the types of the arguments, \T{xs} and \T{l2},
match the required type $\T{SList}\ \tpot{a}{1}$,
which guarantees that the potential stored in the tail and the second list
are sufficient to cover the rest of the evaluation.
Other recursive calls are checked in a similar manner.

Importantly, the inefficient implementation in \autoref{fig:inefficient-common} would not type-check against this signature.
Assuming that \T{member} is soundly annotated with
\begin{nanoml}
member::x:a -> l:List a^1 -> {Bool|_v = (x in elems l)}
\end{nanoml}
(requiring a unit of potential per element of \T{l}),
the guard in line 2 consumes all the potential stored in \T{l2};
hence the occurrence of \T{l2} in line 3 has the type $\T{List}\ \tpot{a}{0}$,
which is not a subtype of $\T{List}\ \tpot{a}{1}$.

\paragraph{Dependent Potential Annotations}
In combination with logical refinements and parametric polymorphism, 
this simple extension to the \synquid's type system turns out to be surprisingly powerful.
Unlike in RaML, potential annotations in \typesys can be \emph{dependent},
\ie mention program variables and the special variable $\nu$.
Dependent annotations can encode fine-grained bounds, which are out of reach for RaML.
As one example, consider function \T{range a b} that builds a list of all integers between $a$ and $b$;
we can express that it takes at most $b - a$ steps
by giving the argument $b$ a type $\tpot{\tsubset{\T{Int}}{\nu \geq a}}{\nu - a}$.
As another example, consider insertion into a sorted list \T{insert x xs};
we can express that it takes at most as many steps as there are elements in $xs$ that are smaller than $x$,
by giving $xs$ the type $\T{SList}\ \tpot{\alpha}{\mathsf{ite}(\nu < x, 1, 0)}$
(\ie only assigning potential to elements that are smaller than $x$).
These fine-grained bounds are checked completely automatically in our system,
by reduction to constraints in SMT-decidable theories.

\begin{figure}[t]  
  \small
  \centering
  \begin{nanoml}
  append::xs:List a^1 -> ys:List a 
    -> {List a|len _v = len xs + len ys}

  triple::l:List Int^2 -> {List n|len _v = 3*(len l)}
  triple = \l.append l (append l l)  

  tripleSlow::l:List Int^3 -> {List n|len _v = 3*(len l)}
  tripleSlow = \l.append (append l l) l 
  \end{nanoml}
  \caption{Append three copies of a list.
  The type of \T{append} specifies that it returns a list whose length is the sum of the lengths of its arguments.
  It also requires one unit of potential on each element of the first list.
  Moreover, \T{append} has a polymorphic type and can be applied to lists with different element types,
  which is crucial for type-checking \T{tripleSlow}.}\label{fig:triple}
\end{figure}

\paragraph{Polymorphism}
Another source of expressiveness in \typesys is parametric polymorphism:
since potential annotations are attached to types,
type polymorphism gives us \emph{resource polymorphism} for free.
Consider two functions in \autoref{fig:triple}, \T{triple} and \T{tripleSlow},
which implement two different ways to append a list \T{l} to two copies of itself.
Both of them make use of a component function \T{append}, 
whose type indicates that it makes a linear traversal of its first argument.
Intuitively, \T{triple} is more efficient that \T{tripleSlow}
because in the former both calls to \T{append} traverse a list of length $n$,
whereas in the latter the outer call traverses a list of length $2{\cdot}n$.
This difference is reflected in the signatures of the two functions:
\T{tripleSlow} requires three units of potential per list element,
while \T{triple} only requires two.

Checking that \T{tripleSlow} satisfies this bound is somewhat nontrivial
because the two applications of \T{append} must have \emph{different types}:
the outer application must return \T{List Int},
while the inner application must return $\T{List}\ \T{Int}^1$
(\ie carry enough potential to be traversed by \T{append}).
RaML's monomorphic type system is unable to assign a single general type to \T{append},
which can be used at both call sites. So the function has be reanalyzed at
every (monomorphic) call site. 
\typesys, on the other hand, handles this example out of the box,
since the type variable \T{a} in the type of \T{append} can be instantiated with \T{Int} for the outer occurrence
and with $\T{Int}^1$ for the inner occurrence,
yielding the type
\begin{nanoml}
xs:List Int^2 -> ys:List Int^1 -> {List Int^1|...}
\end{nanoml}


As a final example, consider the standard \T{map} function:
\begin{nanoml}
map::(a -> b) -> List a -> List b
\end{nanoml}
Although this type has no potential annotations,
it implicitly tells us something about the resource behavior of \T{map}:
namely, that \T{map} applies a function to each list element \emph{at most once}.
This is because \T{a} can be instantiated with a type with an arbitrary amount of potential,
and the only way to pay for this potential is with a list element (which also has type \T{a}).


\Omit{
\paragraph{Linear Potential}

To bound the cost of \T{compress} using AARA, we enrich data types with potential
annotations. For example, to express the upper bound $|\ell|$ on the
cost for \T{compress(l)}, we assign the type \li[1]{\T{int}} to the
argument. The type defines the potential function
$\phi(\ell{:}\li[1]{\T{int}}) = 1\cdot|\ell|$. The intuition is that we
have one potential unit per element of the list that can be used to
cover the resource usage. 

To argue that this potential is indeed sufficient, we use the fact
that
$\phi(\ell{:}\li[1]{\T{int}}) = 1 + \phi(\T{xs}{:}\li[1]{\T{int}})$ if
$\T{xs}$ is the tail of $\ell$. In the first branch of the pattern
match, we have no potential available since the list is
empty. However, we also do not have any cost to cover. In the second
branch, we receive potential $1$ from the head of the list and can use
it to cover the cost of $1$ for the tick functions.  Moreover, the
tail $\T{xs}$ of the list has potential
$\phi(\T{xs}{:}\li[1]{\T{int}})$ which is used to cover the cost of
the recursive call. Since the argument type that is required by the
function is $\li[1]{\T{int}}$, we know that the available potential
is sufficient to cover the rest of the evaluation.
%

Such potential would be necessary to bound
the cost of the inner call of \T{compress} in a nested call
$\T{compress}(\T{compress}(\ell))$. The outer call would still be
typed as before and requires an argument of type
$\li[1]{\T{int}}$. So the inner call should have the type
$$
\T{compress} : \fun[0/0]{\li[2]{\T{int}}}{\li[1]{\T{int}}}
$$
which can be justified similarly to the previous typing. The key
observation is that the $2$ potential units of a list element of the
argument can be used to cover the cost ($1$ unit) and to pass on
potential to the result list ($1$ units). Note that we would waste a
potential unit in the \T{then} branch because we do not alter the
result that passed by the recursive call. However, we can conclude
that the \emph{worst-case} cost for the nested call to compress is
$2|\ell|$.

Consider again the function \T{common'} in
\autoref{fig:efficient-common}. The number of function calls in an
evaluation of \T{(common' l1 l2)} is bounded by
$\max(\T{l1},\T{l2})$. However, existing type systems based on AARA can
only express linear~\cite{Jost03} or polynomial
bounds~\cite{HoffmannW15}. For example, we can assign the following
type
$$
\T{common'} : \fun[0/0]{\li[0]{\T{int}}}{\fun[0/0]{\li[1]{\T{int}}}{\li[0]{\T{int}}}}
$$
It states that the number function calls evaluated by \T{common' xs}
is $0$ and that the result of the evaluation is a function $f$ so that
the expression $f \ell$ evaluates at most $|\ell|$ function calls.
This corresponds to the bound $|l2|$ for the expression \T{(common' l1
  l2)}. In this paper, we build on the linear version of
AARA~\cite{Jost10} that is not able to assign the bound $|l1|$ to the
aforementioned expression. The reason is that we do not restrict the
number of times a function is used and the type
$$\fun[0/0]{\li[1]{\T{int}}}{\fun[0/0]{\li[0]{\T{int}}}{\li[0]{\T{int}}}}$$
would allow us to create a function with linear cost that has type
$\fun[0/0]{\li[0]{\T{int}}}{\li[0]{\T{int}}}$. There are several ways
of extending AARA to allow for such types~\cite{HoffmannW15} but we
focus on the simple case for brevity in this article.


In general, we can describe the type of compress using symbolic
annotations and linear constraints.
$$
\T{compress} : \fun[r/r']{\li[p]{\T{int}}}{\li[q]{\T{int}}} \;\;\; \mid \;\;\; r \geq r' \;\; \land \;\;\;\; p \geq q + 1 
$$
The annotation $r/r'$ on the function arrow means that we require
constant potential $r$ when calling \T{compress} and that the
constant potential $r'$ is available after the call. The return
type $\li[q]{\T{int}}$ specifies potential of the resulting list. 
Such potential would be necessary to bound
the cost of the inner call of \T{compress} in a nested call
$\T{compress}(\T{compress}(\ell))$. The outer call would still be
typed as before and requires an argument of type
$\li[1]{\T{int}}$. So we would have $p = 2$, $q=1$, and $r=r'=0$
in the type of the inner call.
}

\subsection{Resource-guided Synthesis with \tool{}}
\label{sec:background:resyn}

We have extended \synquid with support for \typesys types in a new program synthesizer \tool.
Given a resource-annotated signature for \T{common'} from \autoref{sec:background:re2}
and a component library that includes \T{member},
\tool is able to synthesize the efficient implementation in \autoref{fig:efficient-common}.
%
The key to efficient synthesis
is type-checking each program candidate incrementally as it is being constructed,
and discarding an ill-typed program prefix as early as possible.
For example, while enumerating candidates for the function \T{common'},
we can safely discard the inefficient version from \autoref{fig:inefficient-common}
even \emph{before} constructing the second branch of the conditional
(because the first branch together with the guard use up too many resources).
Hence, as we explain in more detail in \autoref{sec:synthesis},
a key technical challenge in \tool
has been a tight integration of resources into \synquid's \emph{round-trip type checking} mechanism,
which aggressively propagates type information top-down from the goal
and solves constraints incrementally as they arise. 

\paragraph{Termination Checking}
In addition to making the synthesizer resource-aware,
\typesys types also subsume and generalize \synquid's termination checking mechanism.
To avoid generating diverging functions, 
\synquid uses a simple \emph{termination metric} (the tuple of function's arguments),
and checks that this metric decreases at every recursive call.
Using this metric,
\synquid is not able to synthesize the function \T{range} from \autoref{sec:background:re2},
because it requires a recursive call that decreases the \emph{difference} between the arguments, $b - a$.
In contrast, \tool need not reason explicitly about termination,
since potential annotations already encode an upper bound on the number of recursive calls.
Moreover, the flexibility of these annotations enables \tool to synthesize programs that require nontrivial termination metrics,
such as \T{range}.

\section{The \typesys Type System}\label{sec:typesys}

\newcommand{\diw}[1]{{\color{ACMOrange}DW: #1}}
\newcommand{\changebar}[1]{{#1}}

In this section, we define a subset of \typesys as a formal calculus to prove type soundness.
This subset includes Booleans that are refined by their values, and lists that are refined by their lengths.
The programs in \autoref{sec:intro} and \autoref{sec:background} use \synquid's surface syntax.
The gap from the surface language to the core calculus involves inductive types 
and refinement-level measures.
The restriction to this subset in the technical development is only for brevity and
proofs carry over to all the features of \synquid.

\paragraph{Syntax}

\autoref{fig:syntax} presents the grammar of terms in \typesys via abstract binding trees~\cite{book:PFPL16}.
The core language is basically the standard lambda calculus augmented with Booleans and lists.
A \emph{value} $\jval{v}$ is either a boolean constant, a list of values, or a function.
\changebar{%
Expressions in \typesys are in a-normal-form~\cite{LFP:SF92}, which means that syntactic forms occurring in non-tail position allow only \emph{atoms} $\hat{a} \in \mathsf{Atom}$, i.e., variables and values;
this restriction simplifies typing rules for applications,
as we explain below.
We identify a subset $\mathsf{SimpAtom}$ of $\mathsf{Atom}$ that contains atoms \emph{interpretable} in the refinement logic.
Intuitively, the value of an $a \in \mathsf{SimpAtom}$ should be either a Boolean or a list.%
}
The syntactic form $\eimp$ is introduced as a placeholder for unreachable code, e.g., the else-branch of a conditional whose predicate is always true.

The syntactic form $\econsume{c}{e_0}$ is used to specify resource usage, 
and it is intended to cost $c \in \bbZ$ units of resource and then reduce to $e_0$.
If the cost $c$ is negative, then $-c$ units of resource will become available in the system.
\econsumename terms support flexible user-defined cost metrics:
for example, to count recursive calls, the programmer may wrap every such call in $\econsume{1}{\cdot}$;
to keep track of memory consumption, they might wrap every data constructor in $\econsume{c}{\cdot}$,
where $c$ is the amount of memory that constructor allocates.

\begin{figure}[t!]
\[
\begin{array}{r@{\hspace{0.2em}}c@{\hspace{0.2em}}l}
    a & \Coloneqq & x \mid \etrue \mid \efalse \mid \enil \mid \econs{\hat{a}_h}{a_t} \\
    \hat{a} & \Coloneqq & a \mid \eabs{x}{e_0} \mid \efix{f}{x}{e_0} \\
	e & \Coloneqq & \hat{a} \mid \econd{a_0}{e_1}{e_2} \mid \ematl{a_0}{e_1}{x_h}{x_t}{e_2} \mid \eapp{\hat{a}_1}{\hat{a}_2} \\
	& \mid & \elet{e_1}{x}{e_2} \mid \eimp \mid \econsume{c}{e_0} \\
	v & \Coloneqq & \etrue \mid \efalse \mid \enil \mid \econs{v_h}{v_t} \mid \eabs{x}{e_0} \mid \efix{f}{x}{e_0}
\end{array}
\]
\vspace{-2ex}
\caption{Syntax of the core calculus}
\label{fig:syntax}
\vspace{-3ex}
\end{figure}

\begin{figure}[t!]
\[
\begin{array}{r@{\hspace{0.2em}}c@{\hspace{0.2em}}l@{\hspace{0.6em}}r@{\hspace{0.2em}}c@{\hspace{0.2em}}l}
	\multicolumn{3}{l}{\fbox{\text{Refinement}}} \\
	\psi,\phi & \Coloneqq & \multicolumn{4}{@{\hspace{0.1em}}l}{x \mid \top \mid \neg\psi \mid \psi_1\wedge \psi_2 \mid n \mid \psi_1 \le \psi_2 \mid \psi_1 + \psi_2 \mid \psi_1 = \psi_2} \\
	\multicolumn{3}{l}{\fbox{\text{Sort}}}  \\
	\Delta & \Coloneqq & \bbB \mid \bbN \mid \delta_\alpha \\
	\multicolumn{3}{l}{\fbox{\text{Base Type}}} & \multicolumn{3}{l}{\fbox{\text{Resource-Annotated Type}}}  \\
	B & \Coloneqq & \tbool \mid \tlist{T} \mid m \cdot \alpha & \quad T & \Coloneqq & \tpot{R}{\phi}   \\
	\multicolumn{3}{l}{\fbox{\text{Refinement Type}}} & \multicolumn{3}{l}{\fbox{\text{Type Schema}}} \\
	R & \Coloneqq & \tsubset{B}{\psi} \mid \tarrowm{x}{T_x}{T}{m} & \quad S & \Coloneqq & T \mid \forall\alpha. S 
\end{array}
\]
\vspace{-2ex}
\caption{Syntax of the type system}
\label{fig:syntaxre2}
\vspace{-3ex}
\end{figure}

\paragraph{Operational Semantics}
The resource usage of a program is determined by a small-step operational cost semantics.
The semantics is a standard one augmented with a \emph{resource} parameter.
A step in the evaluation judgment has the form $\jstep{e}{e'}{q}{q'}$ where $e$ and $e'$ are expressions and $q,q'  \in \bbZ^+_0$ are nonnegative integers.
%
For example, the following is the rule for $\econsume{c}{e_0}$.
\vspace{-.8ex}
\begin{mathpar}\footnotesize
	\inferrule{ }{ \jstep{\econsume{c}{e_0}}{e_0}{q}{q-c} } \vspace{-.8ex}
\end{mathpar}
The \emph{multi-step} evaluation relation $\mapsto^*$ is the reflexive transitive closure of $\mapsto$.
The judgment $\jsteps{e}{e'}{q}{q'}$ expresses that with $q$ units of available resources, $e$ evaluates to $e'$ without running out of resources and $q'$ resources are left.
Intuitively, the high-water mark resource usage of an evaluation of $e$ to $e'$ is the minimal $q$
such that $\jsteps{e}{e'}{q}{q'}$. For monotone resources like time, the cost is 
the sum of costs of all the evaluated $\mathsf{tick}$ expressions.
\changebar{%
In general, this net cost is invariant, that is, $p-p' = q-q'$
if $\jstepn{e}{e'}{p}{p'}$ and $\jstepn{e}{e'}{q}{q'}$,
where $\mapsto^n$ is the relation obtained by self-composing $\mapsto$ for $n$ times.}

\paragraph{Refinements}
We now combine \synquid's type system with AARA to reason about resource usage.
\autoref{fig:syntaxre2} shows the syntax of the \typesys type system.
Refinements $\psi$ are distinct from program terms and classified by sorts $\Delta$.
\typesys's sorts include Booleans $\bbB$, natural numbers $\bbN$, and \emph{uninterpreted symbols} $\delta_\alpha$.
Refinements can be logical formulas and linear expressions, which may reference program variables.
Logical refinements $\psi$ have sort $\bbB$, while potential annotations $\phi$ have sort $\bbN$.
\typesys interprets a variable of Boolean type as its value, list type as its length, and type variable $\alpha$ as an uninterpreted symbol with a corresponding sort $\delta_\alpha$.
\changebar{%
We use the following \emph{interpretation} $\calI(\cdot)$ to reflect interpretable atoms $a \in \mathsf{SimpAtom}$ in the refinement logic:
\[
\begin{array}{r@{\hspace{0.6em}}c@{\hspace{0.6em}}l@{\hspace{3.0em}}r@{\hspace{0.6em}}c@{\hspace{0.6em}}l}
    \calI(x) & = & x \\
	\calI(\etrue) & = & \top & \calI(\enil) & = & 0 \\
	\calI(\efalse) & = & \bot &  \calI(\econs{\_}{a_t}) & = & \calI(a_t) + 1
\end{array}
\]%
}
%

\vspace{-.8ex}
\paragraph{Types}
We classify types into four categories.
Base types $B$ include Booleans, lists and type variables.
Type variables $\alpha$ are annotated with a \emph{multiplicity} $m \in \bbZ^+_0 \cup \{\infty\}$, which denotes 
an upper bound on the number of usages of a variable like in bounded linear logic~\cite{GirardSS92}.
For example, $\tlist{2 \cdot \alpha}$ denotes a universal list whose elements can be used at most twice.

 Refinement types are \emph{subset types} and \emph{dependent arrow types}.
The inhabitants of the subset type $\tsubset{B}{\psi}$ are values of type $B$ that satisfy the refinement $\psi$.
The refinement $\psi$ is a logical predicate over program variables and a special \emph{value variable} $\nu$, which does not appear in the program and stands for the inhabitant itself.
For example, $\tsubset{\tbool}{\nu}$ is a type of $\etrue$, and $\tsubset{\tlist{\tbool}}{\nu \leq 5}$ represents Boolean lists of length at most 5.
Dependent arrow types $\tarrow{x}{T_x}{T}$ are function types whose return type may reference the formal argument $x$.
\changebar{%
As type variables, these function types are also annotated with a multiplicity $m \in \bbZ^+_0 \cup \{\infty\}$ restricting the number of times the function may be applied.%
} 

To apply the potential method of amortized analysis~\cite{kn:Tarjan85}, we need to define potentials with respect to the data structures in the program.
We introduce \emph{resource-annotated types} as a refinement type augmented with a potential annotation, written $\tpot{R}{\phi}$.
Intuitively, $\tpot{R}{\phi}$ assigns $\phi$ units of potential to values of the refinement type $R$.
The potential annotation $\phi$ may also reference the value variable $\nu$.
For example, $\tpot{\tlist{\tbool}}{5 \times \nu}$ describes Boolean lists $\ell$ with $5|\ell|$ units of potential where $|\ell|$ is the length of $\ell$.
The same potential can be expressed by assigning $5$ units of potential
to every element using the type $\tlist{\tpot{\tbool}{5}}$.

Type schemas represent (possibly) polymorphic types. Note that the type quantifier ${\forall}$ can only appear outermost in a type.

Similar to \synquid, we introduce a notion of \emph{scalar} types, which are resource-annotated base types refined by logical constraints.
\changebar{%
Intuitively, interpretable atoms are scalars and \typesys only allows the refinement-level logic to reason about values of scalar types.%
}
We will abbreviate $1 \cdot \alpha$ as $\alpha$, $\tsubset{B}{\top}$ as $B$, $\tarrowm{x}{T_x}{T}{\infty}$ as $\tarrow{x}{T_x}{T}$, and $\tpot{R}{0}$ as $R$.

\begin{figure*}[t]
  \def \MathparLineskip {\lineskip=0.2cm}
  \begin{flushleft}\small
    \fbox{$\jatyping{\Gamma}{a}{B}$}
  \end{flushleft}
  \begin{mathpar}\footnotesize
    \Rule{SimpAtom-Var}
    { \Gamma(x) = \trefined{B}{\psi}{\phi} }
    { \jatyping{\Gamma}{x}{B} }
    \and
    \Rule{SimpAtom-True}
    { }
    { \jatyping{\Gamma}{\etrue}{\tbool} }
    \and
    \Rule{SimpAtom-False}
    { }
    { \jatyping{\Gamma}{\efalse}{\tbool} }
    \and
    \Rule{SimpAtom-Nil}
    { \jwftype{\Gamma}{T} }
    { \jatyping{\Gamma}{\enil}{\tlist{T}} }
    \and
    \Rule{SimpAtom-Cons}
    { \jctxsharing{\Gamma}{\Gamma_1}{\Gamma_2} \\ \jstyping{\Gamma_1}{\hat{a}_h}{T} \\ \jatyping{\Gamma_2}{a_t}{\tlist{T}} }
    { \jatyping{\Gamma}{\econs{\hat{a}_h}{a_t}}{\tlist{T}} }
  \end{mathpar}
	\begin{flushleft}\small
		\fbox{$\jstyping{\Gamma}{e}{S}$}
	\end{flushleft}
	\begin{mathpar}\footnotesize
	\Rule{T-SimpAtom}
	{ \jatyping{\Gamma}{a}{B} }
	{ \jstyping{\Gamma}{a}{\tsubset{B}{\nu = \calI(a)}} }
	\and
		\Rule{T-Var}
		{ \Gamma(x) = S }
		{ \jstyping{\Gamma}{x}{S} }
		\and
		\Rule{T-Imp}
		{ \jprop{\Gamma}{\bot} \\ \jwftype{\Gamma}{T} }
		{ \jstyping{\Gamma}{\eimp}{T} }
		\and
		\Rule{T-Consume-P}
		{  c \ge 0 \\ \jstyping{\Gamma}{e_0}{T} }
		{ \jstyping{\Gamma,c}{\econsume{c}{e_0}}{T} }
		\and
		\Rule{T-Consume-N}
		{ c < 0 \\ \jstyping{\Gamma,-c}{e_0}{T} }
		{ \jstyping{\Gamma}{\econsume{c}{e_0}}{T} }
		\\
		\Rule{T-Cond}
		{ 
		\jatyping{\Gamma}{a_0}{\tbool} \\\\
		\jstyping{\Gamma,\calI(a_0)}{e_1}{T} \\
		\jstyping{\Gamma,\neg\calI(a_0)}{e_2}{T} }
		{ \jstyping{\Gamma}{\econd{a_0}{e_1}{e_2}}{T} }
		\and
		\Rule{T-MatL}
		{ \jctxsharing{\Gamma}{\Gamma_1}{\Gamma_2} \\
		\jwftype{\Gamma}{T'} \\
		\jatyping{\Gamma_1}{a_0}{\tlist{T}} \\\\
		\jstyping{\Gamma_2,\calI(a_0)=0 }{e_1}{T'} \\
		\jstyping{\Gamma_2,\bindvar{x_h}{T},\bindvar{x_t}{\tlist{T}},\calI(a_0)=x_t+1 }{e_2}{T'} \\
		 }
		{ \jstyping{\Gamma}{\ematl{a_0}{e_1}{x_h}{x_t}{e_2}}{T'} }
		\and
		\Rule{T-Let}
		{ \jctxsharing{\Gamma}{\Gamma_1}{\Gamma_2} \\
		\jwftype{\Gamma}{T_2} \\\\
		\jstyping{\Gamma_1}{e_1}{S_1} \\
		\jstyping{\Gamma_2,\bindvar{x}{S_1}}{e_2}{T_2}  }
		{ \jstyping{\Gamma}{\elet{e_1}{x}{e_2}}{T_2} }
		\\
		\Rule{T-App-SimpAtom}
		{ \jctxsharing{\Gamma}{\Gamma_1}{\Gamma_2} \\
		\jstyping{\Gamma_1}{\hat{a}_1}{\tarrowm{x}{\trefined{B}{\psi}{\phi}}{T}{1}} \\
		\jstyping{\Gamma_2}{a_2}{\trefined{B}{\psi}{\phi}} }
		{ \jstyping{\Gamma}{\eapp{\hat{a}_1}{a_2}}{\subst{\calI(a_2)}{x}{T}} }
		\and
		\Rule{T-App}
		{ \jctxsharing{\Gamma}{\Gamma_1}{\Gamma_2} \\
		\jstyping{\Gamma_1}{\hat{a}_1}{ \tarrowm{x}{T_x}{T}{1} } \\
		\jstyping{\Gamma_2}{\hat{a}_2}{T_x} \\
		\jwftype{\Gamma}{T} }
		{ \jstyping{\Gamma}{\eapp{\hat{a}_1}{\hat{a}_2}}{T} }
		\\
		\Rule{T-Abs}
		{ \jwftype{\Gamma}{T_x} \\
		\jstyping{\Gamma,\bindvar{x}{T_x} }{e_0}{T} \\
		\jctxsharing{\Gamma}{\Gamma}{\Gamma} }
		{ \jstyping{\Gamma}{\eabs{x}{e_0}}{ \tarrow{x}{T_x}{T}} }
		\and
		\Rule{T-Abs-Lin}
		{ \jwftype{\Gamma}{T_x} \\
		\jstyping{\Gamma,\bindvar{x}{T_x}}{e_0}{T} }
		{ \jstyping{m \cdot \Gamma}{\eabs{x}{e_0}}{\tarrowm{x}{T_x}{T}{m}} }
		\and
		\Rule{T-Fix}
		{ R = \tarrow{x}{T_x}{T} \\ \jwftype{\Gamma}{R} \\\\
		\jstyping{\Gamma,\bindvar{f}{R},\bindvar{x}{T_x}}{e_0}{T} \\
		\jctxsharing{\Gamma}{\Gamma}{\Gamma} }
		{ \jstyping{\Gamma}{\efix{f}{x}{e_0}}{R} }
 		\\
		\Rule{S-Gen}
		{ \jval{v} \\ \jstyping{\Gamma,\alpha}{v}{S} \\\\ \jsharing{\Gamma,\alpha}{S}{S}{S} }
		{ \jstyping{\Gamma}{v}{\forall\alpha.S} }
		\and
		\Rule{S-Inst}
		{ \jstyping{\Gamma}{e}{\forall\alpha.S} \\ \jwftype{\Gamma}{\tpot{\tsubset{B}{\psi}}{\phi}} }
		{ \jstyping{\Gamma}{e}{\subst{\tpot{\tsubset{B}{\psi}}{\phi}}{\alpha}{S}} }
		\and
		\Rule{S-Subtype}
		{ \jstyping{\Gamma}{e}{T_1} \\ \jsubty{\Gamma}{T_1}{T_2} }
		{ \jstyping{\Gamma}{e}{T_2} }
		\and
		\Rule{S-Transfer}
		{ \jstyping{\Gamma'}{e}{S} \\\\
		\jprop{\Gamma}{\pot{\Gamma} = \pot{\Gamma'}} }
		{ \jstyping{\Gamma}{e}{S} }
		\and
		\Rule{S-Relax}
		{ \jstyping{\Gamma}{e}{\tpot{R}{\phi}} \\ \jsort{\Gamma}{\phi'}{\bbN} }
		{ \jstyping{\Gamma,\phi'}{e}{\tpot{R}{\phi+\phi'}} }
	\end{mathpar}
	\begin{flushleft}\small
		\fbox{$\jsharing{\Gamma}{S}{S_1}{S_2}$}
	\end{flushleft}
	\begin{mathpar}\footnotesize
		\Rule{Share-Bool}
		{ }
		{ \jsharing{\Gamma}{\tbool}{\tbool}{\tbool} }
		\and
		\Rule{Share-List}
		{ \jsharing{\Gamma}{T}{T_1}{T_2} }
		{ \jsharing{\Gamma}{\tlist{T}}{\tlist{T_1}}{\tlist{T_2}} }
		\and
		\Rule{Share-TVar}
		{ \alpha \in \Gamma \\ {m = m_1 + m_2} }
		{ \jsharing{\Gamma}{m \cdot \alpha}{m_1 \cdot \alpha}{m_2 \cdot \alpha} }
		\and
		\Rule{Share-Poly}
		{ \jsharing{\Gamma,\alpha}{S}{S}{S} }
		{ \jsharing{\Gamma}{\forall\alpha.S}{\forall\alpha.S}{\forall\alpha.S} }
		\\
		\Rule{Share-Subset}
		{ \jsharing{\Gamma}{B}{B_1}{B_2} \\ \jwftype{\Gamma}{\tsubset{B}{\psi}} }
		{ \jsharing{\Gamma}{\tsubset{B}{\psi}}{\tsubset{B_1}{\psi}}{\tsubset{B_2}{\psi}} }
		\and
		\Rule{Share-Arrow}
		{  \jwftype{\Gamma}{\p{\tarrow{x}{T_x}{T}}} \\ m = m_1+m_2 }
		{ \jsharing{\Gamma}{\p{\tarrowm{x}{T_x}{T}{m}}}{\p{\tarrowm{x}{T_x}{T}{m_1}}}{\p{\tarrowm{x}{T_x}{T}{m_2}}} }
		\and
		\Rule{Share-Pot}
		{ \jsharing{\Gamma}{R}{R_1}{R_2} \\ \jprop{\Gamma,\bindvar{\nu}{R}}{\phi=\phi_1+\phi_2} }
		{ \jsharing{\Gamma}{\tpot{R}{\phi}}{\tpot{R_1}{\phi_1}}{\tpot{R_2}{\phi_2}} } 
	\end{mathpar}
	\begin{flushleft}\small
		\fbox{$\jsubty{\Gamma}{T_1}{T_2}$}
	\end{flushleft}
	\begin{mathpar}\footnotesize
		\Rule[rightskip=2ex]{Sub-List}
		{ \jsubty{\Gamma}{T_1}{T_2} }
		{ \jsubty{\Gamma}{\tlist{T_1}}{\tlist{T_2}} }
		\and
		\Rule[rightskip=2ex]{Sub-TVar}
		{ \alpha \in \Gamma \\ {m_1 \ge m_2} }
		{ \jsubty{\Gamma}{m_1 \cdot \alpha}{m_2 \cdot \alpha} }
		\and
		\Rule[rightskip=2ex]{Sub-Subset}
		{ \jsubty{\Gamma}{B_1}{B_2} \\\\ \jprop{\Gamma,\bindvar{\nu}{B_1} }{\psi_1 \implies \psi_2} }
		{ \jsubty{\Gamma}{\tsubset{B_1}{\psi_1}}{\tsubset{B_2}{\psi_2}} }
		\and
		\Rule[rightskip=2ex]{Sub-Arrow}
		{ \jsubty{\Gamma}{T_x'}{T_x} \\ \jsubty{\Gamma,\bindvar{x}{T_x'}}{T}{T'} \\ m \ge m' }
		{ \jsubty{\Gamma}{\tarrowm{x}{T_x}{T}{m}}{\tarrowm{x}{T_x'}{T'}{m'}} }
		\and
		\Rule[rightskip=2ex,leftskip=2ex]{Sub-Pot}
		{ \jsubty{\Gamma}{R_1}{R_2} \\\\ \jprop{\Gamma,\bindvar{\nu}{R_1}}{\phi_1 \ge \phi_2} }
		{ \jsubty{\Gamma}{\tpot{R_1}{\phi_1}}{\tpot{R_2}{\phi_2}} }
	\end{mathpar}
	\caption{Selected typing rules of the \typesys type system}
	\label{fig:typingrules}
\end{figure*}

\vspace{-.8ex}
\paragraph{Typing Rules}
In \typesys, the \emph{typing context} $\Gamma$ is a sequence of variable bindings $\bindvar{x}{S}$, type variables $\alpha$, path conditions $\psi$, and free potentials $\phi$.
Our type system consists of five judgments: sorting, well-formedness, subtyping, sharing, and typing.
We omit sorting and well-formedness rules and 
\iflong
include them in \autoref{sec:appendixre2}.
\else 
include them in the technical report~\cite{Techreport}.
\fi
The sorting judgment $\jsort{\Gamma}{\psi}{\Delta}$ states that a refinement $\psi$ has a sort $\Delta$ under a context $\Gamma$.
A type $S$ is said to be well-formed under a context $\Gamma$, written $\jwftype{\Gamma}{S}$, if every referenced variable in it is in the correct scope.

\autoref{fig:typingrules} presents selected typing rules for \typesys.
The typing judgment $\jstyping{\Gamma}{e}{S}$ states that the expression $e$ has type $S$ in context $\Gamma$.
The intuitive meaning is that if there is \emph{at least} the amount resources as indicated by the potential in the context $\Gamma$ then this suffices to evaluate $e$ to a value $v$, and after the evaluation there are \emph{at least} as many resources available as indicated by the potential in $S$.
\changebar{%
The auxiliary typing judgment $\jatyping{\Gamma}{a}{B}$ assigns base types to interpretable atoms.
Atomic typing is useful in the rule \textsc{(T-SimpAtom)}, 
which uses the interpretation $\calI(\cdot)$ to derive a most precise refinement type for interpretable atoms.
}

The \emph{subtyping} judgment $\jsubty{\Gamma}{T_1}{T_2}$ is defined in a standard way, with the extra requirement that the potential in $T_1$ should be greater than or equal to that in $T_2$.
Subtyping is often used to ``forget'' some program variables in the type to ensure the result type does not reference any locally introduced variable, e.g., the result type of $\elet{e_1}{x}{e_2}$ cannot have $x$ in it and the result type of $\ematl{a_0}{e_1}{x_h}{x_t}{e_2}$ cannot reference $x_h$ or $ x_t$.

To reason about logical refinements, we introduce \emph{validity checking}, written $\jprop{\Gamma}{\psi}$, to state that a logical refinement $\psi$ is always true under any instance of the context $\Gamma$.
The validity checking relation is established upon a denotational semantics for refinements.
Validity checking in \typesys is decidable because it can be reduced to Presburger arithmetic.
The full development of validity checking 
\iflong
is included in \autoref{sec:appendixvalidity}.
\else 
is included in the technical report~\cite{Techreport}.
\fi

\changebar{%
We reason about inductive invariants for lists in rule \textsc{(T-MatL)}, using interpretation $\calI(\cdot)$.
In our formalization, lists are refined by their length thus the invariants are: (i) $\enil$ has length $0$, and (ii) the length of $\econs{\_}{a_t}$ is the length of $a_t$ plus one.
The type system can be easily enriched with more refinements and data types (e.g., the elements of a list are the union of its head and those of its tail) by updating the interpretation $\calI(\cdot)$ as well as the premises of rule \textsc{(T-MatL)}.%
}

Finally, notable are the two typing rules for applications:
\textsc{(T-App)} and \textsc{(T-App-SimpAtom)}. 
In the former case,
the function return type $T$ does not mention $x$,
and hence can be directly used as the type of the application
(this is the case \eg for all higher-order applications,
since our well-formedness rules prevent functions from appearing in refinements).
In the latter case,
$T$ mentions $x$,
but luckily any argument of a scalar type must be a simple atom $a$,
so we can substitute $x$ with its interpretation $\calI(a)$.
The ability to derive precise types for dependent applications
motivates the use of a-normal-form in \typesys.

\paragraph{Resources}

The rule \textsc{(T-Consume-P)} states that an expression \econsume{c}{e_0}
is only well-typed in a context that contains a free potential term $c$.
To transform the context into this form, we can use the rule \textsc{(S-Transfer)}
to transfer potential within the context between variable types and free potential terms,
as long as we can prove that the total amount of potential remains the same.
For example, the combination of \textsc{(S-Transfer)} and \textsc{(S-Relax)}
allows us to derive both $\jstyping{x:\tpot{\tbool}{1}}{x}{\tpot{\tbool}{1}}$
and $\jstyping{x:\tpot{\tbool}{1}}{\econsume{1}{x}}{\tbool}$
(but not $\jstyping{x:\tpot{\tbool}{1}}{\econsume{1}{x}}{\tpot{\tbool}{1}}$).


The typing rules of \typesys form an \emph{affine} type system~\cite{kn:Walker02}.
To use a program variable multiple times, we have to introduce explicit \emph{sharing} to ensure that the program cannot gain potential.
The sharing judgment $\jsharing{\Gamma}{S}{S_1}{S_2}$ means that in the context $\Gamma$, the potential indicated by $S$ is apportioned into two parts to be associated with $S_1$ and $S_2$.
We extend this notion to \emph{context sharing}, written $\jctxsharing{\Gamma}{\Gamma_1}{\Gamma_2}$, which states that $\Gamma_1,\Gamma_2$ has the same sequence of bindings as $\Gamma$, but the potentials of type bindings in $\Gamma$ are shared point-wise, and the free potentials in the $\Gamma$ are also split. 
%
A special context sharing $\jctxsharing{\Gamma}{\Gamma}{\Gamma}$ is used in the typing rules \textsc{(T-Abs)} and \textsc{(T-Fix)} for functions.
The self-sharing indicates that the function can only reference potential-free free variables in the context.
This is also used to ensure that the program cannot gain more potential through free variables by applying the same function multiple times.

\changebar{%
Restricting functions to be defined under potential-free contexts is undesirable in some situations.
For example, a curried function of type $\tarrow{x}{T_x}{\tarrow{y}{T_y}{T}}$ might require nonzero units of potential on its first argument $x$, which is not allowed by rule \textsc{(T-Abs)} or \textsc{(T-Fix)} on the inner function type $\tarrow{y}{T_y}{T}$.
We introduce another rule \textsc{(T-Abs-Lin)} to relax the restriction.
The rule associates a multiplicity $m$ with the function type, which denotes the number of times that the function could be applied.
Instead of context self-sharing, we require the potential in the context to be enough for $m$ function applications.
Note that in \tool's surface syntax used in the \autoref{sec:background},
every curried function type implicitly has multiplicity 1 on the inner function: $\tarrow{x}{T_x}{\tarrowm{y}{T_y}{T}{1}}$.
}

\paragraph{Example}\label{sec:typesys:example}


Recall the function \T{triple} from \autoref{fig:triple},
which can be written as follows in \typesys core syntax:
\vspace{-1.2ex}
\[
\begin{array}{r@{\hspace{0.5em}}c@{\hspace{0.5em}}l}
	\T{triple} & \dblcolon & \tarrow{\ell}{\tlist{\tpot{\tbool}{2}}}{\tsubset{\tlist{\tbool}}{\nu = 3 \times \ell}} \\
	\T{triple} & = & \lambda(\ell. \mathsf{let}(\eapp{\eapp{\T{append}}{\ell}}{\ell}, \ell'.\\
	& & \enskip  \eapp{\eapp{\T{append}}{\ell}}{\ell'})\vspace{-1.2ex}
\end{array}
\]
%
%
%
Next, we illustrate 
how \typesys uses the signature of \T{append}: 
\vspace{-1ex}
\[
\begin{array}{@{\hspace{0.2em}}r@{\hspace{0.2em}}c@{\hspace{0.2em}}l@{\hspace{0.2em}}}
	\T{append} & \dblcolon & \forall\alpha. \tarrow{xs}{ \tlist{\tpot{\alpha}{1}} }{ \tarrowm{ys}{\tlist{\alpha}}{\tsubset{\tlist{\alpha}}{\nu = xs + ys}}{1} }
\end{array}
\]
to 
justify the resource bound $2|\ell|$ on \T{triple}.
%
Suppose $\Gamma$ is a typing context that contains the signature of \T{append}.
The argument $\ell$ is used three times, so we need to use sharing relations  to apportion the potential of $\ell$.
We have
$\Gamma \vdash  \tlist{\tpot{\tbool}{2}}  \sharing \tlist{\tpot{\tbool}{1}} \mid \tlist{\tpot{\tbool}{1}}$,
$\Gamma \vdash  \tlist{\tpot{\tbool}{1}}  \sharing \tlist{\tpot{\tbool}{1}} \mid \tlist{\tpot{\tbool}{0}}$,
and we assign $ \tlist{\tpot{\tbool}{1}}$, $\tlist{\tpot{\tbool}{0}}$, and $ \tlist{\tpot{\tbool}{1}}$ to the three occurrences of $\ell$ respectively in the order they appear in the program.
To reason about $e_1=\eapp{\eapp{\T{append}}{\ell}}{\ell}$, 
we instantiate \T{append} with $\alpha \mapsto \tpot{\tbool}{0}$, 
inferring its type as 
\vspace{-.8ex}
\begin{gather*}
{\tarrow{xs}{\tlist{\tpot{\tbool}{1}}}{\tarrowm{ys}{\tlist{\tpot{\tbool}{0}}}{\tsubset{\tlist{\tpot{\tbool}{0}}}{\nu = xs + ys}}{1}}}\vspace{-.8ex}
\shortintertext{and by \textsc{(T-App-SimpAtom)} we derive the following:}
\jstyping{\Gamma,\ell:\tlist{\tpot{\tbool}{1}}}{e_1}{\tsubset{\tlist{\tpot{\tbool}{0}}}{\nu = \ell + \ell}}.\vspace{-.8ex}
\end{gather*}
We then can typecheck $e_2=\eapp{\eapp{\T{append}}{\ell}}{\ell'}$
with the same instantiation of \T{append}:
\[
\jstyping{\Gamma,\ell:\tlist{\tpot{\tbool}{1}},\ell':T_1}{e_2}{\tsubset{\tlist{\tpot{\tbool}{0}}}{\nu = xs + (xs + xs)}}.
\]
(where $T_1$ is the type of $e_1$).
Finally, by subtyping and the following valid judgment in the refinement logic
\vspace{-.8ex}
\begin{gather*}
\jprop{\Gamma,\ell:\tlist{\tpot{\tbool}{2}}, \nu:\tlist{\tpot{\tbool}{0}} }{\nu = \ell + (\ell + \ell) \implies \nu = 3 \times \ell},
\end{gather*}
we conclude
$\jstyping{\Gamma}{\T{triple}}{\tarrow{\ell}{\tlist{\tpot{\tbool}{2}}}{\tsubset{\tlist{\tbool}}{\nu = 3 \times \ell}}}$.

\paragraph{Soundness}

The type soundness for \typesys is based on progress and preservation.
The progress theorem states that if we derive a bound $q$ for an
expression $e$ with the type system and $p\geq q$ resources are
available, then $\tuple{e,p}$ can make a step if $e$ is not a value. In
this way, progress shows that resource bounds are indeed bounds on the
high-water mark of the resource usage since states $\tuple{e,p}$ in
the small step semantics can be stuck based on resource usage if, for
instance, $p=0$ and $e = \econsume{1}{e'}$.

\begin{theorem}[Progress]
	If $\jstyping{q}{e}{S}$ and $p \ge q$, then either $\jval{e}$ or there exist $e'$ and $p'$ such that $\jstep{e}{e'}{p}{p'}$.
\end{theorem}
\begin{proof}
	By strengthening the assumption to $\jstyping{\Gamma}{e}{S}$ where $\Gamma$ is a sequence of type variables and free potentials, and then induction on $\jstyping{\Gamma}{e}{S}$.
\end{proof}

The preservation theorem accounts for resource consumption by relating
the left over resources after a computation to the type judgment of the
new term.

\begin{theorem}[Preservation]
	If $\jstyping{q}{e}{S}$, $p \ge q$ and $\jstep{e}{e'}{p}{p'}$, then $\jstyping{p'}{e'}{S}$.
\end{theorem}
\begin{proof}
	By strengthening the assumption to $\jstyping{\Gamma}{e}{S}$ where $\Gamma$ is a sequence of free potentials, and then induction on $\jstyping{\Gamma}{e}{S}$, followed by inversion on the evaluation judgment $\jstep{e}{e'}{p}{p'}$.
\end{proof}

The proof of preservation makes use of the following crucial substitution lemma.

\begin{lemma}[Substitution]
	If $\jstyping{\Gamma_1,x:\trefined{B}{\psi}{\phi},\Gamma'}{e}{S}$, $\jstyping{\Gamma_2}{t}{\trefined{B}{\psi}{\phi}}$, $\jval{t}$ and $\jctxsharing{\Gamma}{\Gamma_1}{\Gamma_2}$, then $\jstyping{\Gamma,\subst{\calI(t)}{x}{\Gamma'}}{\subst{t}{x}{e}}{\subst{\calI(t)}{x}{S}}$.
\end{lemma}
\begin{proof}
	By induction on $\jstyping{\Gamma_1,x:\trefined{B}{\psi}{\phi},\Gamma'}{e}{S}$.
\end{proof}

Since we found the purely syntactic soundness statement about results
of computations (they are well-typed values) somewhat unsatisfactory, 
we also introduced a denotational notation of consistency. 
	For example, a list of values $\ell = [v_1,\cdots,v_n]$ is consistent with $\jstyping{q}{\ell}{\tpot{\tlist{\tsubset{\tbool}{\neg\nu}}}{\nu+5}}$, if $q \ge n+5$ and each value $v_i$ of the list is $\efalse$.
%
We then show that well-typed values are \emph{consistent} with their typing judgement.

\begin{lemma}[Consistency]
  If $\jstyping{q}{v}{S}$, then $v$ satisfies the conditions indicated
  by $S$ and $q$ is greater than or equal to the potential stored in
  $v$ with respect to $S$.
\end{lemma}

As a result, we derive the following theorem.

\begin{theorem}[Soundness]
  If $\jstyping{q}{e}{S}$ and $p \ge q$ the either
  \begin{itemize}
  \item $\jsteps{e}{v}{p}{p'}$ and $v$ is consistent with $\jstyping{p'}{v}{S}$ or
  \item for every $n$ there is $\tuple{e',p'}$ such that $\jstepn{e}{e'}{p}{p'}$.
  \end{itemize}
\end{theorem}

Complete proofs can be 
\iflong
found in \autoref{sec:appendixproofs}.
\else 
found in the technical report~\cite{Techreport}.
\fi



\paragraph{Inductive Datatypes and Measures}
We can generalize our development of list types
for inductive types $\mu X. \many{C : \tprod{T}{X^k}}$, where $C$ is the constructor name, $T$ is the element type that does not contain $X$, and $X^k$ is the $k$-element product type $X \times X \times \cdots \times X$.
\changebar{%
The introduction rules and elimination rules are almost the same as \textsc{(T-Nil)}, \textsc{(T-Cons)} and \textsc{(T-MatL)}, respectively, except that we need to capture inductive invariants for each constructor $C$ in the rules correspondingly.%
}
In \synquid, these invariants are specified by inductive \emph{measures} that map values to refinements.
We can introduce new sorting rules for inductive types to embed values as their related measures in the refinement logic.


\paragraph{Constant Resource}
Our type system infers \emph{upper bounds} on resource usage. Recently, AARA has been generalized to verify \emph{constant-resource} behavior~\cite{SP:NDF17}.
A program is said to be constant-resource if its executions on inputs of the same size consume the same amount of resource.
We can adapt the technique in~\cite{SP:NDF17} to \typesys by
\begin{inparaenum}[(i)] 
\item changing the subtyping rules to keep potentials invariant 
(\ie replacing $\ge$ with $=$ in \textsc{(Sub-TVar)}, \textsc{(Sub-Arrow)}, \textsc{(Sub-Pot)}), and 
\item changing the rule \textsc{(Simp-Atom-Var)} to require $\phi = 0$.
\end{inparaenum}
Based on the modified type system, 
our synthesis algorithm can also synthesize constant-time implementations (see \autoref{sec:eval:micro} for more details).


\section{Type-Driven Synthesis with \typesys}\label{sec:synthesis}

In this section, we first show how to turn the type checking rules of \typesys into \emph{synthesis rules},
and then leverage these rules to develop a \emph{synthesis algorithm}.

\subsection{Synthesis Rules}\label{sec:synthesis:rules}


\begin{figure}[t!]
\[
\begin{array}{r@{\hspace{0.2em}}c@{\hspace{0.2em}}l}
  D & \Coloneqq & \cdot \mid D; x \gets e \\ 
	\hole{e} & \Coloneqq & e \mid \ehole \mid \eapp{x}{\ehole} 
  \mid \econd{x}{\ehole}{\ehole} \mid \ematl{x}{\ehole}{x_h}{x_t}{\ehole}
  \mid \elets{D}{\hole{e}}\\
  T & \Coloneqq & \tpot{R}{\phi} \mid \tbot
\end{array}
\]
\vspace{-2ex}
\caption{Extended syntax}
\label{fig:syntax-synth}
\vspace{-3ex}
\end{figure}

\paragraph{Extended Syntax}
To express synthesis rules, we extend \typesys with a new syntactic form \hole{e}
for \emph{expression templates}.
As shown in \autoref{fig:syntax-synth},
templates are expressions that can contain holes \ehole\ in certain positions.
The \emph{flat let} form \elets{D}{\hole{e}}, where $D$ is a sequence of bindings, 
is a shortcut for a nest of let-expressions \elet{x_1}{d_1}{\ldots \elet{x_n}{d_n}{\hole{e}}};
we write $\mathsf{fold}\p{\elets{D}{e}}$ to convert a flat let (without holes) back to the original syntax.
We also extend the language of types with an \emph{unknown type} \tbot,
which is used to build partially defined goal types, as explained below.

\begin{figure*}[t]
  \def \MathparLineskip {\lineskip=0.2cm}
	\begin{flushleft}\small
		\fbox{$\jfill{\Gamma}{\hole{e}}{S}{e}$}
	\end{flushleft}
	\begin{mathpar}\footnotesize
    \Rule{Syn-Gen}
    { \jsharing{\Gamma,\alpha}{S}{S}{S} \\ \jsynth{\Gamma,\alpha}{S}{e} }
    { \jsynth{\Gamma}{\forall\alpha.S}{e} }
    \and      
    \Rule{Syn-Fix}
    { \jsynth{\Gamma,f:\p{\tarrow{x}{T_x}{T}},x:T_x}{T}{e} \\ \jctxsharing{\Gamma}{\Gamma}{\Gamma} }
    { \jsynth{\Gamma}{\p{\tarrow{x}{T_x}{T}}}{\efix{f}{x}{e}} }
    \and
    \Rule{Syn-Abs}
    { \jsynth{\Gamma,x:T_x}{T}{e} }
    { \jsynth{\Gamma}{\p{\tarrowm{x}{T_x}{T}{1}}}{\eabs{x}{e}} }
    \and
    \Rule{Syn-Cond}
    { 
      \jasynth{\Gamma}{\tbool}{\elets{D}{x}} \\ 
      \jfill{\Gamma}{\elets{D}{\econd{x}{\ehole}{\ehole}}}{T}{e} }
    { \jsynth{\Gamma}{T}{e}}
    \and
    \Rule{Syn-MatL}
    { 
      \jwftype{\Gamma}{T}\\  
      \jasynth{\Gamma}{\tlist{T}}{\elets{D}{x}} \\
      \jfill{\Gamma}{\elets{D}{\ematl{x}{\ehole}{x_h}{x_t}{\ehole}}}{T}{e} }
    { \jsynth{\Gamma}{T}{e} }    
    \and
    \Rule{Fill-Cond}
    { \jatyping{\Gamma}{x}{\tbool} \\
    \jsynth{\Gamma, x  }{T}{e1} \\ 
      \jsynth{\Gamma, \neg x  }{T}{e2}  }
    { \jfill{\Gamma}{\econd{x}{\ehole}{\ehole}}{T}{\econd{x}{e1}{e2}}}	  
    \and  
    \Rule{Fill-MatL}
    {  \jctxsharing{\Gamma}{\Gamma_1}{\Gamma_2}  \\ \jatyping{\Gamma_1}{x}{\tlist{T}} \\
      \jsynth{\Gamma_2, x = 0}{T}{e1} \\ 
      \jsynth{\Gamma_2, x_h: T, x_t:\tlist{T}, x = 1 + x_t}{T}{e2}  }
    { \jfill{\Gamma}{\ematl{x}{\ehole}{x_h}{x_t}{\ehole}}{T}{\ematl{x}{e1}{x_h}{x_t}{e_2}}}	  
    \and
    \Rule{Fill-Let}
    { \jctxsharing{\Gamma}{\Gamma_1}{\Gamma_2} \\
      \jstyping{\Gamma_1}{e_1}{T_1} \\ 
      \jfill{\Gamma_2,x:T_1}{\elets{D}{\hole{e}_2}}{T}{e_2}  }
    { \jfill{\Gamma}{\elets{x \gets e_1; D}{\hole{e}_2}}{T}{\elet{e_1}{x}{e_2}} }
    \and  
    \Rule{Syn-Imp}
    { \jprop{\Gamma}{\bot} }
    { \jsynth{\Gamma}{T}{\eimp} }    
    \and  
    \Rule{Syn-Atom}
    { \jasynth{\Gamma}{T}{\elets{D}{a}}  }
    { \jsynth{\Gamma}{T}{\fold{\elets{D}{a}}} }    
	\end{mathpar}
  \begin{flushleft}\small
    \fbox{$\jafill{\Gamma}{\hole{e}}{T}{\elets{D}{a}}$}
  \end{flushleft}
  \begin{mathpar}\footnotesize
    \Rule{ASyn-Var}
    { \jstyping{\Gamma}{x}{T} }
    { \jasynth{\Gamma}{T}{\elets{\cdot}{x}} }
    \and  
    \Rule{ASyn-True}
    { \jstyping{\Gamma}{\etrue}{T} }
    { \jasynth{\Gamma}{T}{\elets{\cdot}{\etrue}} }
    \and
    \Rule{ASyn-False}
    { \jstyping{\Gamma}{\efalse}{T} }
    { \jasynth{\Gamma}{T}{\elets{\cdot}{\efalse}} }
    \and
    \Rule{ASyn-Nil}
    { \jstyping{\Gamma}{\enil}{T} }
    { \jasynth{\Gamma}{T}{\elets{\cdot}{\enil}} }
    \and    
    \Rule{ASyn-App}
    { 
      \jasynth{\Gamma}{\tarrowm{\_}{\tbot}{T}{1}}{\elets{D_1}{x}} \\    
      \jafill{\Gamma}{\elets{D_1}{\eapp{x}{\ehole}}}{T}{\elets{D}{x'}} \\
    }
    { \jasynth{\Gamma}{T}{\elets{D}{x'}} }
    \and  
    \Rule{AFill-Let}
    { \jctxsharing{\Gamma}{\Gamma_1}{\Gamma_2} \\
      \jstyping{\Gamma_1}{e_1}{T_1} \\ 
      \jafill{\Gamma_2,x:T_1}{\elets{D}{\hole{e}_2}}{T}{\elets{D_2}{a}}  }
    { \jafill{\Gamma}{\elets{x \gets e_1; D}{\hole{e}_2}}{T}{\elets{x \gets e_1; D_2}{a}} }
    \and    
    \Rule{AFill-App}
    { \jstyping{\Gamma}{x}{ \tarrowm{\_}{T_1}{T}{1} } \\  
      T_1~\textsf{non-scalar} \\
      \jsynth{\Gamma}{T_1}{\hat{a}} }
    { \jafill{\Gamma,1}{\eapp{x}{\ehole}}{T}{\elets{x' \gets \econsume{1}{\eapp{x}{\hat{a}}}}{x'}} }
    \and    
    \Rule{AFill-App-SimpAtom}
    { \jstyping{\Gamma}{x}{  \tarrowm{y}{T_1}{T'}{1}  } \\
      T_1~\textsf{scalar}  \\
      \jafill{\Gamma}{\ehole}{T_1}{\elets{D}{a}} \\  
      \jstyping{\Gamma}{\fold{ \elets{D}{\eapp{x}{a}}} }{T} }
    { \jafill{\Gamma,1}{\eapp{x}{\ehole}}{T}{\elets{D; x' \gets \econsume{1}{\eapp{x}{a}}}{x'} }}
  \end{mathpar}
	\caption{Selected synthesis rules}
	\label{fig:synthrules}
\end{figure*}

\paragraph{Synthesis for A-Normal-Form}
Our synthesis relation consists of two mutually recursive judgments:
the \emph{synthesis} judgment \jfill{\Gamma}{\hole{e}}{S}{e}
intuitively means that the template \hole{e} can be completed into an expression $e$ 
such that \jstyping{\Gamma}{e}{S};
the purpose of the auxiliary \emph{atomic synthesis} judgment is explained below.
Selected rules for both judgments are given in \autoref{fig:synthrules};
the full technical development can be found in
\iflong
Appendix~\ref{sec:appendixproofs:synth}.
\else 
the technical report~\cite{Techreport}.
\fi

The synthesis rule \textsc{(Syn-Gen)}
handles polymorphic goal types.
The rules \textsc{(Syn-Fix)} and \textsc{(Syn-Abs)} handle arrow types
and derive either a fixpoint term or an abstraction.
The rule \textsc{(Syn-Imp)} derives \eimp in an inconsistent context
(which may arise \eg in a dead branch of a pattern match).
The rest of the rules handle the common case when the goal type $T$ is scalar
and the context is consistent;
in this case the target expression can be
either a conditional, a match, or an \emph{E-term}~\cite{PolikarpovaKS16},
\ie a term made of variables, applications, and constructors. 
Special care must be taken to ensure that these expressions are in a-normal-form:
generally, a-normalizing an expression requires introducing fresh variables and let-bindings for them.
To retain completeness, our synthesis rules need to do the same:
intuitively, in addition to an expression $e$,
a rule might also need to produce a sequence of let-bindings $D$ that define fresh variables in $e$.
To this end, we introduce the \emph{atomic synthesis} judgment \jafill{\Gamma}{\hole{e}}{T}{\elets{D}{a}},
which synthesizes \emph{normalized E-terms},
where $a$ is an atom and each definition in $D$ is an application or a constructor in a-normal-form.

As an example, 
consider the rule \textsc{(Syn-Cond)} for synthesizing conditionals:
ideally, we would like to synthesize a guard $e$ of type \tbool,
and then synthesize the two branches under the assumptions that $e$ evaluates to \etrue\ and \efalse, respectively.
Recall, however, that the guard must be atomic;
hence, to synthesize a well-formed conditional,
we use atomic synthesis to produce a guard \elets{D}{x}.
Now to get a well-scoped program
we must place the whole conditional \emph{inside} the bindings $D$;
to that end, the second premise of \textsc{(Syn-Cond)} uses a nontrivial template
\elets{D}{\econd{x}{\ehole}{\ehole}}.
The rules \textsc{(Fill-Let)} and \textsc{(Fill-Cond)} 
handle this template by integrating it into the typing context and exposing the hole;
along the way \textsc{(Fill-Let)} takes care of context sharing,
which accounts for the potential consumed by the definitions in $D$.
Synthesis of matches works similarly using \text{(Syn-MatL)} and \text{(Fill-MatL)}. 

\paragraph{Atomic Synthesis}
The first four rules of atomic synthesis generate a simple atom if its type matches the goal;
the rest of the rules deal with the hardest part: normalized applications.
Consider the rule \textsc{(ASyn-App)}:
given a goal type $T$ for the application \eapp{e_1}{e_2},
we need to construct goal types for $e_1$ and $e_2$, 
to avoid enumerating them blindly.
Following \synquid's \emph{round-trip type checking} idea, 
we use the type $\tarrow{\_}{\tbot}{T}$ as the goal for $e_1$
(\ie a function from unknown type to $T$).
The subtyping rules for \tbot\
are such that $\jsubty{\Gamma}{\p{\tarrow{y}{T_1}{T_2}}}{\p{\tarrow{\_}{\tbot}{T}}}$ 
holds if $T_2$ and $T$ agree in shape and those refinements that do not mention $y$;
hence this goal type filters out those functions $e_1$ that cannot fulfill the desired goal type $T$,
independently of the choice of $e_2$.
One difference with \synquid is that the goal type for $e_1$ is \emph{linear},
reflecting that we intend to use $e_1$ only once
and allowing it to capture positive potential.

Similarly to the conditional case explained above,
the synthesized left-hand side of the application,  $e_1$, 
has the form \elets{D_1}{x},
and the argument $e_2$ must be synthesized inside the bindings $D_1$.
These bindings are processed by \textsc{(AFill-Let)},
and the actual argument synthesis happens in either \textsc{(AFill-App)} or \textsc{(AFill-SimpAtom)},
depending on whether the argument type is a scalar.
The former corresponds to a higher-order application:
here $T_1$ is an arrow type, and hence the argument cannot occur in the function's return type;
in this case, synthesizing an expression of type $T_1$
must yield an abstraction or fixpoint (since $T_1$ is an arrow),
both of which are atoms.
The latter corresponds to a first-order application:
here the return type $T'$ can mention $y$,
so after synthesizing an argument of type $T_y$,
we still need to check whether the resulting application \elets{D}{\eapp{x}{a}} has the right type $T$.
Note how both \textsc{(AFill-App)} or \textsc{(AFill-SimpAtom)}
return normalized E-terms by generating a fresh variable and binding it to an application.

\vspace{-.8ex}
\paragraph{Cost Metrics}
In the context of synthesis we cannot rely on programmer-written \econsumename terms to model cost.
Instead in our formalization we use a simple cost metric 
where each function application consumes one unit of resource;
hence every application generated by \textsc{(AFill-App)} or \textsc{(AFill-SimpAtom)}
is wrapped in $\econsume{1}{\cdot}$.
Our implementation provides more flexibility
and allows the programmer to annotate any arrow type with a non-negative cost $c$
to denote that applying a function of this type should incur cost $c$.
%


\paragraph{Soundness}

The synthesis rules 
always produce a well-typed expression
\iflong 
(proof can be found in \autoref{sec:appendixproofs:relativesound}).
\else 
(proof can be found in the technical report~\cite{Techreport}).
\fi

\vspace{-.8ex}
\begin{theorem}[Soundness of Synthesis]\label{Thm:SynthSound}
If \jsynth{\Gamma}{S}{e} then \jstyping{\Gamma}{e}{S}.
\end{theorem}

\subsection{Synthesis Algorithm}\label{sec:synthesis:algo}

In this section we discuss how to turn the declarative synthesis rules of \autoref{sec:synthesis:rules} 
into a \emph{synthesis algorithm},
which takes as input a \emph{goal type} $S$, a context $\Gamma$, and a bound $k$ on the program depth,
and either returns a program $e$ of depth at most $k$ such that $\jstyping{\Gamma}{e}{S}$,
or determines that no such program exists.
The core algorithm follows the recipe from prior work on type-driven synthesis~\cite{PolikarpovaKS16,OseraZd15}
and performs a fairly standard goal-directed backtracking proof search 
with $\jsynth{\Gamma}{e}{S}$ as the top-level goal.
%
%
In the rest of this section,
we explain how to make such proof search feasible
by reducing the core sources of non-determinism to constraint solving.

\Omit{
\paragraph{Goal-Directed Search}
First, like in prior work on synthesis of functional programs~\cite{GveroKuKuPi13,OseraZd15,PolikarpovaKS16},
we can restrict synthesis to programs in $\eta$-long form (where all functions are fully applied).
Hence, without loss of completeness
we can always apply \textsc{(Syn-Fix)} when the goal type is \tarrow{x}{T_1}{T_2}
and \textsc{(Syn-Abs)} when the goal type is \tarrowm{x}{T_1}{T_2}{1}.
Note that we also do not lose completeness by preferring a fixpoint term whenever possible,
because whenever \jstyping{\Gamma}{\eabs{x}{e}}{\tarrow{x}{T_1}{T_2}},
we also have \jstyping{\Gamma}{\efix{f}{x}{e}}{\tarrow{x}{T_1}{T_2}}.
Given a scalar goal type $T$ and a consistent context,
the synthesis algorithm first tries to generate a normalized E-term using \textsc{(Syn-Atom)},
and, failing that, tries \textsc{(Syn-Cond)}, 
and then \textsc{(Syn-MatL)}.
Atomic synthesis explores normalized E-terms in the order of depth:
it first tries the rules that generate atoms,
and failing that, 
tries to generate an application (unless the depth bound has been reached).
Finally, note that all ``hole-filling'' rules with a non-trivial template $\hole{e}$
are directed by the syntax of the template
and, in the case of \textsc{(AFill-App)} and \textsc{(AFill-App-SimpAtom)},
the types of known sub-expressions inside the template. 

\paragraph{Algorithmic Typing}
Several rules in \autoref{fig:synthrules} rely on the typing judgment in their premises.
Although the type system presented in \autoref{sec:typesys} is declarative, 
it can be transformed into an algorithmic bidirectional type system similar to the one in \synquid
(we omit the formalization in the interest of space).
As is standard for such transformation,
structural rules, such as \textsc{(S-Inst)}, \textsc{(S-Subtype)}, and \textsc{(S-Transfer)},
are integrated with syntax-directed rules.
}

\vspace{-.8ex}
\paragraph{Typing constraints}
%
The main sources of non-determinism in a synthesis derivation stem from the following premises
of synthesis and typing rules:
\begin{inparaenum}[(1)]
\item whenever a given context $\Gamma$ is shared as \jsharing{\Gamma}{\Gamma}{\Gamma_1}{\Gamma_2},
we need to guess how to apportion potential annotations in $\Gamma$;
\item whenever potential in a given context $\Gamma$ is transfered, 
we need to guess potential annotations in $\Gamma'$ such that $\pot{\Gamma} = \pot{\Gamma'}$; and finally
\item whenever \tpot{\tsubset{B}{\psi}}{\phi} is used to instantiate a type variable,
we need to guess both $\phi$ and $\psi$.
\end{inparaenum}
All three amount to inference of unknown refinement terms of either Boolean or numeric sort.
To infer these terms efficiently, we use the following constraint-based approach.
First, we build a symbolic synthesis derivation, 
which may contain \emph{unknown refinement terms} $U^{\Delta}_{\Gamma}$,
and collect all subtyping, sharing, and transfer premises from the derivation
into a system of \emph{typing constraints}.
Here $\Delta$ records the desired sort of the unknown refinement term,
and $\Gamma$ records the context in which it must be well-formed.
A \emph{solution} to a system of typing constraints,
is a map $\sol : U \to \psi$
such that for every unknown $U^{\Delta}_{\Gamma}$,
$\jsort{\Gamma}{\sol(U)}{\Delta}$ 
and substituting $\sol(U)$ for $U$ within the typing constraints yields
valid subtyping, sharing, and transfer judgments.


\begin{figure}
\small
\begin{align*}
&\fbox{\text{Subtyping constraints}}\\
&\simplify(\jsubty{\Gamma}{m_1 \cdot \alpha}{m_2 \cdot \alpha}) = \{\jprop{\Gamma}{m_1 - m_2 \ge 0}\}\\
&\simplify(\jsubty{\Gamma}{\tsubset{B_1}{\psi_1}}{\tsubset{B_2}{\psi_2}}) = \{\jprop{\Gamma,\bindvar{\nu}{B_1} }{\psi_1 \implies \psi_2}\} \cup \simplify(\jsubty{\Gamma}{B_1}{B_2})\\
&\simplify(\jsubty{\Gamma}{\tpot{R_1}{\phi_1}}{\tpot{R_2}{\phi_2}}) = \{\jprop{\Gamma,\bindvar{\nu}{R_1}}{\phi_1 - \phi_2 \ge 0}\} \cup \simplify(\jsubty{\Gamma}{R_1}{R_2})\\
&\fbox{\text{Sharing constraints}}\\
&\simplify(\jsharing{\Gamma}{m \cdot \alpha}{m_1 \cdot \alpha}{m_2 \cdot \alpha}) = \{\jprop{\Gamma}{m - (m_1 + m_2) \ge 0}, \\
& \quad\quad\jprop{\Gamma}{m_1 + m_2 - m \ge 0}\}\\
&\simplify(\jsharing{\Gamma}{\tpot{R}{\phi}}{\tpot{R_1}{\phi_1}}{\tpot{R_2}{\phi_2}}) = \{\jprop{\Gamma}{\phi - (\phi_1 + \phi_2) \ge 0}, \jprop{\Gamma}{\phi_1 + \phi_2 - \phi \ge 0}\}\\ 
& \quad\quad \cup \simplify(\jsharing{\Gamma}{R}{R_1}{R_2})\\
&\fbox{\text{Transfer constraints}}\\
&\simplify(\pot{\Gamma}=\pot{\Gamma'}) = \{\jprop{\Gamma}{\pot{\Gamma} - \pot{\Gamma'} \ge 0, \pot{\Gamma'} - \pot{\Gamma} \ge 0}\}
\end{align*}
\vspace{-5ex}
\caption{Selected cases for translating typing constraints to validity constraints.}\label{fig:simplify}
\end{figure}

\vspace{-.8ex}
\paragraph{Constraint Solving}
To solve typing constraints, the algorithm  first transforms them into validity constraints
of one of two forms:
\jprop{\Gamma}{\psi \implies \psi'} or \jprop{\Gamma}{\phi \geq 0};
the interesting cases of this translation are shown in \autoref{fig:simplify}.
\iflong
Then, using the definition of validity (\autoref{sec:appendixvalidity}),
\else 
Then, using the definition of validity (in the technical report \cite{Techreport}),
\fi
we further reduce these into a system of:
\begin{enumerate} 
\item \emph{Horn constraints} of the form $\psi_1 \wedge \ldots \wedge \psi_n \implies \psi_0$, and
\item \emph{resource constraints} of the form $\psi_1 \wedge \ldots \wedge \psi_n \implies \phi \geq 0$.
\end{enumerate}
Here any $\psi_i$ can be either a Boolean unknown $U^{\bbB}_{\Gamma}$ or a known refinement term,
and $\phi$ is a sum of zero or more numeric unknowns $U^{\bbN}_{\Gamma}$ and a known (linear) refinement term.
While prior work has shown how to efficiently solve Horn constraints using
predicate abstraction~\cite{RondonKaJh08,PolikarpovaKS16},
resource constraints present a new challenge,
since they contain unknown terms of both Boolean and numeric sorts.
In the interest of efficiency,
our synthesis algorithm does not attempt to solve for both Boolean and numeric terms at the same time.
Instead, it uses existing techniques to find a solution for the Horn constraints,
and then plugs this solution into the resource constraints.
Note that this approach does not sacrifice completeness,
as long as the Horn solver returns the least-fixpoint (\ie strongest) solution for each $U^{\bbB}_{\Gamma}$,
since Boolean unknowns only appear negatively in resource constraints%
\footnote{Our implementation uses \synquid's default greatest-fixpoint Horn solver,
which technically renders this technique incomplete,
however we observed that it works well in practice.}.

\paragraph{Resource Constraints}
The main new challenge then is to solve a system of resource constraints
of the form $\psi \implies \phi \geq 0$, 
where $\psi$ is now a known formula of the refinement logic.
Since potential annotations in \typesys are restricted to linear terms over program variables,
we can replace each unknown term $U^{\bbN}_{\Gamma}$ in $\phi$
with a linear template $\sum_{x \in X} C_i\cdot x$,
where each $C_i$ is an unknown integer coefficient
and $X$ is the set of all variables in $\Gamma$ such that $\jsort{\Gamma}{x}{\bbN}$. 
After normalization, the system of resource constraints is reduced 
to the following doubly-quantified system of linear inequalities:
$$
\many{\exists C_i} . \many{\forall x} . \bigwedge_{r\in R} r(\many{C_i}, \many{x})
$$
where each clause $r$ is of the form $\psi(\many{x}) \implies \sum f(\many{C_i})\cdot x \geq 0$,
$\psi$ is a known formula over the program variables $\many{x}$, 
and each $f$ is a linear function over unknown integer coefficients $\many{C_i}$.

Note a crucial difference between these constraints and those generated by RaML:
since RaML's potential annotations are not dependent---%
\ie $r$ cannot mention program variables $\many{x}$---%
its resource constraints reduce to plain linear inequalities: 
$
\many{\exists C_i} . \bigwedge \sum C_i \geq c
$
(where $c$ is a known constant),
which can be handled by an LP solver.
In our case, the challenge stems both from the double quantification
and the fact that individual clauses $r$ are \emph{bounded} by formulas $\psi$,
which are often nontrivial.
For example, synthesizing the function \T{range} from \autoref{sec:background}
gives rise to the following (simplified) resource constraints:
\begin{align*}
&\exists C_0 \ldots C_3 . \forall a, b, \nu. \\
&\quad(\neg(a \ge b) \wedge \nu = b) \implies (C_0 + 1){\cdot}a + C_1{\cdot}b + (C_2 - 1){\cdot}\nu + C_3 \ge 0\\
&\quad(\neg(a \ge b) \wedge \nu = b) \implies C_0{\cdot}a + C_1{\cdot}b + C_2{\cdot}\nu + C_3 \ge 0
\end{align*}
where a solution only exists if the bounds are taken into account.
One solution is $[C_0 \mapsto -1, C_1 \mapsto 0, C_2 \mapsto 1, C_3 \mapsto 0]$, 
which stands for the potential term $\nu - a$.

\setlength{\textfloatsep}{10pt}
\begin{algorithm}[t]
  \begin{algorithmic}[0] 
    \Require{Constraints $R$, current solution $\cegissol$, examples $\cegisex$}
    \Ensure{New solution and examples $(\cegissol, \cegisex)$ or $\bot$ if no solution}    
    \Procedure{Solve}{$R,\ \cegissol,\ \cegisex$}
      \State $e \gets \mathsf{SMT}(\exists \many{x} . \lnot R(\cegissol, \many{x}))$
      \If {$e = \bot$}  \Comment{No counter-example}
        \State \textbf{return} $(\cegissol, \cegisex)$
      \Else
        \State $\cegisex' \gets \ \cegisex \ \cup \ e$
        \State $R' \gets \{ r \in R \mid \lnot r(\cegissol, e) \}$ 
        \State $\cegissol' \gets \mathsf{SMT}(\exists \many{C_i} . \bigwedge_{e \in \cegisex'} \ R'(\many{C_i}, e))$
        \If{$\cegissol' = \bot$} \textbf{return} $\bot$  \Comment{No solution}
        \Else\ \Call{Solve}{$R,\ \cegissol \ \cup \ \cegissol',\ \cegisex'$}
        \EndIf
      \EndIf
    \EndProcedure
  \end{algorithmic}
  \caption{Incremental solver for resource constraints}
  \label{alg:incremental}
\end{algorithm}

\paragraph{Incremental Solving}
Constraints of this form can be solved using \emph{counter-example guided synthesis} (CEGIS)~\cite{Solar-LezamaTBSS06},
which is, however, relatively expensive.
We observe that in the context of synthesis we have to repeatedly solve similar systems of resource constraints
because a program candidate is type-checked \emph{incrementally} as it is being constructed,
which corresponds to an incrementally growing set of clauses $R$.
Moreover, we observe that as new clauses are added,
only a few existing coefficients $C_i$ are typically invalidated,
so we can avoid solving for all the coefficients from scratch.
To this end, we develop an incremental version of the CEGIS algorithm, 
shown in Algorithm~\ref{alg:incremental}. 

The goal of the algorithm is to find a solution
$\cegissol : C_i \to \bbZ$ that maps unknown coefficients to integers
such that $\many{\forall x} . R(\cegissol, \many{x})$ holds
(we write $R(\cegissol, \many{x})$ as a shorthand for $\bigwedge_{r\in R} r(\cegissol, \many{x})$).
The algorithm takes as input a set of clauses $R$
(which includes both old and new clauses),
the current solution $\cegissol$
(new coefficients $C_i$ are mapped to 0)
and the current set of \emph{examples} $\cegisex$, 
where an example $e \in \cegisex$ is a partial assignment to universally-quantified variables $e: X \to \bbN$.
%


The algorithm first queries the SMT solver for a counter-example $e$ to the current solution.
If no such counter-example exists, the solution is still valid
(this happens surprisingly often, since many resource constraints are trivial).
Otherwise, the current solution needs to be updated.
To this end, a traditional CEGIS algorithm would query the SMT solver with the following \emph{synthesis constraint}:
$\exists \many{C_i} . \bigwedge_{e \in \cegisex'} \ R(\many{C_i}, e)$,
which enforces that all clauses are satisfied on the extended set of examples.
Instead, our incremental algorithm picks out only those clauses $R'$
that are actually violated by the new counter-example;
since in our setting $R'$ is typically small,
this optimization significantly reduces the size of the synthesis constraint
and synthesis times for programs with dependent annotations
(as we demonstrate in \autoref{sec:eval}).

\subsection{Implementation}\label{sec:synthesis:impl}

We implemented the resource-guided synthesis algorithm
in \tool,
which extends \synquid with support for resource-annotated types
and a resource constraint solver.
Note that while our formalization is restricted to Booleans
and length-indexed lists,
our implementation supports the full expressiveness of \synquid's types:
types include integers and user-defined algebraic datatypes,
and refinement formulas support sets and can mention arbitrary user-defined measures.
More importantly, resource terms in \tool can mention integer variables
and use subtraction, multiplication, conditional expressions,
and numeric measures;
finally, multiplicities on type variables can be dependent (mention variables).
These changes have the following implications:
\begin{inparaenum}[(1)]
\item resource terms are not syntactically guaranteed to be non-negative, 
so we emit additional well-formedness constraints to enforce this;
\item resource terms are not syntactically restricted to be linear;
our implementation is incomplete, and simply rejects the program if a nonlinear term arises;
\item subtyping and sharing constraints with conditional resource terms
are decomposed into unconditional ones by moving the guard to the context,
so the search space for all numeric unknowns remains unconditional;
\item to handle measure applications in resource constraints,
we replace them with fresh integer variables,
and avoid spurious counter-examples by explicitly instantiating the congruence axiom
with all applications in the constraint.
\end{inparaenum}

\section{Evaluation}\label{sec:eval}


We evaluated \tool using the following criteria:
%
%
\begin{description}[leftmargin=1em]
\item\textbf{Relative performance:} How do \tool's synthesis times compare to \synquid's?
How much does the additional burden of solving resource constraints 
affect its performance?
\item\textbf{Efficacy of resource analysis:} Can \tool discover
more efficient programs than \synquid?
\item\textbf{Value of round-trip type checking:}
Does round-trip type checking afforded by the tight integration of resource analysis into \synquid
effective at pruning the search space?
How does it compare to the naive combination of synthesis and resource analysis?
%
\item\textbf{Value of incremental solving:} To what extent does incremental solving 
of resource constraints improve \tool's performance?
\end{description}

\subsection{Relative Performance}\label{sec:eval:bench}
To evaluate \tool's performance relative to \synquid,
we selected \numBench problems from \synquid's original suite,
annotated them with resource bounds,
and re-synthesized them with \tool.
The rest of the original 64 benchmarks require non-linear bounds,
and thus are out of scope of \typesys. 
The details of this experiment are shown in \autoref{fig:eval:res}, which 
compares \tool's synthesis times against \synquid's on these 
linear-bounded benchmarks.

Unsurprisingly, due to the additional constraint-solving, 
\tool generally performs worse than \synquid:
the median synthesis time is about \slowdown higher.
Note, however, that in return it provides provable guarantees about the performance
of generated code.
\tool was able to discover a more efficient implementation for only \emph{one} of the original 
\synquid benchmarks (\T{compress}, discussed below). 
In general, these benchmarks contain only the minimal set of components 
required to produce a valid implementation,
which makes it hard for \synquid to find a non-optimal version. 
%
\emph{Four} of the benchmarks in \autoref{fig:eval:res} use advanced features of \typesys:
for example, any function using natural numbers to index or construct 
a data structure requires dependent potential annotations.

\begin{table}[t]
\sisetup{round-mode=places}
\footnotesize
  \begin{minipage}{\columnwidth}
    \resizebox{\columnwidth}{!}{
      \begin{tabular}{@{} r|c| c | cS[round-precision=1]S[round-precision=1] @{}}
        \head{Group} & \head{Description}  & \head{Components} & \head{Code} & \head{Time} & \head{TimeNR}  \\	

        \hhline{======}
        \multirow{22}{*}{\parbox{1cm}{\vspace{-0.85\baselineskip}\center{List}}} & is empty & true, false & 16 & 0.20 & 0.16 \\
        & member & true, false, $=$, $\neq$ & 41 & 0.24 & 0.21 \\
        & duplicate each element &  & 39 & 0.52 & 0.27 \\
        & replicate & 0, inc, dec, $\leq$, $\neq$ & 31 & 2.88 & 0.23 \\
        & append two lists &  & 38 & 1.54 & 0.48 \\
        & take first $n$ elements & 0, inc, dec, $\leq$, $\neq$ & 34 & 2.42 & 0.17 \\
        & drop first $n$ elements & 0, inc, dec, $\leq$, $\neq$ & 30 & 20.37 & 0.32 \\
        & concat list of lists & append & 49 & 3.34 & 0.81 \\
        & delete value & $=$, $\neq$ & 49 & 0.77 & 0.31 \\
        & zip &  & 32 & 0.43 & 0.22 \\
        & zip with &  & 35 & 0.45 & 0.24 \\
        & $i$-th element & 0, inc, dec, $\leq$, $\neq$ & 30 & 0.30 & 0.23 \\
        & index of element & 0, inc, dec, $=$, $\neq$ & 43 & 0.54 & 0.29 \\
        & insert at end &  & 42 & 0.43 & 0.32 \\
        & balanced split & fst, snd, abs & 64 & 9.62 & 1.73 \\
        & reverse & insert at end & 35 & 0.44 & 0.32 \\
        & insert (sorted) & $\leq$, $\neq$ & 57 & 1.98 & 0.69 \\
        & extract minimum & $\leq$, $\neq$ & 71 & 18.14 & 8.31 \\
        & foldr &  & 43 & 1.78 & 0.60 \\
        & length using fold & 0, inc, dec & 39 & 0.27 & 0.21 \\
        & append using fold &  & 42 & 0.32 & 0.26 \\
        & map &  & 27 & 0.28 & 0.19 \\
        \hline\multirow{5}{*}{\parbox{1cm}{\vspace{-0.85\baselineskip}\center{Unique list}}} & insert & $=$, $\neq$ & 49 & 0.83 & 0.44 \\
        & delete & $=$, $\neq$ & 45 & 0.53 & 0.34 \\
        & compress & $=$, $\neq$ & 64 & 4.98 & 1.85 \\
        & integer range & 0, inc, dec, $\leq$, $\neq$ & 46 & 88.35 & 5.14 \\
        & partition & $\leq$ & 71 & 13.04 & 5.46 \\
        \hline\multirow{3}{*}{\parbox{1cm}{\vspace{-0.85\baselineskip}\center{Sorted list}}} & insert & $<$ & 64 & 1.62 & 0.57 \\
        & delete & $<$ & 52 & 0.47 & 0.30 \\
        & intersect & $<$ & 71 & 17.01 & 0.76 \\
        \hline\multirow{4}{*}{\parbox{1cm}{\vspace{-0.85\baselineskip}\center{Tree}}} & node count & 0, 1, + & 34 & 3.84 & 0.52 \\
        & preorder & append & 45 & 3.03 & 0.55 \\
        & to list & append & 45 & 2.99 & 0.54 \\
        & member & false, not, or, $=$ & 63 & 2.17 & 0.60 \\
        \hline\multirow{4}{*}{\parbox{1cm}{\vspace{-0.85\baselineskip}\center{BST}}} & member & true, false, $\leq$, $\neq$ & 72 & 0.50 & 0.32 \\
        & insert & $\leq$, $\neq$ & 90 & 4.49 & 1.55 \\
        & delete & $\leq$, $\neq$ & 103 & 26.77 & 9.25 \\
        & BST sort & $\leq$, $\neq$ & 191 & 8.97 & 4.27 \\
        \hline\multirow{5}{*}{\parbox{1cm}{\vspace{-0.85\baselineskip}\center{Binary Heap}}} & insert & $\leq$, $\neq$ & 90 & 3.21 & 1.03 \\
        & member & false, not, or, $\leq$, $\neq$ & 78 & 2.34 & 0.84 \\
        & 1-element constructor & $\leq$, $\neq$ & 44 & 0.21 & 0.21 \\
        & 2-element constructor & $\leq$, $\neq$ & 91 & 0.67 & 0.34 \\
        & 3-element constructor & $\leq$, $\neq$ & 274 & 21.35 & 4.02 \\
        \hline
      \end{tabular}
    }
  \caption{Comparison of \tool and \synquid.
  For each benchmark, we report the set of provided \head{Components};
  cumulative size of synthesized \head{Code} (in AST nodes) for all goals;
  as well as running times (in seconds) for \tool (\head{Time})
  and \synquid (\head{TimeNR}).
  }\label{fig:eval:res}
  \end{minipage}
\end{table}

\subsection{Case Studies}\label{sec:eval:micro}

\begin{table*}[t]
\sisetup{round-mode=places}
\begin{center}
\footnotesize
\resizebox{\textwidth}{!}{
\begin{tabular}{cc | c | c | S[round-precision=1]S[round-precision=1]S[round-precision=1]S[round-precision=1] | cc }
& \head{Description} & \head{Type Signature} & \head{Components} & \head{T} & \head{T-NR} & \head{T-EAC} & \head{T-NInc} & \head{B} & \head{B-NR} \\ 

\hhline{==========}
1 & triple & $\forall \alpha .                    \tarrow{xs}{\tlist{\tpot{\alpha}{2}}}                      {\tsubset{\tlist{\alpha}}{\T{len} \ \nu = \T{len} \ xs + \T{len} \ xs + \T{len} \ xs }}$ & append & 0.88 & 0.36 & 0.42 & {-} & $\mid xs \mid$ & $\mid xs \mid$ \\
2 & triple' & $\forall \alpha .                    \tarrow{xs}{\tlist{\tpot{\alpha}{2}}}                      {\tsubset{\tlist{\alpha}}{\T{len} \ \nu = \T{len} \ xs + \T{len} \ xs + \T{len} \ xs }}$ & append' & 2.80 & 0.40 & 1.22 & {-} & $\mid xs \mid$ & $\mid xs \mid^2$ \\
3 & concat list of lists & $\forall\alpha .             \tarrow{xxs}{\tlist{\tlist{\tpot{\alpha}{1}}}}               {\tarrow{acc}{\tlist{\alpha}}                 {\tsubset{\tlist{\alpha}}{\T{sumLen} \ xs = \T{len} \nu}}}$ & append & 3.21 & 0.86 & 1.10 & {-} & $\mid xxs \mid$ & $\mid xxs \mid^2$ \\
4 & compress & $\forall \alpha .                    \tarrow{xs}{\tlist{\tpot{\alpha}{1}}}                      {\tsubset{\tclist{\alpha}}{\T{elems} \ xs = \T{elems} \ \nu}}$ & $=$,$\neq$ & 3.82 & 1.10 & 4.13 & {-} & $\mid xs \mid$ & $2^{ \mid xs \mid }$ \\
5 & common & $\forall\alpha .             \tarrow{ys}{\tilist{\tpot{\alpha}{1}}}               {\tarrow{zs}{\tilist{\tpot{\alpha}{1}}}                 {\tsubset{\tlist{\alpha}}{\T{elems} \ \nu = \T{elems} \ ys \cap \T{elems} \ zs}}}$ & $<$, member & 30.79 & 1.07 & TO & {-} & $\mid ys \mid + \mid zs \mid$ & $\mid ys \mid \mid zs \mid$ \\
6 & list difference & $\forall\alpha .             \tarrow{ys}{\tilist{\tpot{\alpha}{1}}}               {\tarrow{zs}{\tilist{\tpot{\alpha}{1}}}                 {\tsubset{\tlist{\alpha}}{\T{elems} \ \nu = \T{elems} \ ys - \T{elems} \ zs}}}$ & $<$, member & 173.54 & 1.33 & TO & {-} & $\mid ys \mid + \mid zs \mid$ & $\mid ys \mid \mid zs \mid$ \\
7 & insert & $\forall\alpha .                 \tarrow{x}{\alpha}                 {\tarrow{xs}{\tilist{\tpot{\alpha}{1}}}                   {\tsubset{\tilist{\alpha}}{\T{elems} \ \nu = [x] \cup \T{elems} \ xs}}}$ & $<$ & 1.30 & 0.43 & {-} & {-} & $\mid xs \mid$ & $\mid xs \mid$ \\
8 & insert' & $\forall\alpha .                    \tarrow{x}{\alpha}                     {\tarrow{xs}{ \tpot{ \tilist{ \alpha }}{\mathsf{numgt}(x,\nu)} }                       {\tsubset{\tilist{\alpha}}{\T{elems} \ \nu = [x] \cup \T{elems} \ xs}}}$ & $<$ & 49.59 & 0.67 & {-} & 102.23 & $\T{numgt}(x,xs)$ & $\mid xs \mid$ \\
9 & insert'' & $\forall\alpha .                    \tarrow{x}{\alpha}                     {\tarrow{xs}{\tilist{\tpot{\alpha}{\mathsf{ite}(x > \nu, 1, 0)}}}                       {\tsubset{\tilist{\alpha}}{\T{elems} \ \nu = [x] \cup \T{elems} \ xs}}}$ & $<$ & 7.69 & 0.38 & {-} & 13.74 & $\T{numgt}(x,xs)$ & $\mid xs \mid$ \\
10 & replicate & $\forall\alpha .             \tarrow{n}{\T{Nat}}               {\tarrow{x}{n \times \tpot{\alpha}{n}}}                 {\tsubset{\tlist{\alpha}}{\T{len} \ \nu = n}}$ & zero, inc, dec & 1.43 & 0.15 & {-} & 2.67 & $n$ & $n$ \\
11 & take & $\forall\alpha .                 \tarrow{n}{\T{Nat}}                 {\tarrow{xs}{\tpot{\tsubset{\tlist{\alpha}}{\T{len} \nu \geq n}}{n}}                     {\tsubset{\tlist{\alpha}}{\T{len} \nu = n}}}$ & zero, inc, dec & 1.22 & 0.11 & {-} & 2.42 & $n$ & $n$ \\
12 & drop & $\forall\alpha .                 \tarrow{n}{\T{Nat}}                 {\tarrow{xs}{\tpot{\tsubset{\tlist{\alpha}}{\T{len} \nu \geq n}}{n}}                     {\tsubset{\tlist{\alpha}}{\T{len} \nu = \T{len} xs - n}}}$ & zero, inc, dec & 12.89 & 0.20 & {-} & 17.08 & $n$ & $n$ \\
13 & range & $\tarrow{lo}{\T{Int}}                 {\tarrow{hi}{\tsubset{\tpot{\T{Int}}{\nu - lo}}{\nu \geq lo}}                   {\tsubset{\tilist{\tsubset{\T{Int}}{lo \leq \nu \leq hi}}}{\T{len} \nu = hi - lo}}}                   {}  $ & inc,dec,$\geq$ & 11.78 & 0.16 & {-} & {-} & $hi - lo$ & - \\
14 & CT insert & $\forall\alpha .                 \tarrow{x}{\alpha}                 {\tarrow{xs}{\tilist{\tpot{\alpha}{1}}}                   {\tsubset{\tilist{\alpha}}{\T{elems} \ \nu = [x] \cup \T{elems} \ xs}}}$ & $<$ & 2.24 & 0.63 & 0.79 & {-} & $\mid xs \mid$ & $\mid xs \mid$ \\
15 & CT compare & $\forall\alpha .                       \tarrow{ys}{\tlist{\tpot{\alpha}{1}}}                         {\tarrow{zs}{\tlist{\alpha}}{\tsubset{\tbool}{\nu = ( \T{len} \ ys = \T{len} \ zs )}}} $ & true, false, and & 14.25 & 0.54 & 9.14 & {-} & $\mid ys \mid$ & $\mid ys \mid$ \\
16 & compare & $\forall\alpha .                       \tarrow{ys}{\tlist{\tpot{\alpha}{1}}}                         {\tarrow{zs}{\tlist{\alpha}}{\tsubset{\tbool}{\nu = ( \T{len} \ ys = \T{len} \ zs )}}} $ & true, false, and & 1.01 & 0.34 & {-} & {-} & $\mid ys \mid$ & $\mid ys \mid$ \\
\end{tabular}
}
\end{center}
\caption{Case Studies.
For each synthesis problem, we report: 
the run time of \tool (\head{T}),
\synquid (\head{T-NR}),
naive combination of \synquid and resource analysis (\head{T-EAC}),
\tool without incremental solving (\head{T-NInc});
as well as the tightest resource bound for the code generated by \tool (\head{B})
and by \synquid (\head{B-NR}).
Here, $SL$ is the type of 
sorted lists, and $CL$ refers to the type of lists without adjacent duplicates. 
TO is 10 min; all benchmarks count recursive calls.
}
\label{fig:eval:micro}
\end{table*}

The value of resource-guided synthesis becomes clear when the library 
of components grows.
To confirm this intuition, 
we assembled a suite of \numMB case studies shown 
in \autoref{fig:eval:micro}, each exemplifying some feature of \tool.

\paragraph{Optimization}
The first six benchmarks showcase \tool's ability to generate faster code 
than \synquid
(the cost metric in each case is the number of recursive calls).
Benchmark 1 is \T{triple} from \autoref{sec:background:re2},
where both \synquid and \tool generate the same efficient solution;
benchmark 2 is slight modification of this example:
it uses a component \T{append'}, which traverses its second argument
(unlike \T{append}, which traverses its first).
In this case, \tool generates the efficient solution,
associating the two calls to \T{append'} \emph{to the left},
while \synquid still generates the same---now inefficient---solution,
associating these calls \emph{to the right}.
In benchmark 3 \tool makes the optimal choice of accumulator
to avoid a quadratic-time implementation.
Benchmark 4 is \T{compress} from \autoref{fig:eval:res}:
the task is to remove adjacent duplicated from a list.
Here \synquid makes an unnecessary recursive call,
resulting in a solution that is slightly shorter but runs in exponential time! 

In other cases, \tool drastically changes the structure of the program 
to find an optimal implementation. 
Benchmark 5 is \T{common} from \autoref{sec:background:synquid}, 
where \tool must find an implementation that does not call \T{member}. 
Benchmark 6 works similarly,
but computes the difference between two lists instead of their intersection.
On these benchmarks, the performance disparity between \tool and \synquid is much worse, 
as \tool must reject many more programs before it finds 
an appropriate implementation.
On the other hand, these benchmarks also showcase the value of \emph{round-trip type checking}:
the column \head{T-EAC} reports synthesis times for a naive combination of synthesis and resource analysis,
where we simply ask \synquid to enumerate functionally correct programs
until one type-checks under \typesys.
As you can see, for benchmarks 5 and 6 this naive version times out after ten minutes.  




\paragraph{Dependent Potentials}
Benchmarks 7--13 showcase fine-grained bounds that leverage dependent potential annotations.
The first three of those synthesize a function \T{insert} that inserts an element into a sorted list.
In benchmark 7 we use a simple linear bound (the length of the list),
while benchmarks 8 and 9 specify a tighter bound: 
\T{insert x xs} can only make one recursive call per element of \T{xs} larger than \T{x}.
These two examples showcase two different styles of specifying precise bounds:
in 8 we define a custom measure \T{numgt} that counts list elements greater than a certain value;
in 9, we instead annotate each list element with a conditional term
indicating that it carries potential only if its value is larger than \T{x}.
%
As discussed in \autoref{sec:background}, 
benchmark 13 (\T{range}) cannot be synthesized by \synquid at all,
because of restrictions on its termination checking mechanism,
while \tool handles this benchmark out of the box.

For benchmarks 8--13, which make use of dependent potential annotations, 
we also report the synthesis times without incremental solving of resource constraints (\head{T-NInc}),
which are up to \nincslowdown higher.
%


\paragraph{Constant Resource}
As discussed in \autoref{sec:typesys}, a simple extension to \typesys
enables it to verify constant-resource implementations. 
We showcase this feature in benchmarks 14--16.
Benchmark 15 is an example from~\cite{SP:NDF17},
which compares a public list $ys$ with a secret list $zs$. 
By allotting potential only to $ys$,
we guarantee that the resource consumption of the generated program is independent of the length of $zs$. 
If this requirement is relaxed (as in benchmark 16),
the generated program indeed terminates early,
potentially revealing the length of $zs$ to an adversary
(in case $zs$ is the shorter of the two lists). 
Benchmark 14 is a constant-time version of benchmark 7 (\T{insert}), 
which is forced to make extra recursive calls so as not to reveal the length of the list.

\section{Related Work}\label{sec:related}

\paragraph{Resource Analysis}

Automatic static resource analysis has been extensively studied and is
an active area of research.  Many advanced techniques for imperative
integer programs apply abstract interpretation to generate numerical
invariants. The obtained \emph{size-change information} forms the
basis for the computation of actual bounds on loop iterations and
recursion depths; using counter instrumentation~\cite{GulwaniMC09},
ranking
functions~\cite{AliasDFG10,AlbertAGP11a,BrockschmidtEFFG14,SinnZV14},
recurrence relations~\cite{Albert12,AlbertAGGP12}, and abstract
interpretation itself~\cite{Zuleger11,CernyHKRZ15}.  Automatic
resource analysis techniques for functional programs are based on
sized types~\cite{Vasconcelos08}, recurrence
relations~\cite{DannerLR15}, term-rewriting~\cite{AvanziniLM15}, and
amortized resource
analysis~\cite{Jost03,Jost10,HoffmannAH10,SimoesVFJH12}.  There exist
several tools that can automatically derive loop and recursion bounds
for imperative programs including SPEED~\cite{GulwaniMC09,GulwaniZ10},
KoAT~\cite{BrockschmidtEFFG14}, PUBS~\cite{AlbertAGGP12},
Rank~\cite{AliasDFG10}, ABC~\cite{BlancHHK10} and
LOOPUS~\cite{Zuleger11,SinnZV14}. These techniques are passive in the
sense that they provide feedback about a program without actively
synthesizing or repairing programs.

\paragraph{Domain-Specific Program Synthesis}
Most program synthesis techniques~\cite{OseraZd15,FeserChDi15,Smith-Albarghouthi:PLDI16,FengMGDC17,FengM0DR17,FengMBD18,WangDS18,WangCB17,SrivastavaGF10,KneussKuKuSu13,PolikarpovaKS16,InalaPQLS17,QiuS17}
do not explicitly take resource usage into account during synthesis.
Many of them, however, leverage \emph{domain knowledge} to restrict the search space to only include efficient programs~\cite{GulwaniJTV11,Cheung13}
or to encode domain-specific performance considerations as part of the functional specification~\cite{InalaS16,Loncaric2016,Loncaric2018}.

\paragraph{Synthesis with Quantitative Objectives}
Two lines of prior work on synthesis are explicitly concerned with optimizing 
resource usage. 
One is quantitative \emph{automata-theoretic synthesis},
which has been used to synthesize optimal Mealy machines~\cite{BloemCHJ09}
and place synchronization in concurrent programs~\cite{CernyCHRS11,GuptaHRST15,CernyCHRRST15}.
In contrast, we focus on synthesis of high-level programs
that can manipulate custom data structures,
which are out of reach for automata-theoretic synthesis. 

The second relevant line of work is \emph{synthesis-aided compilation}~\cite{Schkufza0A13,Phothilimthana14,Sharma15,PhothilimthanaT16}.
This work is limited to generating low-level straight-line code,
which is an easy target for correctness validation and cost estimation.
Perhaps the closest work to ours is the Synapse tool~\cite{Bornholt16},
which supports a richer space of programs,
but requires extensive guidance from the user (in the form of meta-sketches),
and relies on bounded reasoning, which can only provide correctness and optimality guarantees for a finite set of inputs.
In contrast,
we use type-based verification and resource analysis techniques,
which enable \tool to handle high-level recursive programs
and provide guarantees for an unbounded set of inputs.

\begin{acks}                            
  This article is based on research supported by the United States Air
  Force under DARPA AA Contract FA8750-18-C-0092 and DARPA STAC
  Contract FA8750-15-C-0082, and by the National Science Foundation
  under SaTC Award 1801369, SHF Award 1812876, and CAREER Award
  1845514.
  Any opinions, findings, and conclusions contained in this document
  are those of the authors and do not necessarily reflect the views of
  the sponsoring organizations.
\end{acks}

\bibliography{references,resource,db}

\iflong

\appendix
%

\mathtoolsset{showonlyrefs,showmanualtags}
\allowdisplaybreaks

\newcommand{\myinferrule}[3][]{\Rule{#1}{#2}{#3}}
\newcommand{\propref}[1]{Prop.~\ref{prop:#1}}
\newcommand{\lemref}[1]{Lem.~\ref{lem:#1}}
\newcommand{\theoref}[1]{Thm.~\ref{the:#1}}

\begin{figure*}[t!]
	\begin{mathpar}\footnotesize
		\Rule{E-Cond-True}{ }{ \jstep{\econd{\etrue}{e_1}{e_2}}{e_1}{q}{q} }
		\and
		\Rule{E-Cond-False}{ }{ \jstep{\econd{\efalse}{e_1}{e_2}}{e_2}{q}{q} }
		\and
		\Rule{E-Let1}{ \jstep{e_1}{e_1'}{q}{q'} }{ \jstep{\elet{e_1}{x}{e_2}}{\elet{e_1'}{x}{e_2}}{q}{q'} }
		\and
		\Rule{E-Let2}{ \jval{v_1} }{ \jstep{\elet{v_1}{x}{e_2}}{\subst{v_1}{x}{e_2}}{q}{q} }
		\and
		\Rule{E-MatL-Nil}{ }{ \jstep{\ematl{\enil}{e_1}{x_h}{x_t}{e_2}}{e_1}{q}{q} }
		\and
		\Rule{E-MatL-Cons}{ \jval{v_h} \\ \jval{v_t} }{ \jstep{\ematl{\econs{v_h}{v_t}}{e_1}{x_h}{x_t}{e_2}}{\subst{v_h,v_t}{x_h,x_t}{e_2}}{q}{q'} }
		\and
		\Rule{E-App-Abs}{ \jval{v_2} }{ \jstep{\eapp{\eabs{x}{e_0}}{v_2}}{\subst{v_2}{x}{e_0}}{q}{q} }
		\and
		\Rule{E-App-Fix}{ \jval{v_2} }{ \jstep{\eapp{\efix{f}{x}{e_0}}{v_2}}{\subst{\efix{f}{x}{e_0},v_2}{f,x}{e_0}}{q}{q} }
		\and
		\Rule{E-Consume}{ }{ \jstep{\econsume{c}{e_0}}{e_0}{q}{q-c} }
	\end{mathpar}
	\caption{Evaluation rules of the small-step operational cost semantics.}
	\label{fig:semantics}
\end{figure*}

\section{The \typesys Type System}
\label{sec:appendixre2}

\subsection{Scalar Types: $\tscalar{S}$}

In \typesys, we define \emph{scalar types} to be annotated subset types.
Neither arrow types nor type schemas are scalar.

\begin{mathpar}\footnotesize
\inferrule{ }{ \tscalar{\trefined{B}{\psi}{\phi}} }
\end{mathpar}

\subsection{Sorting: $\jsort{\Gamma}{\psi}{\Delta}$}

Refinements are classified by sorts.
The \emph{sorting} judgment $\jsort{\Gamma}{\psi}{\Delta}$ states that a refinement $\psi$ has a sort $\Delta$ under a context $\Gamma$.
The typing context is needed because refinements can reference program variables.
To reflect types of program variables in the refinement level, we define a relation $S \rightsquigarrow \Delta$ as follows.
The relation $\rightsquigarrow$ defines a partial function from types to sorts.

\begin{mathpar}\footnotesize
\inferrule{ }{ \trefined{\tbool}{\psi}{\phi} \rightsquigarrow \bbB }
\and
\inferrule{ }{ \trefined{\tlist{T}}{\psi}{\phi} \rightsquigarrow \bbN }
\and
\inferrule{ }{ \trefined{m \cdot \alpha}{\psi}{\phi} \rightsquigarrow \delta_\alpha }
\end{mathpar}

\autoref{fig:sorting} presents the sorting rules.

\begin{figure}[h]
\begin{mathpar}\footnotesize
	\myinferrule[S-Var]
	{ \jwfctxt{\Gamma} \\ \Gamma(x) \rightsquigarrow \Delta }
	{ \jsort{\Gamma}{x}{\Delta} }
	\and
	\myinferrule[S-Top]
	{ \jwfctxt{\Gamma} }
	{ \jsort{\Gamma}{\top}{\bbB} }
	\and
	\myinferrule[S-Neg]
	{ \jsort{\Gamma}{\psi}{\bbB} }
	{ \jsort{\Gamma}{\neg\psi}{\bbB} }
	\and
	\myinferrule[S-And]
	{ \jsort{\Gamma}{\psi_1}{\bbB} \\ \jsort{\Gamma}{\psi_2}{\bbB} }
	{ \jsort{\Gamma}{\psi_1 \wedge \psi_2}{\bbB} }
	\and
	\myinferrule[S-Nat]
	{ \jwfctxt{\Gamma} }
	{ \jsort{\Gamma}{n}{\bbN} }
	\and
	\myinferrule[S-Rel]
	{ \jsort{\Gamma}{\psi_1}{\bbN} \\ \jsort{\Gamma}{\psi_2}{\bbN} }
	{ \jsort{\Gamma}{\psi_1 \leq \psi_2}{\bbB} }
	\and
	\myinferrule[S-Op]
	{ \jsort{\Gamma}{\psi_1}{\bbN} \\ \jsort{\Gamma}{\psi_2}{\bbN} }
	{ \jsort{\Gamma}{\psi_1 + \psi_2}{\bbN} }
	\and
	\myinferrule[S-Eq]
	{ \jsort{\Gamma}{\psi_1}{\Delta} \\ \jsort{\Gamma}{\psi_2}{\Delta} }
	{ \jsort{\Gamma}{\psi_1 = \psi_2}{\bbB} }
\end{mathpar}
\caption{Sorting rules}
\label{fig:sorting}
\end{figure}

\subsection{Type Wellformedness: $\jwftype{\Gamma}{S}$}

A type $S$ is said to be \emph{wellformed} under a context $\Gamma$ if the following three properties hold:
\begin{itemize}
	\item every referenced program variables in $S$ is in the correct scope, and
	\item polymorphic types can never carry positive potential.
\end{itemize}

\autoref{fig:wftype} presents the type wellformedness rules.

\begin{figure}[h]
\begin{mathpar}\footnotesize
	\myinferrule[Wf-Bool]
	{ \jwfctxt{\Gamma} }
	{ \jwftype{\Gamma}{\tbool} }
	\and
	\myinferrule[Wf-List]
	{ \Omit{\tscalar{T} \\} \jwftype{\Gamma}{T} }
	{ \jwftype{\Gamma}{\tlist{T}} }
	\and
	\myinferrule[Wf-TVar]
	{ \jwfctxt{\Gamma} \\ \alpha \in \Gamma  }
	{ \jwftype{\Gamma}{m \cdot \alpha} }
	\and
	\myinferrule[Wf-Refined]
	{ \jwftype{\Gamma}{B} \\ \jsort{\Gamma,\nu:B}{\psi}{\bbB} }
	{ \jwftype{\Gamma}{\tsubset{B}{\psi}} }
	\and
	\myinferrule[Wf-Arrow]
	{ \jwftype{\Gamma}{T_x}  \\ \jwftype{\Gamma,x:T_x}{T} }
	{ \jwftype{\Gamma}{\tarrowm{x}{T_x }{T}{m}} }
	\and
	\myinferrule[Wf-Pot]
	{ \jwftype{\Gamma}{R} \\ \jsort{\Gamma,\nu:R}{\phi}{\bbN} }
	{  \jwftype{\Gamma}{\tpot{R}{\phi}} }
	\and
	\myinferrule[Wf-Poly]
	{ \jsharing{\Gamma,\alpha}{S}{S}{S} }
	{ \jwftype{\Gamma}{\forall\alpha.S} }
\end{mathpar}
\caption{Type wellformedness rules}
\label{fig:wftype}
\end{figure}

Recall that when we defined sorting rules we proposed a relation $S \rightsquigarrow \Delta$ that is a partial function from types to sorts.
With wellformed types, we can interpret $\rightsquigarrow$ as a better-behaved map.

\begin{proposition}
	The relation $S \rightsquigarrow \Delta$ defines a total map from wellformed scalar types into sorts, i.e., if $\jwftype{\Gamma}{S}$ and $\tscalar{S}$, then there exists a unique $\Delta$ such that $S \rightsquigarrow \Delta$.
\end{proposition}
\begin{proof}
	By induction on $\jwftype{\Gamma}{S}$.
\end{proof}

\subsection{Context Wellformedness: $\jwfctxt{\Gamma}$}

A context $\Gamma$ is said to be \emph{wellformed} if every binding in $\Gamma$ is wellformed under a ``prefix'' context before it.
Recall that the context is a sequence of variable bindings, type variables, path conditions, and free potentials.
\autoref{fig:wfctxt} shows these rules.

\begin{figure}[h]
\begin{mathpar}\footnotesize
	\myinferrule[Wf-Empty]
	{ }
	{ \jwfctxt{\cdot} }
	\and
	\myinferrule[Wf-Bind-Type]
	{ \jwfctxt{\Gamma} \\ \jwftype{\Gamma}{S} }
	{ \jwfctxt{\Gamma,\bindvar{x}{ S}} }
	\and
	\myinferrule[Wf-Bind-Cond]
	{ \jwfctxt{\Gamma} \\ \jsort{\Gamma}{\psi}{\bbB} }
	{ \jwfctxt{\Gamma,\psi} }
	\and
	\myinferrule[Wf-Bind-TVar]
	{ \jwfctxt{\Gamma} }
	{ \jwfctxt{\Gamma,\alpha} }
	\and
	\myinferrule[Wf-Bind-Pot]
	{ \jwfctxt{\Gamma} \\ \jsort{\Gamma}{\phi}{\bbN} }
	{ \jwfctxt{\Gamma,\phi} }
\end{mathpar}
\caption{Context wellformedness rules}
\label{fig:wfctxt}
\end{figure}

\subsection{Context Sharing: $\jctxsharing{\Gamma}{\Gamma_1}{\Gamma_2}$}

We have already presented type sharing rules.
To apportion the associated potential of $\Gamma$ properly to two contexts $\Gamma_1,\Gamma_2$ with the same sequence of bindings, we introduce \emph{context sharing} relations.
The rules are summarized in \autoref{fig:ctxsharing}.

\begin{figure}[h]
\begin{mathpar}\footnotesize
	\myinferrule[Share-Empty]
	{ }
	{ \jctxsharing{\cdot}{\cdot}{\cdot} }
	\and
	\myinferrule[Share-Bind-Type]
	{ \jctxsharing{\Gamma}{\Gamma_1}{\Gamma_2} \\ \jsharing{\Gamma}{S}{S_1}{S_2} }
	{ \jctxsharing{\Gamma,\bindvar{x}{ S}}{\Gamma_1,\bindvar{x}{ S_1}}{\Gamma_2,\bindvar{x}{ S_2}} }
	\and
	\myinferrule[Share-Bind-Cond]
	{ \jctxsharing{\Gamma}{\Gamma_1}{\Gamma_2} \\ \jsort{\Gamma}{\psi}{\bbB} }
	{ \jctxsharing{\Gamma,\psi}{\Gamma_1,\psi}{\Gamma_2,\psi} }
	\and
	\myinferrule[Share-Bind-TVar]
	{ \jctxsharing{\Gamma}{\Gamma_1}{\Gamma_2} }
	{ \jctxsharing{\Gamma,\alpha}{\Gamma_1,\alpha}{\Gamma_2,\alpha} }
	\and
	\myinferrule[Share-Bind-Pot]
	{ \jctxsharing{\Gamma}{\Gamma_1}{\Gamma_2} \\ \jprop{\Gamma}{\phi=\phi_1+\phi_2} }
	{ \jctxsharing{\Gamma,\phi}{\Gamma_1,\phi_1}{\Gamma_2,\phi_2} }
\end{mathpar}
\caption{Context sharing rules}
\label{fig:ctxsharing}
\end{figure}

\subsection{Total Free Potential: $\pot{\Gamma}$}

The \emph{free potentials} of a context $\Gamma$, written $\pot{\Gamma}$, include all the potential bindings, as well as outermost annotated potentials of variable bindings.
\begin{alignat}{2}
	\pot{\cdot} & =  0  & \pot{\Gamma,\alpha} & =  \pot{\Gamma} \\
	\pot{\Gamma,x: \tpot{\tsubset{B}{\psi}}{\phi}} & =  \pot{\Gamma} + \subst{x}{\nu}{\phi} \enskip &  \pot{\Gamma,\psi} & = \pot{\Gamma} \\
	\pot{\Gamma,x: \tpot{\p{\tarrowm{y}{T_y}{T}{m}}}{\phi}} & = \pot{\Gamma} + \phi & \pot{\Gamma,\phi} & = \pot{\Gamma} + \phi \\
	\pot{\Gamma,x: \forall\alpha.S} & = \pot{\Gamma}  \Omit{& \pot{\Gamma,x: \tprod{T_l}{T_r}}  & =  \pot{\Gamma} }
\end{alignat}

\subsection{Type Substitution: $\subst{\trefined{B}{\psi}{\phi}}{\alpha}{S}$}

In \typesys, type substitution is restricted to resource-annotated subset types.
The substitution $\subst{\trefined{B}{\psi}{\phi}}{\alpha}{S}$ should take care of logical refinements and potential annotations from both $S$ and $\trefined{B}{\psi}{\phi}$.
Following gives the definition.
\begin{align}
	\subst{U}{\alpha}{\tbool} & =  \tbool \\
	\subst{U}{\alpha}{\tlist{T}} & =  \tlist{\subst{U}{\alpha}{T}} \\
	\subst{U}{\alpha}{m \cdot \beta} & =  m \cdot \beta \\
	\subst{\tpot{\tsubset{B}{\psi}}{\phi}}{\alpha}{m \cdot \alpha} & =   \tpot{\tsubset{m \times B}{\psi}}{ m \times \phi} \\
	\subst{U}{\alpha}{\tsubset{B}{\psi}} & =  \tpot{\tsubset{B'}{\psi \wedge \psi'}}{\phi'}\\
	& \text{where}~\subst{U}{\alpha}{B} =\tpot{\tsubset{B'}{\psi'}}{\phi'} \\
	\subst{U}{\alpha}{\tarrowm{x}{T_x}{T}{m}} & =  \tarrowm{x}{\subst{U}{\alpha}{T_x}}{\subst{U}{\alpha}{T}}{m} \\
	\subst{U}{\alpha}{\tpot{R}{\phi}} & =  \tpot{R'}{\phi + \phi'}\\
	& \text{where}~ \subst{U}{\alpha}{R} = R'^{\phi'} \\
	\subst{U}{\alpha}{\forall\beta.S} & =  \forall\beta. \subst{U}{\alpha}{S}
\end{align}

Type multiplication is defined as follows.
\begin{align}
	m \times \tbool & =  \tbool \\
	m \times \tlist{T} & =  \tlist{m \times T} \\
	m_1 \times (m_2 \cdot \alpha) & =  (m_1 \cdot m_2) \cdot \alpha
\end{align}

\section{Validity Checking in \typesys}
\label{sec:appendixvalidity}

In this section, we define the \emph{validity checking} judgment $\jprop{\Gamma}{\psi}$ where $\Gamma$ is a wellformed context and $\psi$ is a Boolean-sorted refinement.
Intuitively, the judgment states that the formula $\psi$ is always true under any instance of $\Gamma$.
Our approach is to define a set-based denotational semantics for refinements and then reduce the validity checking in \typesys to Presburger arithmetic. 

\paragraph{Semantics of Sorts}
A sort $\Delta$ represents a set $\interps{\Delta}$ of $\Delta$-sorted refinements.
The following gives the definition of $\interps{\Delta}$.
Note that we only define the semantics for sorts that do \emph{not} contain uninterpreted sorts.
We denote such sorts by $\Delta_o$, defined as $\{\bbB,\bbN\}$.
\begin{align}
	\interps{\bbB} & =  \{ \top, \bot \} \\
	\interps{\bbN} & =  \bbZ^+_0 
\end{align}

\paragraph{Semantics of Types}
As we have already done in the sorting rules, scalar types are reflected in the refinement level.
To interpret a wellformed scalar type as a sort without uninterpreted sorts, we define a transformation $\calT_E(\cdot)$ from types to sorts, parametrized by an \emph{environment} that resolves uninterpreted sorts $\delta_\alpha$.
\begin{align}
	\calT_E(\tbool) & = \bbB \\
	\calT_E(\tlist{T}) & = \bbN \\
	\calT_E(m \cdot \alpha) & = E(\delta_\alpha) 
\end{align}

\paragraph{Semantics of Contexts}
To give a meaning to a context $\Gamma$, we need to assign an instance for each variable binding with a scalar type, as well as type variables.
Intuitively, a context $\Gamma$ represents a set of \emph{environments} that resolves both program variables and uninterpreted sorts.
Making use of semantics for sorts and types defined above, we can define $\interps{\Gamma}$ inductively as follows.
\begin{align}
	\interps{\cdot} & = \{ \emptyset \} \\
	\interps{\Gamma,\bindvar{x}{ \tpot{\tsubset{B}{\psi}}{\phi}}} & =  \{ E[x \mapsto \psi] : E \in \interps{\Gamma} \wedge \psi \in \interps{\calT_E(B)}  \} \\
	\interps{\Gamma,\bindvar{x}{ \tpot{(\tarrowm{y}{T_y}{T}{m})}{\phi}}} & =  \interps{\Gamma} \\
	\interps{\Gamma,\bindvar{x}{ \forall\alpha.S}} & = \interps{\Gamma} \\
	\interps{\Gamma,\alpha} & = \{ E[\delta_\alpha \mapsto \Delta] \mid E \in \interps{\Gamma} \wedge \Delta \in \Delta_o  \}   \\
	\interps{\Gamma,\psi} & = \interps{\Gamma} \\
	\interps{\Gamma,\phi} &=  \interps{\Gamma}
\end{align}

\paragraph{Semantics of Refinements}
The meaning of a refinement $\psi$ is defined with respect to its sorting judgment $\jsort{\Gamma}{\psi}{\Delta}$.
The following defines an \emph{evaluation} map $\interp{\psi} : \interps{\Gamma} \to \interps{\Delta}$, by induction on the derivation of the sorting judgment, or essentially structural induction on $\psi$.
\begin{align}
	\interp{x}(E) & = E(x) \\
	\interp{\top}(E) & = \top \\
	\interp{\neg\psi}(E) & = \neg\interp{\psi}(E) \\
	\interp{\psi_1 \wedge \psi_2}(E) & = \interp{\psi_1}(E) \wedge \interp{\psi_2}(E) \\
	\interp{n}(E) & = n \\
	\interp{\psi_1 \le \psi_2}(E) & = \interp{\psi_1}(E) \le \interp{\psi_2}(E) \\
	\interp{\psi_1 + \psi_2}(E) & =\interp{\psi_1}(E) + \interp{\psi_2}(E) \\
	\interp{\psi_1 = \psi_2}(E) & = \interp{\psi_1}(E) = \interp{\psi_2}(E) 
\end{align}

\paragraph{Validity Checking}
Now we show how to assign meanings to contexts and refinements, then the last step to define $\jprop{\Gamma}{\psi}$ is to collect all the refinement constraints mentioned in $\Gamma$.

We first define how to extract constraints from a type binding.
Note that only scalar types (i.e., subset types) can carry logical refinements.
\begin{align}
	\scrB_\Gamma(\bindvar{x}{\trefined{B}{\psi}{\phi}}) & = \subst{x}{\nu}{\psi} \\
	\scrB_\Gamma(\bindvar{x}{\tpot{(\tarrowm{y}{T_y}{T}{m})}{\phi}}) & = \top \\
	\scrB_\Gamma(\bindvar{x}{\forall\alpha.S}) & = \top
\end{align}

Then we define $\scrB(\Gamma)$ to collect all the constraints from variable bindings and path conditions in $\Gamma$.
It is defined inductively on $\Gamma$.
\begin{align}
	\scrB(\cdot) & = \top \\
	\scrB(\Gamma,\bindvar{x}{S}) & = \scrB(\Gamma) \wedge \scrB_\Gamma(\bindvar{x}{S}) \\
	\scrB(\Gamma,\bindvar{x}{\tpot{(\tarrowm{y}{T_y}{T}{m})}{\phi}}) & = \scrB(\Gamma) \\
	\scrB(\Gamma,\alpha) & = \scrB(\Gamma) \\
	\scrB(\Gamma,\psi) & = \scrB(\Gamma) \wedge \psi \\
	\scrB(\Gamma,\phi) & = \scrB(\Gamma)
\end{align}

Now we can define the validity checking judgment $\jprop{\Gamma}{\psi}$.
\[
\jprop{\Gamma}{\psi} \defeq \forall E \in \interps{\Gamma}\!:  \interp{\scrB(\Gamma) \implies \psi}(E)
\]
Further, we can embed our denotational semantics for refinements in Presburger arithmetic, so we can also write the validity checking as the following formula
\[
\forall E \in \interps{\Gamma}\!: E \models \scrB(\Gamma) \implies \psi,
\]
where $\models$ is interpreted in Presburger arithmetic.

\section{Definition of Consistency for \typesys}

To describe soundness of \typesys, we will need a notion of \emph{consistency}.
Basically, given a typing judgment $\jstyping{\Gamma}{v}{S}$ of a value, we want to know that under the context $\Gamma$, $v$ satisfies the logical conditions indicated by $S$, as well as $\Gamma$ has sufficient amount of potential to be stored in $v$ with respect to $S$.

To start with, we need an interpretation $\calI(\cdot)$ that maps interpretable values into refinements.
The following gives an interpretation of our core calculus for \typesys.
\begin{align}
	\calI(\etrue) & =  \top \\
	\calI(\efalse) & =  \bot \\
	\calI(\enil) & = 0 \\
	\calI(\econs{v_h}{v_t}) & = \calI(v_t) + 1 
\end{align}
Note that $\calI(\cdot)$ is only defined on values of scalar types.

Then we can use $\calI(\cdot)$ to transform a \emph{value stack} $V$ to a \emph{refinement environment} $E$ with respect to a context $\Gamma$.
The stack $V$ maps type variables to concrete types and program variables to values.
The environment $E$ is used to define validity checking in former sections.
The following defines the transformation $\calI_V(\Gamma)$ by induction on $\Gamma$.
\begin{align}
	\calI_V(\cdot) & = \emptyset \\
	\calI_V(\Gamma,\bindvar{x}{\tpot{\tsubset{B}{\psi}}{\phi}}) & = \calI_V(\Gamma)[x \mapsto \calI(V(x))] \\
	\calI_V(\Gamma,\bindvar{x}{\tpot{(\tarrowm{y}{T_y}{T}{m})}{\phi}}) & = \calI_V(\Gamma) \\
	\calI_V(\Gamma,\bindvar{x}{\forall\alpha.S}) & = \calI_V(\Gamma) \\
	\calI_V(\Gamma,\alpha) & = \mathbf{let}~E=\calI_V(\Gamma)~\mathbf{in} \\
	& \quad  E[\delta_\alpha \mapsto  \calT_E(V(\alpha))] \\
	\calI_V(\Gamma,\psi) & =\calI_V(\Gamma) \\
	\calI_V(\Gamma,\phi) & = \calI_V(\Gamma)
\end{align}

Now we define how to extract constraints from a value with respect to its type.
It is similar to how we extract constraints from a typing binding in the refinement level.
The differences are that (i) we need to use the interpretation $\calI(\cdot)$ to map values to refinements, (ii) we need to take care of list elements and pair components, (iii) we need to substitute type variables with concrete types, and (iv) for polymorphic type schemas, we assert that the constraints hold for all instantiations.
\begin{align}
	\condv{V}{b}{\trefined{\tbool}{\psi}{\phi}} & = \subst{\calI(b)}{\nu}{\psi} \\
	\condv{V}{[v_1,\cdots,v_n]}{\trefined{\tlist{T}}{\psi}{\phi}} & = \subst{n}{\nu}{\psi}  \wedge {\bigwedge_{i=1}^n \condv{V}{v_i}{T}} \\
	\condv{V}{v}{\trefined{m \cdot \alpha}{\psi}{\phi}} & = \condv{V}{v}{ \subst{V(\alpha)}{\alpha}{\tsubset{m \cdot \alpha}{\psi}} } \\
 	\condv{V}{v}{\tpot{(\tarrowm{x}{T_x}{T}{m})}{\phi}} & = \top \\
 	\condv{V}{v}{\forall\alpha.S} & = \forall \trefined{B}{\psi}{\phi} \!:  \condv{V'}{v}{S} \\
 	& \text{where}~\jwftype{\Gamma}{\tpot{\tsubset{B}{\psi}}{\phi}} \\
 	& \text{and} ~V'=V[\alpha \mapsto \tpot{\tsubset{B}{\psi}}{\phi}]
\end{align}

The following defines how to collect path conditions of a stack $V$ with respect to its typing context $\Gamma$, written $\condc{V}{\Gamma}$.
\begin{align}
	\condc{V}{\cdot} & = \top \\
	\condc{V}{\Gamma,\bindvar{x }{ \tpot{\tsubset{B}{\psi}}{\phi} }} & = \condc{V}{\Gamma} \wedge \condv{V}{V(x)}{\trefined{B}{\psi}{\phi}} \\
	\condc{V}{\Gamma,\bindvar{x }{ \tpot{(\tarrowm{y}{T_y}{T}{m})}{\phi}}} & = \condc{V}{\Gamma} \\
	\condc{V}{\Gamma,\bindvar{x}{ \forall\alpha.S}} & = \condc{V}{\Gamma} \\
	\condc{V}{\Gamma,\alpha} & = \condc{V}{\Gamma} \\
	\condc{V}{\Gamma,\psi} & = \condc{V}{\Gamma} \wedge \psi\\
	\condc{V}{\Gamma,\phi} & = \condc{V}{\Gamma}
\end{align}

Similar to logical refinements, we can also collect potential annotations.
The following defines $\potv{V}{v}{S}$ as the potential stored in the value $v$ with respect to the type $S$ under the stack $V$.
\begin{align}
	\potv{V}{b}{\trefined{\tbool}{\psi}{\phi}} & =  \subst{\calI(b)}{\nu}{\phi} \\
	\potv{V}{[v_1,\cdots,v_n]}{ \trefined{\tlist{T}}{\psi}{\phi} } & =  \subst{n}{\nu}{\phi}  + {\sum_{i=1}^n \potv{V}{v_i}{ T }} \\
	\potv{V}{v}{ \trefined{m \cdot \alpha}{\psi}{\phi} } & =  \potv{V}{v }{ \subst{V(\alpha)}{\alpha}{ \tpot{(m \cdot \alpha)}{\phi}} } \\
	\potv{V}{v}{ \tpot{(\tarrowm{x}{T_x}{T}{m})}{\phi}} & =  \phi\\
	\potv{V}{v}{ \forall\alpha.S } & =  0
\end{align}

Also we have a stack version for potentials $\potc{V}{\Gamma}$.
\begin{align}
	\potc{V}{\cdot} & =  0 \\
	\potc{V}{\Gamma,\bindvar{x }{ \tpot{\tsubset{B}{\psi}}{\phi}}} & =  \potc{V}{\Gamma} + \potv{V}{V(x) }{ \trefined{B}{\psi}{\phi} } \\
	\potc{V}{\Gamma,\bindvar{x}{ \tpot{(\tarrowm{y}{T_y}{T}{m})}{\phi}}}& =  \potc{V}{\Gamma}+ \phi  \\
	\potc{V}{\Gamma,\bindvar{x}{\forall\alpha.S}} & =  \potc{V}{\Gamma} \\
	\potc{V}{\Gamma,\alpha} & =  \potc{V}{\Gamma} \\
	\potc{V}{\Gamma,\psi} & = \potc{V}{\Gamma} \\
	\potc{V}{\Gamma,\phi} & =  \potc{V}{\Gamma} + \phi
\end{align}

Finally, we are able to define two notions of consistency for values and stacks, respectively.

\begin{definition}[Value consistency]\label{de:valconsistency}
 A value $\jval{v}$ is said to be \emph{consistent} with $\jstyping{\Gamma}{v}{S}$, if for all $\jctxtyping{V}{\Gamma}$, $E=\calI_V(\Gamma)$ such that $E \models \condc{V}{\Gamma}$, we have $E \models \condv{V}{v}{S} \wedge \potc{V}{\Gamma} \ge \potv{V}{v}{S}$.
\end{definition}

\begin{definition}[Stack consistency]\label{de:envconsistency}
An environment $V'$ is said to be \emph{consistent} with $\jctxtyping[\Gamma]{V'}{\Gamma'}$, if for for all $\jctxtyping{V}{\Gamma}$, $E = \calI_V(\Gamma)$ such that $E \models \condc{V}{\Gamma}$, we have $E' \models \condc{V,V'}{\Gamma'} \wedge \potc{V}{\Gamma} \ge \potc{V,V'}{\Gamma'}$ where $E' \defeq \calI_{V,V'}(\Gamma,\Gamma')$.
\end{definition}

\makeatletter
\tagsleft@true
\makeatother
\newtagform{nobrackets}[\underline]{}{}
\usetagform{nobrackets}

\section{Proofs for Soundness}
\label{sec:appendixproofs}

\subsection{Progress}

\begin{lemma}\label{lem:consistentcons}
%
	Let $\Gamma = \overline{q \mid \alpha}$.
	\begin{enumerate}
	  \item If $\jwftype{\Gamma}{T}$, then $\enil$ is consistent with $\jstyping{\Gamma}{\enil}{\tsubset{\tlist{T}}{\nu = \calI(\enil)}}$.
	  \item If $\jctxsharing{\Gamma}{\Gamma_1}{\Gamma_2}$, $v_h$ is consistent with $\jstyping{\Gamma_1}{v_h}{T}$, and $v_t$ is consistent with $\jstyping{\Gamma_2}{v_t}{\tsubset{\tlist{T}}{\nu = \calI(v_t)}}$, then $\econs{v_h}{v_t}$ is consistent with $\jstyping{\Gamma}{\econs{v_h}{v_t}}{\tsubset{\tlist{T}}{\nu = \calI(\econs{v_h}{v_t})}}$.
	\end{enumerate}
\end{lemma}
\begin{proof}[Proof of (1)]
		\begin{alignat}{2}
			& \text{Fix}~\jctxtyping{V}{\Gamma}, E=\calI_V(\Gamma)~\text{s.t.}~E\models \condc{V}{\Gamma} \\
			& \jwftype{\Gamma}{T}  & \text{[premise]} \\
			& \jatyping{\Gamma}{\enil}{\tlist{T}} & \text{[typing]} \\
			& \jstyping{\Gamma}{\enil}{\tsubset{\tlist{T}}{\nu = \calI(\enil)}} & \text{[typing]} \\
            & \condv{V}{\enil}{\tsubset{\tlist{T}}{\nu = \calI(\enil)}} \\
            & \quad = \subst{\calI(\enil)}{\nu}{(\nu = \calI(\enil))}  =\top\\
            & \potv{V}{\enil}{\tpot{\tlist{T}}{0}} = 0 \\
            & \potc{V}{\Gamma} = 0 \\
            & E \models \top \wedge 0 \ge 0 \\
            & \text{done}
		\end{alignat}
\end{proof}
\begin{proof}[Proof of (2)]
		\begin{alignat}{2}
			& \text{Fix}~\jctxtyping{V}{\Gamma}, E=\calI_V(\Gamma)~\text{s.t.}~E\models \condc{V}{\Gamma} \\
			& \jctxsharing{\Gamma}{\Gamma_1}{\Gamma_2} & \text{[premise]} \\
			& \quad \implies \potc{V}{\Gamma} = \potc{V}{\Gamma_1}+\potc{V}{\Gamma_2} & \label{eq:foldcons:sharing} \\
			&  \jstyping{\Gamma_1}{v_h}{T}~\text{consistent} & \text{[premise]} \label{eq:foldcons:vhconsist}\\
			& \jstyping{\Gamma_2}{v_t}{\tsubset{\tlist{T}}{\nu = \calI(v_t)}} ~\text{consistent} & \text{[premise]} \label{eq:foldcons:vtconsist} \\
			& \jatyping{\Gamma}{\econs{v_h}{v_t}}{\tlist{T}} & \text{[typing]} \\
			& \jstyping{\Gamma}{\econs{v_h}{v_t}}{\tsubset{\tlist{T}}{\nu = \calI(\econs{v_h}{v_t}}} & \text{[typing]} \\
			& \condv{V}{\econs{v_h}{v_t}}{\tsubset{\tlist{T}}{\nu = \calI(\econs{v_h}{v_t}}} \\
			& \quad = \subst{\calI(\econs{v_h}{v_t}}{\nu}{(\nu = \calI(\econs{v_h}{v_t}))} \wedge \\
			& \qquad \condv{V}{v_h}{T} \wedge \condv{V}{v_t}{\tlist{T}} \\
			& \quad = \condv{V}{v_h}{T} \wedge \condv{V}{v_t}{\tlist{T}} \\
			& \potv{V}{\econs{v_h}{v_t}}{\tpot{\tlist{T}}{0}} = 0 +\\
			& \quad \potv{V}{v_h}{T} + \potv{V}{v_t}{\tlist{T}} \\
			& \quad = \potv{V}{v_h}{T} + \potv{V}{v_t}{\tlist{T}} \\
			& E \models \condv{V}{v_h}{T} \wedge \potc{V}{\Gamma_1} \ge \potv{V}{v_h}{T} & \text{[\eqref{eq:foldcons:vhconsist}]} \\ 
			& E \models \condv{V}{v_t}{\tlist{T}} \wedge \potc{V}{\Gamma_2} \ge \potv{V}{v_t}{\tpot{\tlist{T}}{0}} & \text{[\eqref{eq:foldcons:vtconsist}]} \\
			& \text{done} & \text{[\eqref{eq:foldcons:sharing}]}
		\end{alignat}
\end{proof}

\begin{proposition}\label{prop:stepdelta}
	If $\jstep{e}{e'}{p}{p'}$ and $c \ge 0$, then $\jstep{e}{e'}{p+c}{p'+c}$.
\end{proposition}
\begin{proof}
	By induction on $\jstep{e}{e'}{p}{p'}$.
\end{proof}

\begin{proposition}\label{prop:subtyping}
	If $\jval{v}$, $\jstyping{\Gamma}{v}{T_1}$, $\jsubty{\Gamma}{T_1}{T_2}$, $\jctxtyping{V}{\Gamma}$ and $E=\calI_V(\Gamma)$ such that $E \models \condc{V}{\Gamma}$, then $E \models \condv{V}{v}{T_1} \implies (\condv{V}{v}{T_2} \wedge \potv{V}{v}{T_1} \ge \potv{V}{v}{T_2})$.
\end{proposition}
\begin{proof}
	By induction on $\jsubty{\Gamma}{T_1}{T_2}$.
\end{proof}

\begin{proposition}\label{prop:sharingsplit}
	If $\jval{v}$, $\jstyping{\Gamma}{v}{S}$, $\jsharing{\Gamma}{S}{S_1}{S_2}$, $\jctxtyping{V}{\Gamma}$ and $E = \calI_V(\Gamma)$ such that $E \models \condc{V}{\Gamma}$, then $E \models \potv{V}{v}{S} = \potv{V}{v}{S_1} + \potv{V}{v}{S_2}$.
\end{proposition}
\begin{proof}
	By induction on $\jsharing{\Gamma}{S}{S_1}{S_2}$.
\end{proof}

\begin{lemma}\label{lem:progressatom}
  If $\Gamma = \overline{q \mid \alpha}$, $\jatyping{\Gamma}{a}{B}$, $\jctxtyping{V}{\Gamma}$ and $p \ge \potc{V}{\Gamma}$, then $\jval{a}$ and $a$ is consistent with $\jstyping{\Gamma}{a}{\tsubset{B}{\nu = \calI(a)}}$.
\end{lemma}
\begin{proof}
  By induction on $\jatyping{\Gamma}{a}{B}$:
  \begin{alignat}{2}\footnotesize
    \shortintertext{\bf\textsc{(SimpAtom-True)}}
    & \text{SPS}~a=\etrue,B=\tbool \\
    & \jval{\etrue} & \text{[value]} \\
    & \condv{V}{\etrue}{\tsubset{\tbool}{\nu = \calI(\etrue)}} \\
    & \quad = \subst{\calI(\etrue)}{\nu}{(\nu = \calI(\etrue))} = \top \\
    & \potv{V}{\etrue}{\tpot{\tbool}{0}} = 0 \le \potc{V}{\Gamma}
		\shortintertext{\bf\textsc{(SimpAtom-False)}}
		& \text{SPS}~a=\efalse, B=\tbool \\
		& \jval{\efalse} & \text{[value]} \\
		& \condv{V}{\efalse}{\tsubset{\tbool}{\nu = \calI(\efalse)}} \\
		& \quad = \subst{\calI(\efalse)}{\nu}{(\nu = \calI(\efalse))} = \top \\
		& \potv{V}{\efalse}{\tpot{\tbool}{0}} = 0 \le \potc{V}{\Gamma}
	\shortintertext{\bf\textsc{(SimpAtom-Nil)}}
	& \text{SPS}~a=\enil,B = \tlist{T} \\
	& \jval{\enil} & \text{[value]} \\
	& \enil~\text{consistent} & \text{[\lemref{consistentcons}]}
	\shortintertext{\bf\textsc{(SimpAtom-Cons)}}
	& \text{SPS}~a=\econs{\hat{a}_h}{a_t}, B = \tlist{T} \\
	& \Gamma~\text{contains no variables} \implies \jval{\hat{a}_h} \label{eq:proga:conshead} \\
	& \jctxsharing{\Gamma}{\Gamma_1}{\Gamma_2} & \text{[premise]} \label{eq:proga:consshare} \\
	& \jstyping{\Gamma_1}{\hat{a}_h}{T} & \text{[premise]} \label{eq:proga:conspremh} \\
	& \jatyping{\Gamma_2}{a_t}{\tlist{T}} & \text{[premise]} \label{eq:proga:conspremt} \\
	& \hat{a}_h~\text{consistent} & \text{[\theoref{progress}, \eqref{eq:proga:conshead}, \eqref{eq:proga:conspremh}]} \\
	& \jval{a_t}, a_t~\text{consistent} & \text{[ind. hyp., \eqref{eq:proga:conspremt}]} \\
	& \jval{\econs{\hat{a}_t}{a_t}} & \text{[value]} \\
	& \econs{\hat{a}_h}{a_t}~\text{consistent} & \text{[\lemref{consistentcons}]}
  \end{alignat}  
\end{proof}

\begin{theorem}[Progress]\label{the:progress}
	If  $\Gamma=\overline{q \mid \alpha}$, $\jstyping{\Gamma}{e}{S}$, $\jctxtyping{V}{\Gamma}$ and $p \ge \potc{V}{\Gamma}$, then either $\jval{e}$ and $e$ is consistent with $\jstyping{\Gamma}{e}{S}$, or there exist $e'$ and $p'$ such that $\jstep{e}{e'}{p}{p'}$.
\end{theorem}
\begin{proof}
	By induction on $\jstyping{\Gamma}{e}{S}$:
	\begin{alignat}{2}\footnotesize
	   \shortintertext{\bf\textsc{(T-SimpAtom)}}
	   & \text{SPS}~e=a, S = \tsubset{B}{\nu = \calI(a)} \\
	   & \jval{a}, a~\text{consistent} & \text{[\lemref{progressatom}]}
	   %
		%
		\shortintertext{\bf\textsc{(T-Imp)}}
		& \text{SPS}~e=\eimp,S=T \\
		& \jprop{\Gamma}{\bot} & \text{[premise]} \\
		&  \top \implies \bot\\
		& \text{exfalso} 
		\shortintertext{\bf\textsc{(T-Consume-P)}}
		& \text{SPS}~\Gamma=(\Gamma',c), e=\econsume{c}{e_0}, c \ge 0 \\
		& p \ge \potc{V}{\Gamma}= \potc{V}{\Gamma'} + c \ge c \\
		& \jstep{e}{e_0}{p}{p-c} & \text{[eval.]}
		\shortintertext{\bf\textsc{(T-Consume-N)}}
		& \text{SPS}~e=\econsume{c}{e_0},c<0 \\
		& \jstep{e}{e_0}{p}{p-c} & \text{[eval.]}
		\shortintertext{\bf\textsc{(T-Cond)}}
		& \text{SPS}~e=\econd{a_0}{e_1}{e_2},S=T \\
        & \jatyping{\Gamma}{a_0}{\tbool} & \text{[premise]} \label{eq:prog:conde0} \\
        & \jval{a_0} & \text{[\lemref{progressatom}]} \label{eq:prog:conde0val} \\
        & \text{inv. on \eqref{eq:prog:conde0} with \eqref{eq:prog:conde0val}} \\
        & \textbf{case}~a_0=\etrue \\
        & \enskip \jstep{e}{e_1}{p}{p} & \text{[eval.]} \\
        & \textbf{case}~a_0=\efalse \\
        & \enskip \jstep{e}{e_2}{p}{p} & \text{[eval.]}
		\shortintertext{\bf\textsc{(T-MatL)}}
		& \text{SPS}~e=\ematl{a_0}{e_1}{x_h}{x_t}{e_2},S=T' \\
		& \jctxsharing{\Gamma}{\Gamma_1}{\Gamma_2}  & \text{[premise]} \\
		& \jatyping{\Gamma_1}{a_0}{\tlist{T}} & \text{[premise]} \label{eq:prog:matle0} \\
		& \jval{a_0} & \text{[\lemref{progressatom}]} \label{eq:prog:matle0val} \\
        & \text{inv. on \eqref{eq:prog:matle0} with \eqref{eq:prog:matle0val}} \\
        & \textbf{case}~a_0=\enil \\
        & \enskip \jstep{e}{e_1}{p}{p} & \text{[eval.]} \\
        & \textbf{case}~a_0=\econs{v_h}{v_t} \\
        & \enskip \jstep{e}{\subst{v_h,v_t}{x_h,x_t}e_2}{p}{p} & \text{[eval.]} 
		\shortintertext{\bf\textsc{(T-Let)}}
		& \text{SPS}~e=\elet{e_1}{x}{e_2},S=T_2 \\
		& \jctxsharing{\Gamma}{\Gamma_1}{\Gamma_2}  & \text{[premise]}  \\
		& \quad \implies \potc{V}{\Gamma}=\potc{V}{\Gamma_1}+\potc{V}{\Gamma_2} \label{eq:prog:letsplit} \\
		& \jstyping{\Gamma_1}{e_1}{S_1} & \text{[premise]} \label{eq:prog:lete1} \\
		& p \ge \potc{V}{\Gamma_1} & \text{[asm., \eqref{eq:prog:letsplit}]} \label{eq:prog:letindhyp} \\
		& \text{ind. hyp. on \eqref{eq:prog:lete1} with \eqref{eq:prog:letindhyp}} \\
		& \textbf{case}~\jstep{e_1}{e_1'}{p}{p'}  \\
		& \enskip \jstep{e}{\elet{e_1'}{x}{e_2}}{p}{p'} & \text{[eval.]} \\
		& \textbf{case}~\jval{e_1}  \\
		& \enskip \jstep{e}{\subst{e_1}{x}{e_2}}{p}{p} & \text{[eval.]}
		\shortintertext{\bf\textsc{(T-App)}}
		& \text{SPS}~e=\eapp{\hat{a}_1}{\hat{a}_2}, S=T \\
		& \jctxsharing{\Gamma}{\Gamma_1}{\Gamma_2} & \text{[premise]} \\
		& \jstyping{\Gamma_1}{\hat{a}_1}{\tarrowm{x}{T_x}{T}{1}} & \text{[premise]} \label{eq:prog:appe1} \\
		& \jstyping{\Gamma_2}{\hat{a}_2}{T_x} & \text{[premise]} \label{eq:prog:appe2} \\
		&\Gamma~\text{contains no variables} \\
		& \quad \implies \jval{\hat{a}_1,\hat{a}_2} \label{eq:prog:appevals} \\
        & \text{inv. on \eqref{eq:prog:appe1} with \eqref{eq:prog:appevals}} \\
        & \textbf{case}~e_1=\eabs{x}{e_0} \\
        & \enskip \jstep{e}{\subst{\hat{a}_2}{x}{e_0}}{p}{p} & \text{[eval.]} \\
        & \textbf{case}~e_2 =\efix{f}{x}{e_0} \\
        & \enskip \jstep{e}{\subst{\efix{f}{x}{e_0},\hat{a}_2}{f,x}{e_0}}{p}{p} & \text{[eval.]}
        \shortintertext{\bf\textsc{(T-App-SimpAtom)}}
        & \text{SPS}~e=\eapp{\hat{a}_1}{a_2},S=\subst{\calI(a_2)}{x}{T} \\
        & \jctxsharing{\Gamma}{\Gamma_1}{\Gamma_2} & \text{[premise]} \\
        & \jstyping{\Gamma_1}{\hat{a}_1}{\tarrowm{x}{\trefined{B}{\psi}{\phi}}{T}{1}} & \text{[premise]} \label{eq:prog:appatome1} \\
        & \Gamma~\text{contains no variables} \\
        & \quad \implies \jval{\hat{a}_1} \label{eq:prog:appatome1val} \\
        & \jstyping{\Gamma_2}{\hat{a}_2}{\trefined{B}{\psi}{\phi}} & \text{[premise]} \\
        & \jval{a_2} & \text{[\lemref{progressatom}]} \label{eq:prog:appatome2val} \\
        & \text{inv. on \eqref{eq:prog:appatome1} with \eqref{eq:prog:appatome1val}} \\
        & \textbf{case}~e_1=\eabs{x}{e_0} \\
        & \enskip \jstep{e}{\subst{a_2}{x}{e_0}}{p}{p} & \text{[eval.]} \\
        & \textbf{case}~e_2 =\efix{f}{x}{e_0} \\
        & \enskip \jstep{e}{\subst{\efix{f}{x}{e_0},a_2}{f,x}{e_0}}{p}{p} & \text{[eval.]}
		\shortintertext{\bf\textsc{(T-Abs)}}
		& \text{SPS}~e =\eabs{x}{e_0},S = \tarrow{x}{T_x}{T} \\
		& \jval{\eabs{x}{e_0}} & \text{[value]} \\
		& \condv{V}{\eabs{x}{e_0}}{\tarrow{x}{T_x}{T}} = \top \\
		& \potv{V}{\eabs{x}{e_0}}{\tpot{(\tarrow{x}{T_x}{T})}{0}} = 0 \le \potc{V}{\Gamma}
		\shortintertext{\bf\textsc{(T-Abs-Lin)}}
		& \text{SPS}~\Gamma=m\cdot \Gamma', e =\eabs{x}{e_0},S = \tarrowm{x}{T_x}{T}{m} \\
		& \jval{\eabs{x}{e_0}} & \text{[value]} \\
		& \condv{V}{\eabs{x}{e_0}}{\tarrowm{x}{T_x}{T}{m}} = \top \\
		& \potv{V}{\eabs{x}{e_0}}{\tpot{(\tarrowm{x}{T_x}{T}{m})}{0}} = 0 \le \potc{V}{\Gamma}
		\shortintertext{\bf\textsc{(T-Fix)}}
		& \text{SPS}~e=\efix{f}{x}{e_0},S =R,R=\tarrow{x}{T_x}{T} \\
		& \jstyping{\Gamma,\bindvar{f}{ R},\bindvar{x}{T_x} }{e_0}{T} & \text{[premise]} \label{eq:prog:fixlambda} \\
		& \jval{\efix{f}{x}{e_0}} & \text{[value]} \\
		& \condv{V}{\efix{f}{x}{e_0}}{R} = \top \\
		& \potv{V}{\efix{f}{x}{e_0}}{\tpot{R}{0}} = 0 \le \potc{V}{\Gamma}
		%
		%
		\shortintertext{\bf\textsc{(S-Gen)}}
		& \text{SPS}~e=v,S=\forall\beta.S' \\
		& \jstyping{\Gamma,\beta}{v}{S'} & \text{[premise]} \label{eq:prog:genprem} \\
		& \jval{v} & \text{[premise]} \\
		& \potv{V}{v}{\forall\beta.S'} = 0 \le \potc{V}{\Gamma} \\
		& \textbf{for all $\jwftype{\Gamma}{\tpot{\tsubset{B}{\psi}}{\phi}}$} \\
		& \enskip \textbf{let}~V'=V[\beta \mapsto \tpot{\tsubset{B}{\psi}}{\phi}] \\
		& \enskip \potc{V'}{\Gamma,\beta} = \potc{V}{\Gamma} \\
		& \enskip \text{ind. hyp. on \eqref{eq:prog:genprem} with $p \ge \potc{V'}{\Gamma,\beta}$} \\
		& \enskip \textbf{case}~\jstep{v}{e'}{p}{p'} \\
		& \quad \text{contradict}~\jval{v} \\
		& \enskip \textbf{case}~\jval{v} \\
		& \quad \condv{V'}{v}{S'} = \top & \text{[ind. hyp.]} \\
		& \implies \condv{V}{v}{\forall\beta.S'} = \top
		\shortintertext{\bf\textsc{(S-Inst)}}
		& \text{SPS}~S=\subst{\tpot{\tsubset{B}{\psi}}{\phi}}{\alpha'}{S'} \\
		& \jstyping{\Gamma}{e}{\forall\alpha'.S'} & \text{[premise]} \label{eq:prog:instprem} \\
		& \text{ind.~hyp. on \eqref{eq:prog:instprem} with $p \ge \potc{V}{\Gamma}$} \\
		& \textbf{case}~\jstep{e}{e'}{p}{p'} \\
		& \enskip \text{done} \\
		& \textbf{case}~\jval{e} \\
		& \enskip \condv{V}{e}{\forall\alpha'.S'} = \top & \text{[ind. hyp.]} \\
		& \enskip  \condv{V[\alpha' \mapsto \tpot{\tsubset{B}{\psi}}{\phi}]}{e}{S'} = \top \\
		& \enskip  \condv{V}{e}{\subst{\tpot{\tsubset{B}{\psi}}{\phi}}{\alpha'}{S'}} = \top \\
		& \enskip \jsharing{\Gamma,\alpha'}{S'}{S'}{S'} & \text{[wellformed.]} \\
		& \enskip  \potv{V[\alpha' \mapsto \tpot{\tsubset{B}{\psi}}{\phi}]}{e}{S'} = 0 & \text{[Prop.~\ref{prop:sharingsplit}]} \\
		& \enskip  \potv{V}{e}{\subst{\tpot{\tsubset{B}{\psi}}{\phi}}{\alpha'}{S'}} = 0 \\
		& \enskip  \potv{V}{e}{\subst{\tpot{\tsubset{B}{\psi}}{\phi}}{\alpha'}{S'}} \le \potc{V}{\Gamma}
		\shortintertext{\bf\textsc{(S-Subtype)}}
		& \text{SPS}~S=T_2 \\
		& \jstyping{\Gamma}{e}{T_1} & \text{[premise]} \label{eq:prog:subtypeprem} \\
		& \jsubty{\Gamma}{T_1}{T_2} & \text{[premise]} \label{eq:prog:subtyperel} \\
		& \text{ind. hyp. on \eqref{eq:prog:subtypeprem} with $p \ge \potc{V}{\Gamma}$} \\
		& \textbf{case}~\jstep{e}{e'}{p}{p'} \\
		& \enskip \text{done} \\
		& \textbf{case}~\jval{e} \\
		& \enskip \condv{V}{e}{T_1} = \top & \text{[ind. hyp.]} \\
		& \enskip \condv{V}{e}{T_1} \implies \condv{V}{e}{T_2} & \text{[Prop.~\ref{prop:subtyping}, \eqref{eq:prog:subtyperel}]} \\
		& \enskip \condv{V}{e}{T_2} = \top \\
		& \enskip \potv{V}{e}{T_1} \le \potc{V}{\Gamma} & \text{[ind. hyp.]} \\ 
		& \enskip \condv{V}{e}{T_1} \implies ( \potv{V}{e}{T_1} \ge \potv{V}{e}{T_2}) & \text{[Prop.~\ref{prop:subtyping}, \eqref{eq:prog:subtyperel}]} \\
		& \enskip \potv{V}{e}{T_2} \le \potc{V}{\Gamma}
		\shortintertext{\bf\textsc{(S-Transfer)}}
		& \jstyping{\Gamma'}{e}{S} & \text{[premise]} \label{eq:prog:transe} \\
		& \jprop{\Gamma}{\pot{\Gamma}=\pot{\Gamma'}} & \text{[premise]} \\
		& \Gamma'=\overline{q' \mid \alpha} \wedge \potc{V}{\Gamma}=\potc{V}{\Gamma'} \label{eq:prog:transequal} \\
		& p \ge \potc{V}{\Gamma'} \label{eq:prog:transindhyp} \\
		& \text{ind. hyp. on \eqref{eq:prog:transe} with \eqref{eq:prog:transindhyp}} \\
		& \textbf{case}~\jstep{e}{e'}{p}{p'} \\
		& \enskip \text{done} \\
		& \textbf{case}~\jval{e} \\
		& \enskip \condv{V}{e}{S} = \top & \text{[ind. hyp.]} \\
		& \enskip \potv{V}{e}{S} \le \potc{V}{\Gamma'} & \text{[ind. hyp.]} \\
		& \enskip \potv{V}{e}{S} \le \potc{V}{\Gamma} & \text{[\eqref{eq:prog:transequal}]}
		\shortintertext{\bf\textsc{(S-Relax)}}
		& \text{SPS}~\Gamma=(\Gamma',\phi'), S=\tpot{R}{\phi+\phi'} \\
		& \jstyping{\Gamma'}{e}{\tpot{R}{\phi}} & \text{[premise]}  \label{eq:prog:relaxprem} \\
		& p \ge \potc{V}{\Gamma',\phi'} = \potc{V}{\Gamma'}+\phi' \label{eq:prog:relaxindhyp} \\
		& \text{ind. hyp. on \eqref{eq:prog:relaxprem} with \eqref{eq:prog:relaxindhyp}} \\
		& \textbf{case}~\jstep{e}{e'}{p}{p'} \\
		& \enskip \text{done} \\
		& \textbf{case}~\jval{e} \\
		& \enskip \condv{V}{e}{R} = \top & \text{[ind. hyp.]} \\
		& \enskip \potv{V}{e}{\tpot{R}{\phi}} \le \potc{V}{\Gamma'} & \text{[ind. hyp.]} \\
		& \enskip \potv{V}{e}{\tpot{R}{\phi+\phi'}} \le \potc{V}{\Gamma',\phi'} & \text{[\eqref{eq:prog:relaxindhyp}]}
	\end{alignat}
\end{proof}

\subsection{Substitution}

\begin{proposition}\label{prop:ctxweaken}
	If $\jstyping{\Gamma}{e}{S}$ and $\jwfctxt{\Gamma,\Gamma'}$, then $\jstyping{\Gamma,\Gamma'}{e}{S}$.
\end{proposition}
\begin{proof}
	By induction on $\jstyping{\Gamma}{e}{S}$.
\end{proof}

\begin{proposition}\label{prop:ctxrelax}
	If $\jstyping{\Gamma_1}{e}{S}$ and $\jctxsharing{\Gamma}{\Gamma_1}{\Gamma_2}$, then $\jstyping{\Gamma}{e}{S}$.
\end{proposition}
\begin{proof}
	By induction on $\jstyping{\Gamma_1}{e}{S}$.
\end{proof}

\begin{proposition}\label{prop:typingvalinterp}
	If $\jstyping{\Gamma}{v}{\tpot{\tsubset{B}{\psi}}{\phi}}$ and $\jval{v}$, then $\jstyping{\Gamma}{v}{\tpot{\tsubset{B}{\nu = \calI(v)}}{\phi}}$.
\end{proposition}
\begin{proof}
	By induction on $\jstyping{\Gamma}{v}{\tpot{\tsubset{B}{\psi}}{\phi}}$.
\end{proof}

\begin{proposition}\label{prop:freetofree}
	If $\jstyping{\Gamma}{v}{\tpot{R}{\phi}}$ and $\jval{v}$, then $\jprop{\Gamma}{\pot{\Gamma} \ge \subst{\calI(v)}{\nu}{\phi}}$.
\end{proposition}
\begin{proof}
	By induction on $\jstyping{\Gamma}{v}{\tpot{R}{\phi}}$.
\end{proof}

\begin{proposition}\label{prop:typingvalsharingsplit}
	If $\jstyping{\Gamma}{v}{S}$, $\jsharing{\Gamma}{S}{S_1}{S_2}$ and $\jval{v}$, then there exist $\Gamma_1$ and $\Gamma_2$ such that $\jctxsharing{\Gamma}{\Gamma_1}{\Gamma_2}$, and $\jstyping{\Gamma_1}{v}{S_1}$, $\jstyping{\Gamma_2}{v}{S_2}$.
\end{proposition}
\begin{proof}
	By induction on $\jstyping{\Gamma}{v}{S}$.
\end{proof}

\begin{proposition}\label{prop:typingvalzero}
	If $\jstyping{\Gamma}{v}{S}$, $\jsharing{\Gamma}{S}{S}{S}$ and $\jval{v}$, then there exists $\Gamma'$ such that $\jctxsharing{\Gamma}{\Gamma}{\Gamma'}$ (so $\jctxsharing{\Gamma'}{\Gamma'}{\Gamma'}$), and $\jstyping{\Gamma'}{v}{S}$.
\end{proposition}
\begin{proof}
	By induction on $\jstyping{\Gamma}{v}{S}$.
\end{proof}

\begin{lemma}\label{lem:substprop}
	If $\Gamma,\psi,\Gamma' \vdash \calJ$ and $\jprop{\Gamma}{\psi}$, then $\Gamma,\Gamma' \vdash \calJ$.
\end{lemma}
\begin{proof}
	By induction on $\Gamma,\psi,\Gamma' \vdash \calJ$.
\end{proof}

\begin{lemma}\label{lem:substintoref}
	Suppose $\calJ$ is a judgment other than typing.
	\begin{enumerate}
		\item If $\Gamma_1,\bindvar{x}{\tpot{\tsubset{B}{\psi}}{\phi}},\Gamma' \vdash \calJ$, $\jstyping{\Gamma_2}{t}{\tpot{\tsubset{B}{\psi}}{\phi}}$, $\jval{t}$ and $\jctxsharing{\Gamma}{\Gamma_1}{\Gamma_2}$, then $\Gamma,\subst{\calI(t)}{x}{\Gamma'} \vdash \subst{\calI(t)}{x}{\calJ}$.
		\item If $\Gamma_1,\bindvar{x}{S_x},\Gamma' \vdash \calJ$, $S_x$ is non-scalar/poly, $\jstyping{\Gamma_2}{t}{S_x}$, $\jval{t}$ and $\jctxsharing{\Gamma}{\Gamma_1}{\Gamma_2}$, then $\Gamma,\Gamma' \vdash \calJ$.
	\end{enumerate}
\end{lemma}
\begin{proof}
	By induction on $\Gamma,\bindvar{x}{S_x},\Gamma' \vdash \calJ$.
\end{proof}

\begin{lemma}\label{lem:substatom}\
  \begin{enumerate}
    \item If $\jatyping{\Gamma_1,\bindvar{x}{\trefined{B_x}{\psi}{\phi}},\Gamma'}{e}{B}$, $\jstyping{\Gamma_2}{t}{\trefined{B_x}{\psi}{\phi}}$, $\jval{t}$ and $\jctxsharing{\Gamma}{\Gamma_1}{\Gamma_2}$, then $\jatyping{\Gamma,\subst{\calI(t)}{x}{\Gamma'}}{\subst{t}{x}{a}}{\subst{\calI(t)}{x}{B}}$. 
    \item If $\jatyping{\Gamma_1,\bindvar{x}{S_x},\Gamma'}{a}{B}$, $S_x$ is non-scalar/poly, $\jstyping{\Gamma_2}{t}{S_x}$, $\jval{t}$ and $\jctxsharing{\Gamma}{\Gamma_1}{\Gamma_2}$, then $\jatyping{\Gamma,\Gamma'}{\subst{t}{x}{a}}{B}$.
  \end{enumerate}
\end{lemma}
\begin{proof}[Proof of (1)]
  By induction on $\jatyping{\Gamma_1,\bindvar{x}{\trefined{B_x}{\psi}{\phi}},\Gamma'}{a}{B}$:
  \begin{alignat}{2}
    \shortintertext{\bf{\textsc{(SimpAtom-Var)}$=$}}
    & \text{SPS}~a=x, B=B_x\\
		& \subst{t}{x}{a} = t, \subst{\calI(t)}{x}{B} = B_x \\
		& \jstyping{\Gamma}{t}{\tpot{\tsubset{B_x}{\psi}}{\phi}} & \text{[Prop.~\ref{prop:ctxrelax}]} \\
		& \jstyping{\Gamma}{t}{\tpot{\tsubset{B_x}{\nu = \calI(t)}}{\phi}} & \text{[Prop.~\ref{prop:typingvalinterp}]} \\
		& \jstyping{\Gamma,\subst{\calI(t)}{x}{\Gamma'}}{t}{\tpot{\tsubset{B_x}{\nu = \calI(t)}}{\phi}} & \text{[Prop.~\ref{prop:ctxweaken}]} \\
		& \jatyping{\Gamma,\subst{\calI(t)}{x}{\Gamma'}}{t}{B_x} & \text{[typing]}
		\shortintertext{\bf{\textsc{(SimpAtom-Var)}$\neq$}}
		& \text{SPS}~a=y \\
		& \subst{t}{x}{a} = y \\
		& \textbf{case}~y \in \Gamma \\
		& \enskip B = \text{base of}~ \Gamma_1(y) \\
		& \enskip \jsharing{\Gamma}{\Gamma(y)}{\Gamma_1(y)}{\Gamma_2(y)} \\
		& \enskip \Gamma(y) = \trefined{B}{\psi'}{\phi'} & \\
		& \enskip \jatyping{\Gamma,\subst{\calI(t)}{x}{\Gamma'}}{y}{B} & \text{[typing]} \\
		& \textbf{case}~y \in \Gamma' \\
		& \enskip B = \text{base of}~\Gamma'(y), \Gamma'(y) = \trefined{B}{\psi'}{\phi'} \\
		& \enskip (\subst{\calI(t)}{x}{\Gamma'})(y) =\\
		& \enskip \quad \trefined{\subst{\calI(t)}{x}{B}}{\subst{\calI(t)}{x}{\psi'}}{\subst{\calI(t)}{x}{\phi'}} \\
		& \enskip \jatyping{\Gamma,\subst{\calI(t)}{x}{\Gamma'}}{y}{\subst{\calI(t)}{x}{B}} & \text{[typing]}
    \shortintertext{\bf\textsc{(SimpAtom-Nil)}}
    & \text{SPS}~a=\enil,B=\tlist{T} \\
    & \jwftype{\Gamma_1,\bindvar{x}{\trefined{B_x}{\psi}{\phi}},\Gamma'}{T} & \text{[premise]} \\
    & \jwftype{\Gamma,\subst{\calI(t)}{x}{\Gamma'}}{\subst{\calI(t)}{x}{T}} & \text{[\lemref{substintoref}]} \\
    & \jatyping{\Gamma,\subst{\calI(t)}{x}{\Gamma'}}{\enil}{\tlist{\subst{\calI(t)}{x}{T}}} & \text{[typing]}
    \shortintertext{\bf\textsc{(SimpAtom-Cons)}}
    & \text{SPS}~a=\econs{\hat{a}_h}{a_t},B=\tlist{T} \\
    & \vdash \Gamma_1,\bindvar{x}{\trefined{B_x}{\psi}{\phi}},\Gamma' \sharing \\
    & \quad \Gamma_{11},\bindvar{x}{\trefined{B_1}{\psi}{\phi_1},\Gamma_1'} \mid \\
    & \quad \Gamma_{12},\bindvar{x}{\trefined{B_2}{\psi}{\phi_2},\Gamma_2'} & \text{[premise]} \\
    & \jstyping{\Gamma_{11},\bindvar{x}{\trefined{B_1}{\psi}{\phi_1}},\Gamma_1'}{\hat{a}_h}{T} & \text{[premise]}  \label{eq:substa:conspremh}\\
    & \jatyping{\Gamma_{12},\bindvar{x}{\trefined{B_2}{\psi}{\phi_2}},\Gamma_2'}{a_t}{\tlist{T}} & \text{[premise]} \label{eq:substa:conspremt} \\
    & \text{exist}~\Gamma_{21},\Gamma_{22}~\text{s.t.}~\jctxsharing{\Gamma_2}{\Gamma_{21}}{\Gamma_{22}}, \\
    & \quad \jstyping{\Gamma_{21}}{t}{\trefined{B_1}{\psi}{\phi_1}},\jstyping{\Gamma_{22}}{t}{\trefined{B_2}{\psi}{\phi_2}} & \text{[\propref{typingvalsharingsplit}]} \\
    & \sharing(\Gamma_{11},\Gamma_{21}), \subst{\calI(t)}{x}{\Gamma_1'} \vdash \\
    & \quad \subst{t}{x}{\hat{a}_h} \dblcolon \subst{\calI(t)}{x}{T} & \text{[\theoref{substitution}, \eqref{eq:substa:conspremh}]} \\
    & \text{ind. hyp. on \eqref{eq:substa:conspremt}} \\
    & \sharing(\Gamma_{12},\Gamma_{22}), \subst{\calI(t)}{x}{\Gamma_2'} \vdash \\
    & \quad \subst{t}{x}{a_t} : \tlist{\subst{\calI(t)}{x}{T}} \\
    & \jctxsharing{\Gamma}{\Gamma_1}{\Gamma_2} \implies & \\
    & \quad \jctxsharing{\Gamma}{(\sharing(\Gamma_{11},\Gamma_{21})}{\sharing(\Gamma_{12},\Gamma_{22})} \\
    & \Gamma,\bindvar{x}{\trefined{B_x}{\psi}{\phi}} \vdash \Gamma' \sharing \Gamma_1'\mid \Gamma_2' \implies  \\
    & \quad \Gamma \vdash \subst{\calI(t)}{x}{\Gamma'} \sharing \subst{\calI(t)}{x}{\Gamma_1'} \mid \subst{\calI(t)}{x}{\Gamma_2'} & \text{[\lemref{substintoref}]} \\
    & \jatyping{\Gamma,\subst{\calI(t)}{x}{\Gamma'}}{a}{\tlist{\subst{\calI(t)}{x}{T}}} & \text{[typing]}
  \end{alignat}
\end{proof}

\begin{theorem}[Substitution]\label{the:substitution}\
	\begin{enumerate}
		\item If $\jstyping{\Gamma_1,\bindvar{x}{\tpot{\tsubset{B}{\psi}}{\phi}},\Gamma'}{e}{S}$, $\jstyping{\Gamma_2}{t}{\tpot{\tsubset{B}{\psi}}{\phi}}$, $\jval{t}$ and $\jctxsharing{\Gamma}{\Gamma_1}{\Gamma_2}$, then $\jstyping{\Gamma,\subst{\calI(t)}{x}{\Gamma'}}{\subst{t}{x}{e}}{\subst{\calI(t)}{x}{S}}$.
		\item If $\jstyping{\Gamma_1,\bindvar{x}{S_x},\Gamma'}{e}{S}$, $S_x$ is non-scalar/poly, $\jstyping{\Gamma_2}{t}{S_x}$, $\jval{t}$ and $\jctxsharing{\Gamma}{\Gamma_1}{\Gamma_2}$, then $\jstyping{\Gamma,\Gamma'}{\subst{t}{x}{e}}{S}$.
	\end{enumerate}
\end{theorem}
\begin{proof}[Proof of (1)]
	By induction on $\jstyping{\Gamma_1,\bindvar{x}{\tpot{\tsubset{B}{\psi}}{\phi}},\Gamma'}{e}{S}$:
	\begin{alignat}{2}
	     \shortintertext{\bf\textsc{(T-SimpAtom)}}
	     & \text{SPS}~e=a,S=\tsubset{B'}{\nu = \calI(a)} \\
	     & \jatyping{\Gamma_1,\bindvar{x}{\trefined{B}{\psi}{\phi}},\Gamma'}{a}{B'} & \text{[premise]} \\
	     & \jatyping{\Gamma,\subst{\calI(t)}{x}{\Gamma'}}{\subst{t}{x}{a}}{\subst{\calI(t)}{x}{B'}} & \text{[\lemref{substatom}]} \\
	     &  {\Gamma,\subst{\calI(t)}{x}{\Gamma'}} \vdash {\subst{t}{x}{a}} \dblcolon \\
	     & \quad {\tsubset{\subst{\calI(t)}{x}{B'}}{\nu = \calI(\subst{t}{x}{a})}} & \text{[typing]} \\
	     & \tsubset{\subst{\calI(t)}{x}{B'}}{\nu = \calI(\subst{t}{x}{a})} = \\
	     & \quad \subst{\calI(t)}{x}{(\tsubset{B'}{\nu = \calI(a)})}
		\shortintertext{\bf{\textsc{(T-Var)}$=$}}
		& \text{SPS}~e=x,S= \tpot{\tsubset{B}{\psi}}{\phi} \\
		& \subst{t}{x}{e} = t, \subst{\calI(t)}{x}{S} = \tpot{\tsubset{B}{\psi}}{\phi} \\
		& \jstyping{\Gamma}{t}{\tpot{\tsubset{B}{\psi}}{\phi}} & \text{[Prop.~\ref{prop:ctxrelax}]} \\
		& \jstyping{\Gamma,\subst{\calI(t)}{x}{\Gamma'}}{t}{\tpot{\tsubset{B}{\psi}}{\phi}} & \text{[Prop.~\ref{prop:ctxweaken}]}
		\shortintertext{\bf{\textsc{(T-Var)}$\ne$}}
		& \text{SPS}~e=y,S=\Gamma(y) \\
		& \subst{t}{x}{e} = y \\
		& \textbf{case}~y \in \Gamma \\
		& \enskip \text{WLOG}~\Gamma(y) = \tpot{\tsubset{B'}{\psi'}}{\phi'} \\
		& \enskip \jsharing{\Gamma}{\Gamma(y)}{\Gamma_1(y)}{\Gamma_2(y)} \\
		& \enskip \textbf{let}~\Gamma_1(y) = \tpot{\tsubset{B_1'}{\psi_1'}}{\phi_1'} \\
		& \enskip \subst{\calI(t)}{x}{S} = S = \tpot{\tsubset{B_1'}{\psi_1'}}{\phi_1'} \\
		& \enskip \jstyping{\Gamma,\subst{\calI(t)}{x}{\Gamma'}}{y}{\tpot{\tsubset{B'}{\psi'}}{\phi'}} & \text{[typing]} \\
		& \textbf{case}~y \in \Gamma' \\
		& \enskip \text{WLOG}~\Gamma'(y) = \tpot{\tsubset{B'}{\psi'}}{\phi'} \\
		& \enskip S = \tpot{\tsubset{B'}{\psi'}}{\phi'} \\
		& \enskip \subst{\calI(t)}{x}{S} = \\
		& \enskip \quad \tpot{\tsubset{\subst{\calI(t)}{x}{B'}}{\subst{\calI(t)}{x}{\psi'}}}{\subst{\calI(t)}{x}{\phi'}} \\
		& \enskip (\subst{\calI(t)}{x}{\Gamma'})(y) = \\
		& \enskip \quad \tpot{\tsubset{\subst{\calI(t)}{x}{B'}}{\subst{\calI(t)}{x}{\psi'}}}{\subst{\calI(t)}{x}{\phi'}} \\
		& \enskip \Gamma,\subst{\calI(t)}{x}{\Gamma'} \vdash y \dblcolon \\
		& \enskip \quad  \tpot{\tsubset{\subst{\calI(t)}{x}{B'}}{ \subst{\calI(t)}{x}{\psi'} }}{\subst{\calI(t)}{x}{\phi'}} & \text{[typing]}
		\shortintertext{\bf\textsc{(T-Imp)}}
		& \text{SPS}~e=\eimp, S=T \\
		& \subst{t}{x}{e} = \eimp \\
		& \subst{\calI(t)}{x}{S} = \subst{\calI(t)}{x}{T} \\
		& \jprop{\Gamma_1,\bindvar{x}{\tpot{\tsubset{B}{\psi}}{\phi}},\Gamma'}{\bot} & \text{[premise]} \label{eq:subst:impbot} \\
		& \jwftype{\Gamma_1,\bindvar{x}{\tpot{\tsubset{B}{\psi}}{\phi}},\Gamma'}{T} & \text{[premise]} \label{eq:subst:impwftype} \\
		& \jprop{\Gamma,\subst{\calI(t)}{x}{\Gamma'}}{\bot} & \text{[Lem.~\ref{lem:substintoref}, \eqref{eq:subst:impbot}]} \\
		& \jwftype{\Gamma,\subst{\calI(t)}{x}{\Gamma'}}{\subst{\calI(t)}{x}{T}} & \text{[Lem.~\ref{lem:substintoref}, \eqref{eq:subst:impwftype}]} \\
		& \jstyping{\Gamma,\subst{\calI(t)}{x}{\Gamma'}}{\eimp}{\subst{\calI(t)}{x}{T}} & \text{[typing]}
		\shortintertext{\bf\textsc{(T-Consume-P)}}
		& \text{SPS}~e=\econsume{c}{e_0}, c\ge 0, S=T \\
		& \text{SPS}~\Gamma'= \Gamma'',c & \text{[premise]} \\
		& \subst{t}{x}{e} = \econsume{c}{\subst{t}{x}{e_0}} \\
		& \subst{\calI(t)}{x}{S} = \subst{\calI(t)}{x}{T} \\
		& \jstyping{\Gamma_1,\bindvar{x}{\tpot{\tsubset{B}{\psi}}{\phi}},\Gamma''}{e_0}{T} & \text{[premise]} \label{eq:subst:consumepe0} \\
		& \text{ind. hyp. on \eqref{eq:subst:consumepe0}} \\
		& \jstyping{\Gamma,\subst{\calI(t)}{x}{\Gamma''}}{\subst{t}{x}{e_0}}{\subst{\calI(t)}{x}{T}} \\
		& \jstyping{\Gamma,\subst{\calI(t)}{x}{\Gamma''},c}{\econsume{c}{\subst{t}{x}{e_0}}}{\subst{\calI(t)}{x}{T}} & \text{[typing]} \\
		& \jstyping{\Gamma,\subst{\calI(t)}{x}{\Gamma',c}}{\econsume{c}{\subst{t}{x}{e_0}}}{\subst{\calI(t)}{x}{T}}
		\shortintertext{\bf\textsc{(T-Consume-N)}}
		& \text{SPS}~e=\econsume{c}{e_0},c<0,S=T \\
		& \subst{t}{x}{e} = \econsume{c}{\subst{t}{x}{e_0}} \\
		& \subst{\calI(t)}{x}{S} = \subst{\calI(t)}{x}{T} \\
		& \jstyping{\Gamma_1,\bindvar{x}{\tpot{\tsubset{B}{\psi}}{\phi}},\Gamma',-c}{e_0}{T} & \text{[premise]} \label{eq:subst:consumene0} \\
		& \text{ind. hyp. on \eqref{eq:subst:consumene0}} \\
		& \jstyping{\Gamma,\subst{\calI(t)}{x}{\Gamma'},-c}{\subst{t}{x}{e_0}}{\subst{\calI(t)}{x}{T}} \\
		& \jstyping{\Gamma,\subst{\calI(t)}{x}{\Gamma'}}{\econsume{c}{\subst{t}{x}{e_0}}}{\subst{\calI(t)}{x}{T}} & \text{[typing]}
		\shortintertext{\bf\textsc{(T-Cond)}}
		& \text{SPS}~e=\econd{a_0}{e_1}{e_2},S=T \\
		& \subst{t}{x}{e} = \econd{\subst{t}{x}{a_0}}{\subst{t}{x}{e_1}}{\subst{t}{x}{e_2}} \\
		& \subst{\calI(t)}{x}{S} = \subst{\calI(t)}{x}{T} \\
		& \jatyping{\Gamma_{1},\bindvar{x}{\tpot{\tsubset{B}{\psi}}{\phi}},\Gamma'}{a_0}{\tbool} & \text{[premise]} \label{eq:subst:conde0} \\
		& \Gamma_{1},\bindvar{x}{\tpot{\tsubset{B}{\psi}}{\phi}},\Gamma',\calI(a_0) \vdash \\
		& \quad 	e_1 \dblcolon T & \text{[premise]} \label{eq:subst:conde1} \\
		& \Gamma_{1},\bindvar{x}{\tpot{\tsubset{B}{\psi}}{\phi}},\Gamma', \neg\calI(a_0) \vdash \\
		& \quad e_2 \dblcolon T & \text{[premise]} \label{eq:subst:conde2} \\
       & \sharing(\Gamma_{1},\Gamma_{2}), \subst{\calI(t)}{x}{\Gamma'} \vdash \subst{t}{x}{a_0} : \tbool & \text{[\lemref{substatom}]} \label{eq:subst:conde0subst}  \\
		& \text{ind. hyp. on \eqref{eq:subst:conde1}} \\
		& \sharing(\Gamma_{1},\Gamma_{2}), \subst{\calI(t)}{x}{\Gamma'}, \subst{\calI(t)}{x}{\calI(a_0)} \vdash \\
		& \quad \subst{t}{x}{e_1} \dblcolon \subst{\calI(t)}{x}{T} \label{eq:subst:conde1subst} \\
		& \text{ind. hyp. on \eqref{eq:subst:conde2}} \\
		& \sharing(\Gamma_{1},\Gamma_{2}), \subst{\calI(t)}{x}{\Gamma'}, \subst{\calI(t)}{x}{\neg\calI(a_0)} \vdash \\
		& \quad \subst{t}{x}{e_2} \dblcolon \subst{\calI(t)}{x}{T} \label{eq:subst:conde2subst} \\
		& \text{typing on \eqref{eq:subst:conde0subst}, \eqref{eq:subst:conde1subst}, \eqref{eq:subst:conde2subst}} \\
		& \Gamma,\subst{\calI(t)}{x}{\Gamma'} \vdash \\
		& \quad \econd{\subst{t}{x}{e_0}}{\subst{t}{x}{e_1}}{\subst{t}{x}{e_2}} \dblcolon \subst{\calI(t)}{x}{T}
		\shortintertext{\bf\textsc{(T-MatL)}}
		& \text{SPS}~e=\ematl{a_0}{e_1}{x_h}{x_t}{e_2}=T' \\
		& \subst{t}{x}{e} = \ematl{\subst{t}{x}{a_0}}{\subst{t}{x}{e_1}}{x_h}{x_t}{\subst{t}{x}{e_2}} \\
		& \subst{\calI(t)}{x}{S} = \subst{\calI(t)}{x}{T'} \\
		& \vdash \Gamma_1,\bindvar{x}{\tpot{\tsubset{B}{\psi}}{\phi}},\Gamma' \sharing \\
		& \quad \Gamma_{11},\bindvar{x}{\tpot{\tsubset{B_1}{\psi}}{\phi_1}},\Gamma_1' \mid \\
		& \quad \Gamma_{12},\bindvar{x}{\tpot{\tsubset{B_2}{\psi}}{\phi_2}},\Gamma_2' & \text{[premise]} \\
		& \jatyping{\Gamma_{11},\bindvar{x}{\tpot{\tsubset{B_1}{\psi}}{\phi_1}},\Gamma'_1}{a_0}{\tlist{T}} & \text{[premise]} \label{eq:subst:matle0} \\
		& \jstyping{\Gamma_{12},\bindvar{x}{\tpot{\tsubset{B_2}{\psi}}{\phi_2}},\Gamma'_2,\calI(a_0)=0}{e_1}{T'} & \text{[premise]} \label{eq:subst:matle1} \\
		& \Gamma_{12},\bindvar{x}{\tpot{\tsubset{B_2}{\psi}}{\phi_2}},\Gamma'_2, \\
		& \quad x_h:T,x_t:\tlist{T},\calI(a_0)=x_t+1 \vdash {e_2} \dblcolon {T'} & \text{[premise]} \label{eq:subst:matle2} \\
		& \text{exist}~\Gamma_{21},\Gamma_{22}~\text{s.t.}~\jctxsharing{\Gamma_2}{\Gamma_{21}}{\Gamma_{22}}, \\
		& \quad \jstyping{\Gamma_{21}}{t}{\tpot{\tsubset{B_1}{\psi}}{\phi_1}}, \jstyping{\Gamma_{22}}{t}{\tpot{\tsubset{B_2}{\psi}}{\phi_2}} & \text{[Prop.~\ref{prop:typingvalsharingsplit}]} \label{eq:subst:matlsplit} \\
		& \sharing(\Gamma_{11},\Gamma_{21}), \subst{\calI(t)}{x}{\Gamma'_1} \vdash \\
		& \quad \subst{t}{x}{a_0} : \tlist{\subst{\calI(t)}{x}{T}} & \text{[\lemref{substatom}]} \label{eq:subst:matle0subst} \\
		& \text{ind. hyp. on \eqref{eq:subst:matle1},\eqref{eq:subst:matle2} with \eqref{eq:subst:matlsplit}} \\
		& \sharing(\Gamma_{12},\Gamma_{22}), \subst{\calI(t)}{x}{\Gamma'_2,\subst{\calI(t)}{x}{(\calI(a_0)=0)}} \vdash \\
		& \quad \subst{t}{x}{e_1} \dblcolon \subst{\calI(t)}{x}{T'} \label{eq:subst:matle1subst} \\
		& \sharing(\Gamma_{12},\Gamma_{22}), \subst{\calI(t)}{x}{\Gamma'_2},\\
		& \quad x_h:\subst{\calI(t)}{x}{T}, x_t:\tlist{\subst{\calI(t)}{x}{T}}, \\
		& \quad \subst{\calI(t)}{x}{(\calI(a_0)=x_t+1)} \vdash \\
		& \quad \subst{t}{x}{e_2} \dblcolon \subst{\calI(t)}{x}{T'} \label{eq:subst:matle2subst} \\
		& \jctxsharing{\Gamma}{\Gamma_1}{\Gamma_2} \implies \\
		& \quad \jctxsharing{\Gamma}{(\sharing(\Gamma_{11},\Gamma_{21}))}{(\sharing(\Gamma_{12},\Gamma_{22}))} \\
		& \Gamma,\bindvar{x}{\tpot{\tsubset{B}{\psi}}{\phi}} \vdash \Gamma' \sharing \Gamma'_1 \mid \Gamma'_2 \implies \\
		& \quad \Gamma \vdash \subst{\calI(t)}{x}{\Gamma'} \sharing \subst{\calI(t)}{x}{\Gamma'_1} \mid \subst{\calI(t)}{x}{\Gamma'_2} & \text{[Lem.~\ref{lem:substintoref}]} \\
		& \text{typing on \eqref{eq:subst:matle0subst}, \eqref{eq:subst:matle1subst}, \eqref{eq:subst:matle2subst}} \\ 
		& \Gamma,\subst{\calI(t)}{x}{\Gamma'} \vdash \\
		& \quad \ematl{\subst{t}{x}{a_0}}{\subst{t}{x}{e_1}}{x_h}{x_t}{\subst{t}{x}{e_2}} \\
		& \quad \dblcolon \subst{\calI(t)}{x}{T'}
		\shortintertext{\bf\textsc{(T-Let)}}
		& \text{SPS}~e=\elet{e_1}{y}{e_2},S=T_2 \\
		& \subst{t}{x}{e}=\elet{\subst{t}{x}{e_1}}{y}{\subst{t}{x}{e_2}} \\
		& \subst{\calI(t)}{x}{S} = \subst{\calI(t)}{x}{T_2} \\
		& \vdash \Gamma_1,\bindvar{x}{\tpot{\tsubset{B}{\psi}}{\phi}},\Gamma' \sharing \\
		& \quad \Gamma_{11},\bindvar{x}{\tpot{\tsubset{B_1}{\psi}}{\phi_1}},\Gamma_1' \mid \\
		& \quad \Gamma_{12},\bindvar{x}{\tpot{\tsubset{B_2}{\psi}}{\phi_2}},\Gamma_2' & \text{[premise]} \\
		& \jstyping{\Gamma_{11},\bindvar{x}{\tpot{\tsubset{B_1}{\psi}}{\phi_1}},\Gamma_1'}{e_1}{S_1} & \text{[premise]} \label{eq:subst:lete1} \\
		& \jstyping{\Gamma_{12},\bindvar{x}{\tpot{\tsubset{B_2}{\psi_2}}{\phi_2}},\Gamma_2',\bindvar{y}{S_1}}{e_2}{T_2} & \text{[premise]} \label{eq:subst:lete2} \\
		& \text{exist}~\Gamma_{21},\Gamma_{22}~\text{s.t.}~\jctxsharing{\Gamma_2}{\Gamma_{21}}{\Gamma_{22}}, \\
		& \quad \jstyping{\Gamma_{21}}{t}{\tpot{\tsubset{B_1}{\psi}}{\phi_1}}, \jstyping{\Gamma_{22}}{t}{\tpot{\tsubset{B_2}{\psi}}{\phi_2}} & \text{[Prop.~\ref{prop:typingvalsharingsplit}]} \label{eq:subst:letsplitval} \\
		& \text{ind. hyp. on \eqref{eq:subst:lete1} with \eqref{eq:subst:letsplitval}} \\
		& \jstyping{\sharing(\Gamma_{11},\Gamma_{21}),\subst{\calI(t)}{x}{\Gamma_1'} }{\subst{t}{x}{e_1}}{\subst{\calI(t)}{x}{S_1}} \label{eq:subst:lete1subst} \\
		& \text{ind. hyp. on \eqref{eq:subst:lete2} with \eqref{eq:subst:letsplitval}} \\
		& \sharing(\Gamma_{12},\Gamma_{22}), \subst{\calI(t)}{x}{\Gamma_2'}, \bindvar{y}{\subst{\calI(t)}{x}{S_1}} \vdash \\
		& \quad  \subst{t}{x}{e_2} \dblcolon \subst{\calI(x)}{t}{T_2} \label{eq:subst:lete2subst} \\
		& \jctxsharing{\Gamma}{\Gamma_1}{\Gamma_2} \implies \\
		& \quad \jctxsharing{\Gamma}{(\sharing(\Gamma_{11},\Gamma_{21}))}{(\sharing(\Gamma_{12},\Gamma_{22}))} \\
		& \Gamma,\bindvar{x}{\tpot{\tsubset{B}{\psi}}{\phi	}} \vdash \Gamma' \sharing \Gamma'_1 \mid \Gamma'_2 \implies \\
		& \quad \Gamma \vdash \subst{\calI(t)}{x}{\Gamma'} \sharing \subst{\calI(t)}{x}{\Gamma'_1} \mid \subst{\calI(t)}{x}{\Gamma'_2} & \text{[Lem.~\ref{lem:substintoref}]} \\
		& \text{typing on \eqref{eq:subst:lete1subst}, \eqref{eq:subst:lete2subst}} \\
		& \Gamma,\subst{\calI(t)}{x}{\Gamma'} \vdash \\
		 & \quad \elet{\subst{t}{x}{e_1}}{y}{\subst{t}{x}{e_2}} \dblcolon \subst{\calI(t)}{x}{T_2}
		\shortintertext{\bf\textsc{(T-Abs)}}
		& \text{SPS}~e=\eabs{y}{e_0},S=\tpot{(\tarrow{y}{T_y}{T})}{0} \\
		& \subst{t}{x}{e} = \eabs{y}{\subst{t}{x}{e_0}} \\
		&  \subst{\calI(t)}{x}{S} = \tpot{(\tarrow{y}{\subst{\calI(t)}{x}{T_y}}{\subst{\calI(t)}{x}{T}})}{0} \\
		& \jstyping{\Gamma_1,\bindvar{x}{\tpot{\tsubset{B}{\psi}}{\phi}},\Gamma',\bindvar{y}{T_y}}{e_0}{T} & \text{[premise]} \label{eq:subst:abse0} \\
		& \vdash {\Gamma_1,\bindvar{x}{\tpot{\tsubset{B}{\psi}}{\phi}},\Gamma'} \sharing \\
		& \quad {\Gamma_1,\bindvar{x}{\tpot{\tsubset{B}{\psi}}{\phi}},\Gamma'} \mid \\
		& \quad {\Gamma_1,\bindvar{x}{\tpot{\tsubset{B}{\psi}}{\phi}},\Gamma'} & \text{[premise]} \\
		& \jsharing{\Gamma}{\tpot{\tsubset{B}{\psi}}{\phi}}{\tpot{\tsubset{B}{\psi}}{\phi}}{\tpot{\tsubset{B}{\psi}}{\phi}} \\
		& \text{exist}~\Gamma_2'~\text{s.t.}~\jstyping{\Gamma_2'}{t}{\tpot{\tsubset{B}{\psi}}{\phi}}, \jctxsharing{\Gamma_2}{\Gamma_2}{\Gamma_2'} & \text{[Prop.~\ref{prop:typingvalzero}]} \\
		& \text{ind. hyp. on \eqref{eq:subst:abse0}} \\
		& \sharing(\Gamma_1,\Gamma_2'), \subst{\calI(t)}{x}{\Gamma'}, \bindvar{y}{\subst{\calI(t)}{x}{T_y}} \vdash \\
		& \quad \subst{t}{x}{e_0} \dblcolon \subst{\calI(t)}{x}{T} \label{eq:subst:abse0subst} \\
		& \jctxsharing{(\sharing(\Gamma_1,\Gamma_2'))}{(\sharing(\Gamma_1,\Gamma_2'))}{(\sharing(\Gamma_1,\Gamma_2'))} \\
		& \Gamma, x:\tpot{\tsubset{B}{\psi}}{\phi} \vdash \Gamma' \sharing \Gamma' \mid \Gamma' \implies \\
		& \quad \Gamma \vdash \subst{\calI(t)}{x}{\Gamma'} \sharing \subst{\calI(t)}{x}{\Gamma'} \mid \subst{\calI(t)}{x}{\Gamma'} & \text{[Lem.~\ref{lem:substintoref}]} \\
		& \vdash \sharing(\Gamma_1,\Gamma_2'),\subst{\calI(t)}{x}{\Gamma'} \sharing \\
		& \quad \sharing(\Gamma_1,\Gamma_2'),\subst{\calI(t)}{x}{\Gamma'} \mid \\
		& \quad \sharing(\Gamma_1,\Gamma_2'),\subst{\calI(t)}{x}{\Gamma'} \\
		& \text{typing on \eqref{eq:subst:abse0subst}} \\
		& \sharing(\Gamma_1,\Gamma_2') ,\subst{\calI(v)}{x}{\Gamma'} \vdash \\
		& \quad \eabs{y}{\subst{t}{x}{e_0}} \dblcolon \tarrow{y}{\subst{\calI(t)}{x}{T_y}}{\subst{\calI(t)}{x}{T}} \\
		& \Gamma,\subst{\calI(v)}{x}{\Gamma'} \vdash \\
		& \quad \eabs{y}{\subst{t}{x}{e_0}} \dblcolon \tarrow{y}{\subst{\calI(t)}{x}{T_y}}{\subst{\calI(t)}{x}{T}} & \text{[Prop.~\ref{prop:ctxrelax}]}
		\shortintertext{\bf\textsc{(T-Abs-Lin)}}
		& \text{SPS}~ e=\eabs{y}{e_0},S=\tarrowm{y}{T_y}{T}{m} \\
		& \text{SPS}~\Gamma_1,\bindvar{x}{\trefined{B}{\psi}{\phi}},\Gamma' = \\
		& \quad m \cdot (\Gamma_1'',\bindvar{x}{\trefined{B''}{\psi}{\phi''}},\Gamma''' ) \\
		& \subst{t}{x}{e} = \eabs{y}{\subst{t}{x}{e_0}} \\
		&  \subst{\calI(t)}{x}{S} = \tarrowm{y}{\subst{\calI(t)}{x}{T_y}}{\subst{\calI(t)}{x}{T}}{m} \\
		& \jstyping{\Gamma_1'',\bindvar{x}{\tpot{\tsubset{B''}{\psi}}{\phi''}},\Gamma''',\bindvar{y}{T_y}}{e_0}{T} & \text{[premise]} \label{eq:subst:absline0} \\
		& \text{exist}~\Gamma_2''~\text{s.t.}~\Gamma_2=\sharing(m\cdot\Gamma_2'',\_),~\text{and} & \text{[\propref{typingvalsharingsplit},} \\
		& \quad \jstyping{\Gamma_2''}{t}{\tpot{\tsubset{B''}{\psi}}{\phi''}} & \text{ \ref{prop:typingvalzero}]} \\
		& \text{ind. hyp. on \eqref{eq:subst:absline0}} \\
		& \sharing(\Gamma_1'',\Gamma_2''), \subst{\calI(t)}{x}{\Gamma'''}, \bindvar{y}{\subst{\calI(t)}{x}{T_y}} \vdash \\
		& \quad \subst{t}{x}{e_0} \dblcolon \subst{\calI(t)}{x}{T} \label{eq:subst:absline0subst} \\
		& \text{typing on \eqref{eq:subst:absline0subst}} \\
		& m \cdot (\sharing(\Gamma_1'',\Gamma_2'') ,\subst{\calI(v)}{x}{\Gamma'''}) \vdash \\
		& \quad \eabs{y}{\subst{t}{x}{e_0}} \dblcolon \\
		& \quad \tarrowm{y}{\subst{\calI(t)}{x}{T_y}}{\subst{\calI(t)}{x}{T}}{m} \\
		& \Gamma,\subst{\calI(v)}{x}{\Gamma'} \vdash \\
		& \quad \eabs{y}{\subst{t}{x}{e_0}} \\
		& \quad \dblcolon \tarrowm{y}{\subst{\calI(t)}{x}{T_y}}{\subst{\calI(t)}{x}{T}}{m} 
		\shortintertext{\bf\textsc{(T-Fix)}}
		& \text{SPS}~e=\efix{f}{y}{e_0},S=\tpot{R}{0},R=\tarrow{y}{T_y}{T} \\
		& \subst{t}{x}{e} = \efix{f}{y}{\subst{t}{x}{e_0}} \\
		& \subst{\calI(t)}{x}{\tpot{R}{0}} = \tpot{\subst{\calI(t)}{x}{R}}{0} \\
		& \vdash \Gamma_1,\bindvar{x}{\trefined{B}{\psi}{\phi}},\Gamma' \sharing \\
		& \quad \Gamma_1,\bindvar{x}{\trefined{B}{\psi}{\phi}},\Gamma' \mid \\
		& \quad \Gamma_1,\bindvar{x}{\trefined{B}{\psi}{\phi}},\Gamma' & \text{[premise]} \\
		& \jstyping{\Gamma_1,\bindvar{x}{\tpot{\tsubset{B}{\psi}}{\phi}},\Gamma',f:\tpot{R}{0},y:T_y}{e_0}{T} \\
		& \text{ind. hyp.} \\
		& \Gamma,\subst{\calI(t)}{x}{\Gamma'},f:\tpot{\subst{\calI(t)}{x}{R}}{0},y:T_y \vdash \\
		& \quad {\subst{\calI(t)}{x}{e_0}} \dblcolon \subst{\calI(t)}{x}{T} \\
		& \jstyping{\Gamma,\subst{\calI(t)}{x}{\Gamma'}}{\efix{f}{y}{\subst{t}{x}{e_0}}}{\tpot{\subst{\calI(t)}{x}{R}}{0}}
		\shortintertext{\bf\textsc{(T-App-SimpAtom)}}
		& \text{SPS}~e=\eapp{\hat{a}_1}{a_2},S=\subst{\calI(a_2)}{y}{T} \\
		& \subst{t}{x}{e} = \eapp{\subst{t}{x}{\hat{a}_1}}{\subst{t}{x}{a_2}} \\
		& \subst{\calI(t)}{x}{S} = \subst{\calI(t)}{x}{\subst{\calI(a_2)}{y}{T}} \\
		& \vdash \Gamma_1,\bindvar{x}{\tpot{\tsubset{B}{\psi}}{\phi}},\Gamma' \sharing \\
		& \quad \Gamma_{11},\bindvar{x}{\tpot{\tsubset{B_1}{\psi}}{\phi_1}},\Gamma_1' \mid \\
		& \quad \Gamma_{12},\bindvar{x}{\tpot{\tsubset{B_2}{\psi}}{\phi_2}},\Gamma_2' & \text{[premise]} \\
		& {\Gamma_{11},\bindvar{x}{\tpot{\tsubset{B_1}{\psi}}{\phi_1}},\Gamma_1'} \vdash {\hat{a}_1} \\
		& \quad \dblcolon {\tarrowm{y}{\trefined{B_y}{\psi_y}{\phi_y}}{T}{1}} & \text{[premise]} \label{eq:subst:appae1} \\ 
		& \jstyping{\Gamma_{12},\bindvar{x}{\tpot{\tsubset{B_2}{\psi}}{\phi_2}},\Gamma_2'}{a_2}{\trefined{B_y}{\psi_y}{\phi_y}} & \text{[premise]} \label{eq:subst:appae2} \\
		& \text{exist}~\Gamma_{21},\Gamma_{22}~\text{s.t.}~\jctxsharing{\Gamma_2}{\Gamma_{21}}{\Gamma_{22}}, \\
		& \quad \jstyping{\Gamma_{21}}{t}{\tpot{\tsubset{B_1}{\psi}}{\phi_1}}, \jstyping{\Gamma_{22}}{t}{\tpot{\tsubset{B_2}{\psi}}{\phi_2}} & \text{[Prop.~\ref{prop:typingvalsharingsplit}]} \label{eq:subst:appasplitval} \\
		& \text{ind. hyp. on \eqref{eq:subst:appae1} with \eqref{eq:subst:appasplitval}} \\
		& \sharing(\Gamma_{11},\Gamma_{21}), \subst{\calI(t)}{x}{\Gamma'_1} \vdash \subst{t}{x}{\hat{a}_1} \dblcolon \\
		& \quad  {\tarrowm{y}{\subst{\calI(t)}{x}{ \trefined{B_y}{\psi_y}{\phi_y} }}{\subst{\calI(t)}{x}{T}}{1}}\label{eq:subst:appae1subst} \\
		& \text{ind. hyp. on \eqref{eq:subst:appae2} with \eqref{eq:subst:appasplitval}} \\
		& \sharing(\Gamma_{12},\Gamma_{22}), \subst{\calI(t)}{x}{\Gamma'_2} \vdash \\
		& \quad \subst{t}{x}{a_2} \dblcolon \subst{\calI(t)}{x}{ \trefined{B_y}{\psi_y}{\phi_y} } \label{eq:subst:appae2subst} \\
		& \jctxsharing{\Gamma}{\Gamma_1}{\Gamma_2} \implies \\
		& \quad \jctxsharing{\Gamma}{(\sharing(\Gamma_{11},\Gamma_{21}))}{(\sharing(\Gamma_{12},\Gamma_{22}))} \\
		& \Gamma,\bindvar{x}{\tpot{\tsubset{B}{\psi}}{\phi	}} \vdash \Gamma' \sharing \Gamma'_1 \mid \Gamma'_2 \implies \\
		& \quad \Gamma \vdash \subst{\calI(t)}{x}{\Gamma'} \sharing \subst{\calI(t)}{x}{\Gamma'_1} \mid \subst{\calI(t)}{x}{\Gamma'_2} & \text{[Lem.~\ref{lem:substintoref}]} \\
		& \text{typing on \eqref{eq:subst:appae1subst}, \eqref{eq:subst:appae2subst}} \\
		& \Gamma,\subst{\calI(t)}{x}{\Gamma'} \vdash \\
		& \quad \eapp{\subst{t}{x}{\hat{a}_1}}{\subst{t}{x}{a_2}} \dblcolon \subst{\calI(t),\calI(a_2)}{x,y}{T}
		\shortintertext{\bf\textsc{(T-App)}}
		& \text{SPS}~e=\eapp{e_1}{e_2},S=T \\
		& \subst{t}{x}{e} = \eapp{\subst{t}{x}{\hat{a}_1}}{\subst{t}{x}{\hat{a}_2}} \\
		& \subst{\calI(t)}{x}{S} = \subst{\calI(t)}{x}{T} \\
		& \vdash \Gamma_1,\bindvar{x}{\tpot{\tsubset{B}{\psi}}{\phi}},\Gamma' \sharing \\
		& \quad \Gamma_{11},\bindvar{x}{\tpot{\tsubset{B_1}{\psi}}{\phi_1}},\Gamma_1' \mid \\
		& \quad \Gamma_{12},\bindvar{x}{\tpot{\tsubset{B_2}{\psi}}{\phi_2}},\Gamma_2' & \text{[premise]} \\
		& \jstyping{\Gamma_{11},\bindvar{x}{\tpot{\tsubset{B_1}{\psi}}{\phi_1}},\Gamma_1'}{\hat{a}_1}{{\tarrowm{y}{T_y}{T}{1}}} & \text{[premise]} \label{eq:subst:appe1} \\ 
		& \jstyping{\Gamma_{12},\bindvar{x}{\tpot{\tsubset{B_2}{\psi}}{\phi_2}},\Gamma_2'}{\hat{a}_2}{T_y} & \text{[premise]} \label{eq:subst:appe2} \\
		& \text{exist}~\Gamma_{21},\Gamma_{22}~\text{s.t.}~\jctxsharing{\Gamma_2}{\Gamma_{21}}{\Gamma_{22}}, \\
		& \quad \jstyping{\Gamma_{21}}{t}{\tpot{\tsubset{B_1}{\psi}}{\phi_1}}, \jstyping{\Gamma_{22}}{t}{\tpot{\tsubset{B_2}{\psi}}{\phi_2}} & \text{[Prop.~\ref{prop:typingvalsharingsplit}]} \label{eq:subst:appsplitval} \\
		& \text{ind. hyp. on \eqref{eq:subst:appe1} with \eqref{eq:subst:appsplitval}} \\
		& \sharing(\Gamma_{11},\Gamma_{21}), \subst{\calI(t)}{x}{\Gamma'_1} \vdash \subst{t}{x}{\hat{a}_1} \dblcolon \\
		& \quad  {\tarrowm{y}{\subst{\calI(t)}{x}{T_y}}{\subst{\calI(t)}{x}{T}}{1}} \label{eq:subst:appe1subst} \\
		& \text{ind. hyp. on \eqref{eq:subst:appe2} with \eqref{eq:subst:appsplitval}} \\
		& \sharing(\Gamma_{12},\Gamma_{22}), \subst{\calI(t)}{x}{\Gamma'_2}\vdash \\
		& \quad \subst{t}{x}{\hat{a}_2} \dblcolon \subst{\calI(t)}{x}{T_y} \label{eq:subst:appe2subst} \\
		& \jctxsharing{\Gamma}{\Gamma_1}{\Gamma_2} \implies \\
		& \quad \jctxsharing{\Gamma}{(\sharing(\Gamma_{11},\Gamma_{21}))}{(\sharing(\Gamma_{12},\Gamma_{22}))} \\
		& \Gamma,\bindvar{x}{\tpot{\tsubset{B}{\psi}}{\phi	}} \vdash \Gamma' \sharing \Gamma'_1 \mid \Gamma'_2 \implies \\
		& \quad \Gamma \vdash \subst{\calI(t)}{x}{\Gamma'} \sharing \subst{\calI(t)}{x}{\Gamma'_1} \mid \subst{\calI(t)}{x}{\Gamma'_2} & \text{[Lem.~\ref{lem:substintoref}]} \\
		& \text{typing on \eqref{eq:subst:appe1subst}, \eqref{eq:subst:appe2subst}} \\
		& \Gamma,\subst{\calI(t)}{x}{\Gamma'} \vdash \\
		& \quad \eapp{\subst{t}{x}{\hat{a}_1}}{\subst{t}{x}{\hat{a}_2}} \dblcolon \subst{\calI(t)}{x}{T}
		%
		%
		\shortintertext{\bf\textsc{(S-Gen)}}
		& \text{SPS}~e=v, S=\forall\alpha.S' \\
		& \subst{t}{x}{e} = \subst{t}{x}{v} \\
		& \subst{\calI(t)}{x}{S} = \forall\alpha. \subst{\calI(t)}{x}{S'} \\
		& \jstyping{\Gamma_1,\bindvar{x}{\trefined{B}{\psi}{\phi}},\Gamma',\alpha}{v}{S'} & \text{[premise]} \\
		& \text{ind. hyp.} \\
		& \jstyping{\Gamma,\subst{\calI(t)}{x}{\Gamma'},\alpha}{\subst{t}{x}{v}}{\subst{\calI(t)}{x}{S'}} \\
		& \jstyping{\Gamma,\subst{\calI(t)}{x}{\Gamma'}}{\subst{t}{x}{v}}{\forall\alpha.\subst{\calI(t)}{x}{S'}} & \text{[typing]}
		\shortintertext{\bf\textsc{(S-Inst)}}
		& \text{SPS}~S=\subst{\trefined{B'}{\psi'}{\phi'}}{\alpha}{S'} \\
		& \subst{\calI(t)}{x}{S} = \\
		& \quad \subst{\subst{\calI(t)}{x}{\trefined{B'}{\psi'}{\phi'}}}{\alpha}{\subst{\calI(t)}{x}{S'}} \\
		& \jstyping{\Gamma_1,\bindvar{x}{\trefined{B}{\psi}{\phi}},\Gamma'}{e}{\forall\alpha.S'} & \text{[premise]} \\
		& \text{ind. hyp.} \\
		& \jstyping{\Gamma,\subst{\calI(t)}{x}{\Gamma'}}{\subst{t}{x}{e}}{\forall\alpha. \subst{\calI(t)}{x}{S'}} \\
		& \jwftype{\Gamma_1,\bindvar{x}{\trefined{B}{\psi}{\phi}},\Gamma'}{\trefined{B'}{\psi'}{\phi'}} & \text{[premise]} \\
		& \jwftype{\Gamma,\subst{\calI(t)}{x}{\Gamma'}}{\subst{\calI(t)}{x}{\trefined{B'}{\psi'}{\phi'}}} & \text{[Lem.~\ref{lem:substintoref}]} \\
		& \Gamma,\subst{\calI(t)}{x}{\Gamma'} \vdash \subst{t}{x}{e} \dblcolon \\
		& \quad  \subst{\subst{\calI(t)}{x}{\trefined{B'}{\psi'}{\phi'}}}{\alpha}{\subst{\calI(t)}{x}{S'}} & \text{[typing]}
		\shortintertext{\bf\textsc{(S-Subtype)}}
		& \text{SPS}~S=T_2 \\
		& \subst{\calI(t)}{x}{S} = \subst{\calI(t)}{x}{T_2} \\
		& \jstyping{\Gamma_1,\bindvar{x}{\trefined{B}{\psi}{\phi}},\Gamma'}{e}{T_1} & \text{[premise]} \\
		& \text{ind. hyp.} \\
		& \jstyping{\Gamma,\subst{\calI(t)}{x}{\Gamma'}}{\subst{t}{x}{e}}{\subst{\calI(t)}{x}{T_1}} \\
		& \jsubty{\Gamma_1,\bindvar{x}{\trefined{B}{\psi}{\phi}},\Gamma'}{T_1}{T_2} & \text{[premise]} \\
		& \jsubty{\Gamma,\subst{\calI(t)}{x}{\Gamma'}}{\subst{\calI(t)}{x}{T_1}}{\subst{\calI(t)}{x}{T_2}} & \text{[Lem.~\ref{lem:substintoref}]} \\
		& \jstyping{\Gamma,\subst{\calI(t)}{x}{\Gamma'}}{\subst{t}{x}{e}}{\subst{\calI(t)}{x}{T_2}} & \text{[typing]}
		\shortintertext{\bf\textsc{(S-Transfer)}}
		& \text{SPS}~\Gamma_o = \Gamma_1',\bindvar{x}{\trefined{B}{\psi}{\phi'}},\Gamma'' \\
		& \textbf{let}~\tilde{\Gamma} = \Gamma_1,\bindvar{x}{\trefined{B}{\psi}{\phi}},\Gamma' \\
		& \jstyping{\Gamma_o}{e}{S} & \text{[premise]} \label{eq:subst:transprem} \\
		& \jprop{\tilde{\Gamma}}{\pot{\tilde{\Gamma}} = \pot{\Gamma_o}} & \text{[premise]} \label{eq:subst:transequal} \\
		& \jsharing{\Gamma}{\trefined{B}{\psi}{\phi}}{\trefined{B_0}{\psi}{\phi}}{\trefined{B}{\psi}{0}} \\
		& \text{Lem.~\ref{prop:sharingsplit}, exist $\Gamma_2'$ and $\Gamma_2''$ s.t.} \\
		& \quad \jctxsharing{\Gamma_2}{\Gamma_2'}{\Gamma_2''} \\
		& \quad \jstyping{\Gamma_2'}{t}{\trefined{B_0}{\psi}{\phi}} \\
		& \qquad \implies \jprop{\Gamma_2}{\pot{\Gamma_2'} \ge \subst{\calI(t)}{\nu}{\phi}}  & \text{[Prop.~\ref{prop:freetofree}]} \\
		& \quad \jstyping{\Gamma_2''}{t}{\trefined{B}{\psi}{0}} \\
		& \jstyping{\Gamma_2'',\subst{\calI(t)}{\nu}{\phi'}}{t}{\trefined{B}{\psi}{\phi'}} & \text{[relax]} \\
		& \text{ind. hyp. on \eqref{eq:subst:transprem}} \\
		& \sharing(\Gamma_1',\Gamma_2'',\subst{\calI(t)}{\nu}{\phi'}) ,\subst{\calI(t)}{x}{\Gamma''} \vdash \\
		& \quad \subst{t}{x}{e} \dblcolon \subst{\calI(t)}{x}{S} \label{eq:subst:transind} \\
		& \text{Lem.~\ref{lem:substintoref} on \eqref{eq:subst:transequal}} \\
		& \Gamma,\subst{\calI(t)}{x}{\Gamma'} \models \\
		& \quad \subst{\calI(t)}{x}{\pot{\tilde{\Gamma}}} = \subst{\calI(t)}{x}{\pot{\Gamma_o}} \\
		& \subst{\calI(t)}{x}{ \pot{\tilde{\Gamma}}  } = \pot{\Gamma_1} +\\
		& \quad  \subst{\calI(t)}{\nu}{\phi} + \pot{\subst{\calI(t)}{x}{\Gamma'}} & \text{[def.]} \label{eq:subst:transsubst1} \\
		& \subst{\calI(t)}{x}{\pot{\Gamma_o}} = \pot{\Gamma_1'} + \\
		& \quad \subst{\calI(t)}{\nu}{\phi'} + \pot{\subst{\calI(t)}{x}{\Gamma''}} & \text{[def.]} \label{eq:subst:transsubst2} \\
		& \pot{\Gamma,\subst{\calI(t)}{x}{\Gamma'}} = \\
		& \quad \pot{\Gamma_1}+\pot{\Gamma_2} + \pot{\subst{\calI(t)}{\nu}{\Gamma'}} = \\
		& \quad \pot{\Gamma_1'}+\pot{\Gamma_2'} + \pot{\Gamma_2''} +\pot{\subst{\calI(t)}{\nu}{\Gamma''}} + \\
		& \qquad \subst{\calI(t)}{\nu}{(\phi' - \phi)} \ge & \text{[\eqref{eq:subst:transsubst1}, \eqref{eq:subst:transsubst2}]} \\
		& \quad \pot{\Gamma_1'}+\pot{\Gamma_2''}+\pot{\subst{\calI(t)}{x}{\Gamma''}} + \\
		& \qquad \subst{\calI(t)}{x}{\phi'} = \\
		& \quad \pot{\sharing(\Gamma_1',\Gamma_2'',\subst{\calI(t)}{\nu}{\phi'} ), \subst{\calI(t)}{x}{\Gamma''}} \\
		& \text{recall \eqref{eq:subst:transind}, and then typing, relax} \\
		& \jstyping{\Gamma,\subst{\calI(t)}{x}{\Gamma'}}{\subst{\calI(t)}{x}{e}}{\subst{\calI(t)}{x}{S}}
		\shortintertext{\bf\textsc{(S-Relax)}}
		& \text{SPS}~S=\tpot{R}{\phi+\phi'} \\
		& \subst{\calI(t)}{x}{S} = \tpot{\subst{\calI(t)}{x}{R}}{\subst{\calI(t)}{x}{\phi} +\subst{\calI(t)}{x}{\phi'}} \\
		& \jstyping{\Gamma_1,\bindvar{x}{\trefined{B}{\psi}{\phi}},\Gamma'}{e}{\tpot{R}{\phi}} & \text{[premise]} \\
		& \text{ind. hyp.} \\
		& \jstyping{\Gamma,\subst{\calI(t)}{x}{\Gamma'}}{\subst{t}{x}{e}}{\tpot{\subst{\calI(t)}{x}{R}}{\subst{\calI(t)}{x}{\phi}}} \\
		& \jsort{\Gamma_1,\bindvar{x}{\trefined{B}{\psi}{\phi}},\Gamma'}{\phi'}{\bbN} & \text{[premise]} \\
		& \jsort{\Gamma,\subst{\calI(t)}{x}{\Gamma'}}{\subst{\calI(t)}{x}{\phi'}}{\bbN} & \text{[Lem.~\ref{lem:substintoref}]} \\
		& \Gamma,\subst{\calI(t)}{x}{\Gamma'},\subst{\calI(t)}{x}{\phi'} \vdash \subst{t}{x}{e} \dblcolon \\
		& \quad \tpot{\subst{\calI(t)}{x}{R}}{\subst{\calI(t)}{x}{\phi}+\subst{\calI(t)}{x}{\phi'}} & \text{[typing]}
	\end{alignat}
\end{proof}

\subsection{Preservation}

\Omit{
\begin{lemma}\label{lem:unfoldconst}
	Let $\Gamma= \overline{q}$ and $\jctxsharing{\Gamma}{\Gamma_1}{\Gamma_2}$.
	\begin{enumerate}
		\item If $\jtunfoldnil{\Gamma}{\Gamma_n}{T_\ell}$ and $\jstyping{\Gamma_1}{\enil}{T_\ell}$ , then $\emptyset$ is consistent with $\jctxtyping[\Gamma_1]{\emptyset}{\Gamma_n}$.
		
		\item If $\jtunfoldcons{\Gamma}{\Gamma_c}{T_\ell}$ and $\jstyping{\Gamma_1}{\econs{v_h}{v_t}}{T_\ell}$, then $V' \defeq \{ x_h \mapsto v_h, x_t \mapsto v_t\}$ is consistent with $\jctxtyping[\Gamma_1]{V'}{\Gamma_c}$.
	\end{enumerate}
\end{lemma}
\begin{proof}[Proof of (1)]
		\begin{alignat}{2}
			& \textbf{let}~V=\emptyset, E =\emptyset \\
			& \text{Obs.}~\jctxtyping{V}{\Gamma_1}, E=\calI_V(\Gamma_1), E \models \condc{V}{\Gamma_1} \\
			& \text{SPS}~\jtunfoldnil{\Gamma}{\Gamma_n}{T_\ell} \\
			& \Gamma_n = \psi',\phi'; T_\ell = \tpot{\tsubset{\tlist{T}}{\psi}}{\phi} & \text{[premise]} \\
			& \jsubty{\Gamma}{\tpot{\tsubset{\tlist{T}}{\psi \wedge \nu = 0}}{\phi}}{\tpot{\tsubset{\tlist{T}}{\psi'}}{\phi'}} & \text{[premise]} \label{eq:unfoldnil:subtype} \\
			& \textbf{let}~V'=\emptyset,E'=\emptyset \\
			& \text{Obs.}~\jctxtyping[\Gamma_1]{V'}{\Gamma_n}, E' = \calI_{V,V'}(\Gamma_1,\Gamma_n) \\
			& \text{inv. on \eqref{eq:unfoldnil:subtype}} \\
			& \jprop{\Gamma,\nu:\tsubset{\tlist{T}}{\psi \wedge \nu = 0}}{\psi' \wedge (\phi \ge \phi')} \label{eq:unfoldnil:invsubtype} \\
			& \mathbf{let}~V''=\{ \nu \mapsto \enil \},E'' = \{ \nu \mapsto 0 \} \\
			& E'' \models (\psi \wedge \nu = 0) \implies \psi' \wedge (\phi \ge \phi') & \text{[\eqref{eq:unfoldnil:invsubtype}]} \\
			& E' \models \subst{0}{\nu}{\psi} \implies \psi' \wedge (\subst{0}{\nu}{\phi} \ge \phi') \label{eq:unfoldnil:1} \\
			& \text{SPS}~\jstyping{\Gamma_1}{\enil}{T_\ell},~\text{then by Thm.~\ref{the:progress}} \\
			& E \models \condv{V}{\enil}{T_\ell} \wedge \potc{V}{\Gamma_1} \ge \potv{V}{\enil}{T_\ell} \\
			& E \models \subst{0}{\nu}{\psi} \wedge \potc{V}{\Gamma_1} \ge \subst{0}{\nu}{\phi} \\
			& E' \models \psi' \wedge \potc{V}{\Gamma_1} \ge \phi' & \text{[\eqref{eq:unfoldnil:1}]} \\
			& \condc{V,V'}{\Gamma_n} = \psi' & \text{[def.]} \\
			& \potc{V,V'}{\Gamma_n} = \phi' & \text{[def.]} \\
			& \text{done}
		\end{alignat}
\end{proof}
\begin{proof}[Proof of (2)]
		\begin{alignat}{2}
			& \textbf{let}~V=\emptyset,E=\emptyset \\
			& \text{Obs.}~\jctxtyping{V}{\Gamma_1}, E=\calI_V(\Gamma_1), E \models \condc{V}{\Gamma_1} \\
			& \text{SPS}~\jctxtyping{\Gamma}{\Gamma_c}{T_\ell} \\
			& \Gamma_c=\bindvar{x_h}{T},\bindvar{x_t}{\tlist{T}}, \psi',\phi';  T_\ell = \tpot{\tsubset{\tlist{T}}{\psi}}{\phi} & \text{[premise]} \\
			& \Gamma,\bindvar{x_h}{T},\bindvar{x_t}{\tlist{T}} \vdash \\
			& \quad \tpot{\tsubset{\tlist{T}}{\psi \wedge \nu = x_t+1}}{\phi} <: \tpot{\tsubset{\tlist{T}}{\psi'}}{\phi'} & \text{[premise]} \label{eq:unfoldcons:subtype} \\
			& \textbf{let}~V'=\{x_h \mapsto v_h, x_t \mapsto v_t\} \\
			& \textbf{let}~E' = \{ x_h \mapsto \calI(v_h), x_t \mapsto \calI(v_t) \} \\
			& \text{SPS}~\jstyping{\Gamma_1}{\econs{v_h}{v_t}}{T_\ell} \label{eq:unfoldcons:vht} \\
			& \text{Obs.}~\jctxtyping[\Gamma_1]{V'}{\Gamma_c},E'=\calI_{V,V'}(\Gamma_1,\Gamma_c) & \text{[inv. on \eqref{eq:unfoldcons:vht}]} \\
			& \text{inv. on \eqref{eq:unfoldcons:subtype}} \\
			& \Gamma,\bindvar{x_h}{T},\bindvar{x_t}{\tlist{T}},\bindvar{\nu}{\tsubset{\tlist{T}}{\psi \wedge \nu = x_t+1}} \models \\
			& \quad \psi' \wedge (\phi \ge \phi') \label{eq:unfoldcons:invsubtype} \\
			& \textbf{let}~V''=V'[\nu \mapsto \econs{v_h}{v_t}] \\
			& \textbf{let}~E''=E'[\nu \mapsto \calI(v_t) + 1] \\
			& E'' \models (\condv{V''}{v_h}{T} \wedge \condv{V''}{v_t}{\tlist{T}} \wedge \\
			& \quad \psi \wedge \nu = x_t+ 1) \implies (\psi' \wedge \phi \ge \phi') & \text{[\eqref{eq:unfoldcons:invsubtype}]} \\
			& E' \models (\condv{V}{v_h}{T} \wedge \condv{V}{v_t}{\tlist{T}} \wedge \\
			& \quad \subst{x_t+1}{\nu}{\psi}) \implies (\psi' \wedge \subst{x_t+1}{\nu}{\phi} \ge \phi') \label{eq:unfoldcons:1} \\
			& \text{Thm.~\ref{the:progress} on \eqref{eq:unfoldcons:vht}} \\
			& E \models \condv{V}{\econs{v_h}{v_t}}{T_\ell} \wedge \\
			& \quad \potc{V}{\Gamma_1} \ge \potv{V}{\econs{v_h}{v_t}}{T_\ell} \\
			& E \models \condv{V}{v_h}{T} \wedge \condv{V}{v_t}{\tlist{T}} \wedge  \\
			& \quad \subst{\calI(v_t)+1}{\nu}{\psi} \wedge\potc{V}{\Gamma_1} \ge  \\
			& \quad  \potv{V}{v_h}{T} + \potv{V}{v_t}{\tlist{T}} + \subst{\calI(v_t)+1}{\nu}{\phi} \\
			& E' \models \condv{V}{v_h}{T} \wedge \condv{V}{v_t}{\tlist{T}} \wedge \psi'  \wedge \\
			& \quad \potc{V}{\Gamma_1} \ge \\
			& \quad \potv{V}{v_h}{T}+\potv{V}{v_T}{\tlist{T}} + \phi' & \text{[\eqref{eq:unfoldcons:1}]} \\
			& \condc{V,V'}{\Gamma_c} = \condv{V}{v_h}{T} \wedge \condv{V}{v_t}{\tlist{T}} \wedge \psi' & \text{[def.]} \\
			& \potc{V,V'}{\Gamma_c} = \potv{V}{v_h}{T} + \potv{V}{v_t}{T} + \phi' & \text{[def.]} \\
			& \text{done}
		\end{alignat}
\end{proof}
}

\begin{proposition}\label{prop:evaldeter}
	If $\jstep{e}{e'}{p}{p'}$ and $\jstep{e}{e''}{q}{q'}$, then $e' = e''$ and $q - p = q' - p'$.
\end{proposition}
\begin{proof}
	By induction on $\jstep{e}{e'}{p}{p'}$ and then inversion on $\jstep{e}{e''}{q}{q'}$.
\end{proof}

\begin{theorem}[Preservation]\label{the:preservation}
	If $\Gamma=\overline{q}$, $\jstyping{\Gamma}{e}{S}$, $p \ge \potc{\emptyset}{\Gamma}$ and $\jstep{e}{e'}{p}{p'}$, then $\jstyping{p'}{e'}{S}$.
\end{theorem}
\begin{proof}
	By induction on $\jstyping{\Gamma}{e}{S}$:
	\begin{alignat}{2}\footnotesize
		\shortintertext{\bf\textsc{(T-Consume-P)}}
		& \text{SPS}~\Gamma=(\Gamma',c),e=\econsume{c}{e_0},c\ge 0 \\
		& \text{SPS}~S=T \\
		& \jstyping{\Gamma'}{e_0}{T} & \text{[premise]} \label{eq:pres:consumepe0} \\
		& \text{inv. on}~\jstep{e}{e'}{p}{p'} \\
		& e' = e_0, p' = p-c \ge \potc{\emptyset}{\Gamma} - c = \potc{\emptyset}{\Gamma'} \\
		& \jstyping{p'}{e_0}{T} & \text{[relax, \eqref{eq:pres:consumepe0}]}
		\shortintertext{\bf\textsc{(T-Consume-N)}}
		& \text{SPS}~e=\econsume{c}{e_0},c<0,S=T \\
		& \jstyping{\Gamma,-c}{e_0}{T} & \text{[premise]} \label{eq:pres:consumene0} \\
		& \text{inv. on}~\jstep{e}{e'}{p}{p'} \\
		& e' = e_0, p' = p -c \ge \potc{\emptyset}{\Gamma}-c \\
		& \jstyping{p'}{e_0}{T} & \text{[relax, \eqref{eq:pres:consumene0}]}
		\shortintertext{\bf\textsc{(T-Cond)}}
		& \text{SPS}~e=\econd{a_0}{e_1}{e_2},S=T \\
		& \jatyping{\Gamma}{a_0}{ \tbool } & \text{[premise]} \label{eq:pres:conde0} \\
		& \jstyping{\Gamma, \calI(a_0) }{e_1}{T} & \text{[premise]} \label{eq:pres:conde1} \\
		& \jstyping{\Gamma, \neg\calI(a_0)}{e_2}{T} & \text{[premise]} \label{eq:pres:conde2} \\
		& \text{inv. on}~\jstep{e}{e'}{p}{p'} \\
		& \textbf{case}~\jstep{e}{e_1}{p}{p} \\
		& \enskip a_0 = \etrue & \text{[premise]} \\
		& \enskip \calI(a_0) = \top \\
		& \enskip  \jprop{\Gamma}{\top} \\
		& \enskip \jstyping{\Gamma}{e_1}{T} & \text{[\lemref{substprop}, \eqref{eq:pres:conde1}]} \\
		& \enskip p \ge \potc{\emptyset}{\Gamma} & \text{[asm.]} \\
		& \enskip \jstyping{p}{e_1}{T} & \text{[relax]} \\
		& \textbf{case}~\jstep{e}{e_2}{p}{p} \\
		& \enskip a_0 = \efalse & \text{[premise]} \\
		& \enskip \text{similar to $a_0=\etrue$}
		\shortintertext{\bf\textsc{(T-MatL)}}
		& \text{SPS}~e=\ematl{a_0}{e_1}{x_h}{x_t}{e_2},S=T' \\
		& \jctxsharing{\Gamma}{\Gamma_1}{\Gamma_2} & \text{[premise]} \\
		& \quad \implies \potc{\emptyset}{\Gamma}=\potc{\emptyset}{\Gamma_1}+\potc{\emptyset}{\Gamma_2} \label{eq:pres:matlsplit} \\
		& \jatyping{\Gamma_1}{a_0}{\tlist{T}} & \text{[premise]} \label{eq:pres:matle0} \\
		& \jstyping{\Gamma_2,\calI(a_0)=0}{e_1}{T'} & \text{[premise]} \label{eq:pres:matle1} \\
		&\jstyping{\Gamma_2,x_h:T,x_t:\tlist{T},\calI(a_0)=x_t+1}{e_2}{T'} & \text{[premise]} \label{eq:pres:matle2} \\
		& \text{inv. on}~\jstep{e}{e'}{p}{p'} \\
		& \textbf{case}~\jstep{e}{e_1}{p}{p} \\
		& \enskip a_0 = \enil & \text{[premise]} \\
      & \enskip \calI(a_0) = 0 \\
      & \enskip \jprop{\Gamma_2}{\calI(a_0) = 0} \\
      & \enskip \jstyping{\Gamma_2}{e_1}{T'} & \text{[\lemref{substprop}, \eqref{eq:pres:matle1}]} \\
      & \enskip p \ge \potc{\emptyset}{\Gamma_1} + \potc{\emptyset}{\Gamma_2} & \text{[asm., \eqref{eq:pres:matlsplit}]} \\
      & \enskip  \jstyping{p}{e_1}{T'} & \text{[relax]} \\
		& \textbf{case}~\jstep{e}{\subst{v_h,v_t}{x_h,x_t}{e_2}}{p}{p} \\
		& \enskip a_0 = \econs{v_h}{v_t} & \text{[premise]} \\
          & \enskip \calI(a_0) = \calI(v_t) + 1 \\
          & \enskip \jstyping{\Gamma_{11}}{v_h}{T}, \jatyping{\Gamma_{12}}{v_t}{\tlist{T}}, \jctxsharing{\Gamma_1}{\Gamma_{11}}{\Gamma_{12}} & \text{[inv.]} \label{eq:pres:matlinv} \\
          & \enskip \sharing(\Gamma_2,\Gamma_1) ,\calI(a_0) = \calI(v_t) + 1 \vdash & \text{[Thm.~\ref{the:substitution},} \\
          & \enskip\quad \subst{v_h,v_t}{x_h,x_t}{e_2} \dblcolon T' & \text{\eqref{eq:pres:matle2}, \eqref{eq:pres:matlinv}]}  \\
& \enskip \jprop{\Gamma}{\calI(a_0)=\calI(v_t)+1} \\
& \enskip \Gamma \vdash \subst{v_h,v_t}{x_h,x_t}{e_2} \dblcolon T' & \text{[\lemref{substprop}]} \\
		& \enskip p \ge \potc{\emptyset}{\Gamma} & \text{[asm., \eqref{eq:pres:matlsplit}]} \\
		& \enskip \jstyping{p}{e'}{T'} & \text{[relax]}
		\shortintertext{\bf\textsc{(T-Let)}}
		& \text{SPS}~e=\elet{e_1}{x}{e_2},S=T_2 \\
		& \jctxsharing{\Gamma}{\Gamma_1}{\Gamma_2} \\
		& \quad \implies \potc{\emptyset}{\Gamma} = \potc{\emptyset}{\Gamma_1}+\potc{\emptyset}{\Gamma_2} & \text{[premise]} \label{eq:pres:letsplit} \\
		& \jstyping{\Gamma_1}{e_1}{S_1} & \text{[premise]} \label{eq:pres:lete1} \\
		& \jstyping{\Gamma_2,\bindvar{x}{S_1}}{e_2}{T_2} & \text{[premise]} \label{eq:pres:lete2} \\
		& \text{inv. on}~\jstep{e}{e'}{p}{p'} \\
		& \textbf{case}~\jstep{e}{\elet{e_1'}{x}{e_2}}{p}{p'} \\
		& \enskip \jstep{e_1}{e_1'}{p}{p'} & \text{[premise]} \label{eq:pres:lete1step} \\
		& \enskip p-\potc{\emptyset}{\Gamma_2} \ge \potc{\emptyset}{\Gamma_1} & \text{[asm., \eqref{eq:pres:letsplit}]} \label{eq:pres:letindhyp} \\
		& \enskip \text{Thm.~\ref{the:progress} on \eqref{eq:pres:lete1} with \eqref{eq:pres:letindhyp}} \\
		& \enskip \jstep{e_1}{e_1'}{p-\potc{\emptyset}{\Gamma_2}}{p'-\potc{\emptyset}{\Gamma_2}} & \text{[Prop.~\ref{prop:evaldeter}, \eqref{eq:pres:lete1step}]} \label{eq:pres:lete1step2} \\
		& \enskip \text{ind. hyp. on \eqref{eq:pres:lete1} with \eqref{eq:pres:lete1step2}, \eqref{eq:pres:letindhyp}} \\
		& \enskip \jstyping{p'-\potc{\emptyset}{\Gamma_2}}{e_1'}{S_1} \\
		& \enskip \jstyping{\sharing(p'-\potc{\emptyset}{\Gamma_2},\Gamma_2)}{\elet{e_1'}{x}{e_2}}{T_2} & \text{[typing]} \\
		& \enskip \jstyping{p'}{e'}{T_2} & \text{[transfer]} \\
		& \textbf{case}~\jstep{e}{\subst{e_1}{x}{e_2}}{p}{p} \\
		& \enskip \jval{e_1} & \text{[premise]} \label{eq:pres:lete1val} \\
		& \enskip \jstyping{\sharing(\Gamma_1,\Gamma_2)}{\subst{e_1}{x}{e_2}}{T_2} & \text{[Thm.~\ref{the:substitution}, \eqref{eq:pres:lete2}]} \\
		& \enskip \jstyping{\potc{\emptyset}{\Gamma_1}+\potc{\emptyset}{\Gamma_2}}{e'}{T_2} & \text{[transfer]} \\
		& \enskip \jstyping{p}{e'}{T_2 } & \text{[relax]}
		\shortintertext{\bf\textsc{(T-App-SimpAtom)}}
		& \text{SPS}~e=\eapp{\hat{a}_1}{a_2},S=T \\
		& \jctxsharing{\Gamma}{\Gamma_1}{\Gamma_2} \\
		& \quad \implies \potc{\emptyset}{\Gamma} = \potc{\emptyset}{\Gamma_1}+\potc{\emptyset}{\Gamma_2} & \text{[premise]} \label{eq:pres:appasplit} \\
		& \jstyping{\Gamma_1}{\hat{a}_1}{\tarrowm{x}{\trefined{B_x}{\psi_x}{\phi_x}}{T}{1}} & \text{[premise]} \label{eq:pres:appae1} \\
		& \jstyping{\Gamma_2}{a_2}{\trefined{B_x}{\psi_x}{\phi_x}} & \text{[premise]} \label{eq:pres:appae2} \\
		& \text{inv. on}~\jstep{e}{e'}{p}{p'} \\
		& \textbf{case}~\jstep{e}{\subst{a_2}{x}{e_0}}{p}{p} \\
		& \enskip \hat{a}_1=\eabs{x}{e_0}, \jval{a_2} & \text{[premise]} \\
		& \enskip \text{inv. on}~\eqref{eq:pres:appae1} \\
		& \enskip \jstyping{\Gamma_1, \bindvar{x}{\trefined{B_x}{\psi_x}{\phi_x}}}{e_0}{T} \Omit{\potc{\emptyset}{\Gamma_1} \ge \phi_1} \label{eq:pres:appainvlambda} \\
		& \enskip \jstyping{\Gamma}{\subst{a_2}{x}{e_0}}{\subst{\calI(a_2)}{x}{T}} & \text{[Thm.~\ref{the:substitution}, \eqref{eq:pres:appainvlambda}]} \\
        & \enskip p \ge \potc{\emptyset}{\Gamma} & \text{[asm.]} \\
		& \enskip \jstyping{p}{e'}{T} & \text{[relax]} \\
		& \textbf{case}~\jstep{e}{\subst{e_1,a_2}{f,x}{e_0}}{p}{p} \\
		& \enskip e_1=\efix{f}{x}{e_0}, \jval{a_2} & \text{[premise]} \\
		& \enskip \text{similar to $e_1=\eabs{x}{e_0}$}
		\shortintertext{\bf\textsc{(T-App)}}
		& \text{SPS}~e=\eapp{\hat{a}_1}{\hat{a}_2},S=T \\
		& \jctxsharing{\Gamma}{\Gamma_1}{\Gamma_2} \\
		& \quad \implies \potc{\emptyset}{\Gamma} = \potc{\emptyset}{\Gamma_1}+\potc{\emptyset}{\Gamma_2} & \text{[premise]} \label{eq:pres:appsplit} \\
		& \jstyping{\Gamma_1}{\hat{a}_1}{\tarrowm{x}{T_x}{T}{1}} & \text{[premise]} \label{eq:pres:appe1} \\
		& \jstyping{\Gamma_2}{\hat{a}_2}{T_x} & \text{[premise]} \label{eq:pres:appe2} \\
		& \text{inv. on}~\jstep{e}{e'}{p}{p'} \\
		& \textbf{case}~\jstep{e}{\subst{\hat{a}_2}{x}{e_0}}{p}{p} \\
		& \enskip \hat{a}_1=\eabs{x}{e_0}, \jval{\hat{a}_2} & \text{[premise]} \\
		& \enskip \text{inv. on}~\eqref{eq:pres:appe1} \\
		& \enskip \jstyping{\Gamma_1, \bindvar{x}{T_x}}{e_0}{T} \Omit{\potc{\emptyset}{\Gamma_1} \ge \phi_1} \label{eq:pres:appinvlambda} \\
		& \enskip \jstyping{\Gamma}{\subst{\hat{a}_2}{x}{e_0}}{T} & \text{[Thm.~\ref{the:substitution}, \eqref{eq:pres:appinvlambda}]} \\
        & \enskip p \ge \potc{\emptyset}{\Gamma} & \text{[asm.]} \\
		& \enskip \jstyping{p}{e'}{T} & \text{[relax]} \\
		& \textbf{case}~\jstep{e}{\subst{e_1,\hat{a}_2}{f,x}{e_0}}{p}{p} \\
		& \enskip e_1=\efix{f}{x}{e_0}, \jval{\hat{a}_2} & \text{[premise]} \\
		& \enskip \text{similar to $e_1=\eabs{x}{e_0}$}
		%
		%
		\shortintertext{\bf\textsc{(S-Inst)}}
		& \text{SPS}~S=\subst{\tpot{\tsubset{B}{\psi}}{\phi}}{\alpha}{S'} \\
		& \jstyping{\Gamma}{e}{\forall\alpha.S' } & \text{[premise]} \label{eq:pres:instprem} \\
		& \text{ind. hyp. on \eqref{eq:pres:instprem}} \\
		& \jstyping{p'}{e'}{\forall\alpha.S'} \\
		& \jstyping{p'}{e'}{\subst{\tpot{\tsubset{B}{\psi}}{\phi}}{\alpha}{S'}} & \text{[typing]}
		\shortintertext{\bf\textsc{(S-Subtype)}}
		& \text{SPS}~S=T_2 \\
		& \jstyping{\Gamma}{e}{T_1} & \text{[premise]} \label{eq:pres:subtypeprem} \\
		& \jsubty{\Gamma}{T_1}{T_2} & \text{[premise]} \label{eq:pres:subtyperel} \\
		& \text{ind. hyp. on \eqref{eq:pres:subtypeprem}} \\
		& \jstyping{p'}{e'}{T_1} \\
		& \jstyping{p'}{e'}{T_2} & \text{[typing]}
		\shortintertext{\bf\textsc{(S-Transfer)}}
		& \jstyping{\Gamma'}{e}{S}, \jprop{\Gamma}{\pot{\Gamma}=\pot{\Gamma'}} & \text{[premise]} \label{eq:pres:transprem} \\
		& \Gamma' = \overline{q'} \wedge \potc{\emptyset}{\Gamma}=\potc{\emptyset}{\Gamma'} \\
		& p \ge \potc{\emptyset}{\Gamma'} \label{eq:pres:transindhyp} \\
		& \text{ind. hyp. on \eqref{eq:pres:transprem} with \eqref{eq:pres:transindhyp}} \\
		& \jstyping{p'}{e'}{S}
		\shortintertext{\bf\textsc{(S-Relax)}}
		& \text{SPS}~\Gamma=(\Gamma',\phi'), S=\tpot{R}{\phi+\phi'} \\
		& \jstyping{\Gamma'}{e}{\tpot{R}{\phi}}  & \text{[premise]} \label{eq:pres:relaxprem} \\
		& p -\phi' \ge \potc{\emptyset}{\Gamma'} & \text{[asm.]} \label{eq:pres:relaxindhyp} \\
		& \text{Thm.~\ref{the:progress} on \eqref{eq:pres:relaxprem} with \eqref{eq:pres:relaxindhyp}} \\
		& \jstep{e}{e'}{p-\phi'}{p'-\phi'} & \text{[Prop.~\ref{prop:evaldeter}, asm.]} \label{eq:pres:relaxstep} \\
		& \text{ind. hyp. on \eqref{eq:pres:relaxprem} with \eqref{eq:pres:relaxstep}, \eqref{eq:pres:relaxindhyp}} \\
		& \jstyping{p'-\phi'}{e'}{\tpot{R}{\phi}} \\
		& \jstyping{p'-\phi',\phi'}{e'}{\tpot{R}{\phi+\phi'}} & \text{[relax]} \\
		& \jstyping{p'}{e'}{\tpot{R}{\phi+\phi'}} & \text{[transfer]}
	\end{alignat}
\end{proof}

\section{Synthesis Rules}
\label{sec:appendixproofs:synth}

\subsection{Program Templates}

\begin{eqnarray*}
  D & \Coloneqq & \cdot \mid D; x \gets e \\ 
	\hole{e} & \Coloneqq & e \mid \ehole \mid \eapp{x}{\ehole} 
  \mid \econd{x}{\ehole}{\ehole} \mid \ematl{x}{\ehole}{x_h}{x_t}{\ehole}
  \mid \elets{D}{\hole{e}} \\
\end{eqnarray*}

\subsection{Types}

\begin{eqnarray*}
  T & \Coloneqq & \tpot{R}{\phi} \mid \tbot
\end{eqnarray*}

\subsection{Type Wellformedness: $\jwftype{\Gamma}{T}$}

\begin{mathpar}\footnotesize
	\inferrule[Wf-TUnk]
	{ }
	{ \jwftype{\Gamma}{\tbot} }
\end{mathpar}

\subsection{Restricted denotation \jproppara{\Gamma}{\psi}{X}}

\begin{align*}
	\condcpara{V}{X}{\Gamma,x : \tbot } & = 
     \begin{cases}
     \bot \quad\text{if}\ x\in X\\
     \condcpara{V}{X}{\Gamma} \quad\text{otherwise}
     \end{cases}\\
  \condcpara{V}{X}{\Gamma} &= \condc{V}{\Gamma} \quad\text{otherwise}
\end{align*}

\[
\jproppara{\Gamma}{\theta}{X} \defeq \forall V \in \interp{\Gamma}. \condcpara{V}{X}{\Gamma} \implies \interp{\theta}^\Gamma_\bbB(V)
\]

\subsection{Subtyping: $\jsubty{\Gamma}{T}{T}$}

\begin{mathpar}\footnotesize
  \inferrule[Sub-TBot]
	{  }
	{ \jsubty{\Gamma}{\tbot}{T} }
  \and
	\inferrule[Sub-Refined]
	{ \jsubty{\Gamma}{B_1}{B_2} \\ \jproppara{\Gamma,\nu:B_1}{\psi_1 \implies \psi_2}{\varsof{\psi_1 \implies \psi_2}} }
	{ \jsubty{\Gamma}{\tsubset{B_1}{\psi_1}}{\tsubset{B_2}{\psi_2}} }  
\end{mathpar}

\subsection{Atomic synthesis: $\jafill{\Gamma}{\hole{e}}{T}{\elets{D}{a}}$}

\begin{mathpar}\footnotesize
	\inferrule[ASyn-Var]
	{ \jstyping{\Gamma}{x}{T} }
	{ \jasynth{\Gamma}{T}{\elets{\cdot}{x}} }
	\and  
	\inferrule[ASyn-True]
	{ \jstyping{\Gamma}{\etrue}{T} }
	{ \jasynth{\Gamma}{T}{\elets{\cdot}{\etrue}} }
	\and
	\inferrule[ASyn-False]
	{ \jstyping{\Gamma}{\efalse}{T} }
	{ \jasynth{\Gamma}{T}{\elets{\cdot}{\efalse}} }
	\and
	\inferrule[ASyn-Nil]
	{ \jstyping{\Gamma}{\enil}{T} }
	{ \jasynth{\Gamma}{T}{\elets{\cdot}{\enil}} }
	\and    
	\inferrule[ASyn-Cons]{ \jasynth{\Gamma}{T}{\elets{D_h}{a_h}} \\ \jafill{\Gamma}{\elets{D_h}{\econs{a_h}{\ehole}}}{\trefined{\tlist{T}}{\psi}{\phi}}{\elets{D}{a}}  }{ \jasynth{\Gamma}{\trefined{\tlist{T}}{\psi}{\phi}}{\elets{D}{a}} }
	\and  
	\inferrule[ASyn-App]
	{ 
    \jasynth{\Gamma}{\tarrowm{\_}{\tbot}{T}{1}}{\elets{D_1}{x}} \\    
    \jafill{\Gamma}{\elets{D_1}{\eapp{x}{\ehole}}}{T}{\elets{D}{x'}} \\
  }
	{ \jasynth{\Gamma}{T}{\elets{D}{x'}} }
	\and
	\inferrule[AFill-Cons]{ \jatyping{\Gamma}{a_h}{T} \\ \jasynth{\Gamma}{\trefined{\tlist{T}}{\psi'}{\phi'}}{\elets{D}{a_t}} \\ \jstyping{\Gamma}{\fold{\elets{D}{\econs{a_h}{a_t}}}}{\trefined{\tlist{T}}{\psi}{\phi}} \\ \psi' = \subst{\nu+1}{\nu}{\psi} \\ \phi' = \subst{\nu+1}{\nu}{\phi} }{ \jafill{\Gamma}{\econs{a_h}{\ehole}}{\trefined{\tlist{T}}{\psi}{\phi}}{\elets{D}{\econs{a_h}{a_t}}} }

  \and    
	\inferrule[AFill-App-SimpAtom]
	{  \jstyping{\Gamma}{x}{  \tarrowm{y}{T_1}{T'}{1}  } \\
     T_1~\textsf{scalar}  \\
    \jafill{\Gamma}{\ehole}{T_1}{\elets{D}{a}} \\  
    \jstyping{\Gamma}{\fold{\elets{D}{\eapp{x}{a}}}}{T}   }
	{ \jafill{\Gamma,1}{\eapp{x}{\ehole}}{T}{\elets{D; x' \gets \econsume{1}{\eapp{x}{a}}}{x'} }}
	\and    
	\inferrule[AFill-App]
	{ \jstyping{\Gamma}{x}{ \tarrowm{\_}{T_1}{T}{1} } \\  
    T_1~\text{non-scalar} \\
    \jsynth{\Gamma}{T_1}{\hat{a}} \\ \jstyping{\Gamma}{\eapp{x}{\hat{a}}}{T} }
	{ \jafill{\Gamma,1}{\eapp{x}{\ehole}}{T}{\elets{x' \gets \econsume{1}{\eapp{x}{\hat{a}}}}{x'}} }
	\and  
	\inferrule[AFill-Let]
	{ \jctxsharing{\Gamma}{\Gamma_1}{\Gamma_2} \\
    \jstyping{\Gamma_1}{e_1}{S_1} \\ 
    \jafill{\Gamma_2,x:S_1}{\elets{D}{\hole{e}_2}}{T}{\elets{D_2}{a}}  }
	{ \jafill{\Gamma}{\elets{x \gets e_1; D}{\hole{e}_2}}{T}{\elets{x \gets e_1; D_2}{a}} }
	\and  
	\inferrule[AFill-Let-Emp]
	{ \jafill{\Gamma}{\hole{e}}{T}{\elets{D}{a}}  }
	{ \jafill{\Gamma}{\elets{\cdot}{\hole{e}}}{T}{\elets{D}{a}} }  
	\and  
	\inferrule[AFill-Transfer]
	{\jprop{\Gamma}{ \pot{\Gamma} = \pot{\Gamma'}} \\
    \jafill{\Gamma'}{\hole{e}}{T}{\elets{D}{a}}  }
	{ \jafill{\Gamma}{\hole{e}}{T}{\elets{D}{a}} }    
\end{mathpar}

\subsection{Synthesis: $\jfill{\Gamma}{\hole{e}}{S}{e}$}

\begin{mathpar}\footnotesize
	\inferrule[Syn-Imp]
	{ \jprop{\Gamma}{\bot} }
	{ \jsynth{\Gamma}{T}{\eimp} }    
  \and  
	\inferrule[Syn-Cond]
	{ 
    \jasynth{\Gamma}{\tbool}{\elets{D}{x}} \\ 
    \jfill{\Gamma}{\elets{D}{\econd{x}{\ehole}{\ehole}}}{T}{e} }
  { \jsynth{\Gamma}{T}{e}}
	\and
	\inferrule[Syn-MatL]
	{ 
    \jwftype{\Gamma}{T}\\  
    \jasynth{\Gamma}{\tlist{T}}{\elets{D}{x}} \\
	  \jfill{\Gamma}{\elets{D}{\ematl{x}{\ehole}{x_h}{x_t}{\ehole}}}{T}{e} }
	{ \jsynth{\Gamma}{T}{e} }    
	\and
	\inferrule[Syn-Fix]
  { \jsynth{\Gamma,f:\p{\tarrow{x}{T_x}{T}},x:T_x}{T}{e} \\ \jctxsharing{\Gamma}{\Gamma}{\Gamma} \\ \jwftype{\Gamma}{( \tarrow{x}{T_x}{T} ) } }
  { \jsynth{\Gamma}{\p{\tarrow{x}{T_x}{T}}}{\efix{f}{x}{e}} }
	\and
	\inferrule[Syn-Abs-Lin]
	{ \jsynth{\Gamma,x:T_x}{T}{e} \\ \jwftype{\Gamma}{T_x} }
	{ \jsynth{\Gamma}{\p{\tarrowm{x}{T_x}{T}{1}}}{\eabs{x}{e}} }
  \and
	\inferrule[Syn-Gen]
	{ \jsharing{\Gamma,\alpha}{S}{S}{S} \\ \jsynth{\Gamma,\alpha}{S}{e} \\ \jval{e} }
	{ \jsynth{\Gamma}{\forall\alpha.S}{e} }
	\and    
	\inferrule[Fill-Cond]
	{ \jatyping{\Gamma}{x}{\tbool} \\
	\jsynth{\Gamma, x  }{T}{e_1} \\ 
    \jsynth{\Gamma, \neg x  }{T}{e_2}  }
  { \jfill{\Gamma}{\econd{x}{\ehole}{\ehole}}{T}{\econd{x}{e_1}{e_2}}}	  
	\and  
	\inferrule[Fill-MatL]
	{  \jctxsharing{\Gamma}{\Gamma_1}{\Gamma_2}  \\ \jatyping{\Gamma_1}{x}{\tlist{T}} \\
    \jsynth{\Gamma_2, x = 0}{T}{e_1} \\ 
    \jsynth{\Gamma_2, x_h: T, x_t:\tlist{T}, x = x_t+1}{T}{e_2}  }
  { \jfill{\Gamma}{\ematl{x}{\ehole}{x_h}{x_t}{\ehole}}{T}{\ematl{x}{e_1}{x_h}{x_t}{e_2}}}	  
  \and
	\inferrule[Fill-Let]
	{ \jctxsharing{\Gamma}{\Gamma_1}{\Gamma_2} \\
    \jstyping{\Gamma_1}{e_1}{S_1} \\ 
    \jfill{\Gamma_2,x:S_1}{\elets{D}{\hole{e}_2}}{T}{e_2} \\ \jwftype{\Gamma}{T}  }
	{ \jfill{\Gamma}{\elets{D; x \gets e_1}{\hole{e}_2}}{T}{\elet{e_1}{x}{e_2}} }
	\and  
	\inferrule[Fill-Let-Emp]
	{ \jfill{\Gamma}{\hole{e}}{T}{e}  }
	{ \jfill{\Gamma}{\elets{\cdot}{\hole{e}}}{T}{e} }  
	\and  
	\inferrule[Syn-Atom]
	{ \jasynth{\Gamma}{T}{\elets{D}{a}}  }
	{ \jsynth{\Gamma}{T}{\fold{\elets{D}{a}}} }    
\end{mathpar}

\section{Soundness of Synthesis}
\label{sec:appendixproofs:relativesound}

\begin{proposition}\label{prop:insert-tick}
  If $\jstyping{\Gamma}{\fold{\elets{D}{e}}}{T}$, then $\jstyping{\Gamma,c}{\fold{\elets{D; x' \gets \econsume{c}{ e}}{x'}}}{T}$ for $c \ge 0$.
\end{proposition}
\begin{proof}
  By induction on the length of $D$.
  \begin{itemize}
    \item $D=\cdot$: We have $\fold{\elets{\cdot}{e}} = e$ and $\fold{\elets{x' \gets \econsume{c}{e}}{x'}} = \elet{\econsume{c}{e}}{x'}{x'}$.
    By $\jstyping{\Gamma}{e}{T}$ we know $\jstyping{\Gamma,c}{\econsume{c}{e}}{T}$ by \textsc{(T-Consume-P)}.
    Therefore by \textsc{(T-Let)} we derive $\jstyping{\Gamma,c}{\elet{\econsume{c}{e}}{x'}{x'}}{T}$.
    
    \item $D = x_1 \gets e_1; D'$:
    We have $\fold{\elets{x_1 \gets e_1;D'}{e}} = \elet{e_1}{x_1}{\fold{\elets{D'}{e}}}$.
    By inversion on
    \[
    \jstyping{\Gamma}{\elet{e_1}{x_1}{\fold{\elets{D'}{e}}}}{T}
    \] we know there exist $\Gamma_1,\Gamma_2,S_1$ such that $\jctxsharing{\Gamma}{\Gamma_1}{\Gamma_2}$, $\jstyping{\Gamma_1}{e_1}{S_1}$ and $\jstyping{\Gamma_2,\bindvar{x_1}{S_1}}{\fold{\elets{D'}{e}}}{T}$.
    Thus by I.H. we have $\jstyping{\Gamma_2,\bindvar{x_1}{S_1},c}{\fold{\elets{D';x'\gets\econsume{c}{e}}{x'}}}{T}$.
    Again by \textsc{(T-Let)} and \textsc{(T-Transfer)} we derive $\jstyping{\Gamma,c}{\elet{e_1}{x_1}{\fold{\elets{D';x' \gets \econsume{c}{e}}{x'}}}}{T}$, i.e., $\jstyping{\Gamma,c}{\fold{\elets{D;x'\gets\econsume{c}{e}}{x'}}}{T}$.
  \end{itemize}
\end{proof}

\begin{lemma}\label{lem:atomic-synth}
	If \jafill{\Gamma}{\hole{e}}{T}{\elets{D}{a}}, 
  then \jstyping{\Gamma}{\fold{\elets{D}{a}}}{T}.
\end{lemma}
\begin{proof}
By induction on the derivation of \jafill{\Gamma}{\hole{e}}{T}{\elets{D}{a}}.
\begin{itemize}
  \item
  $\footnotesize
  \Rule{ASyn-Var}{ \jstyping{\Gamma}{x}{T} }{ \jasynth{\Gamma}{T}{\elets{\cdot}{x}} }
  $
  
  We have $\fold{\elets{\cdot}{x}} = x$ and thus conclude $\jstyping{\Gamma}{x}{T}$ by the premise.
  
  Cases \textsc{(ASyn-True)}, \textsc{(ASyn-False)}, \textsc{(ASyn-Nil)} are similar to this case.
  
  \item
 $\footnotesize
 \Rule{ASyn-Cons}{ \jasynth{\Gamma}{T}{\elets{D_h}{a_h}} \\ \jafill{\Gamma}{\elets{D_h}{\econs{a_h}{\ehole}}}{\trefined{\tlist{T}}{\psi}{\phi}}{\elets{D}{a}}  }{ \jasynth{\Gamma}{\trefined{\tlist{T}}{\psi}{\phi}}{\elets{D}{a}} }
 $  
 
 By I.H. on the second premise.
  
  \item
  $\footnotesize
  \Rule{ASyn-App}{ \jasynth{\Gamma}{\tarrowm{\_}{?}{T}{1}}{\elets{D_1}{x }}  \\ \jafill{\Gamma}{\elets{D_1}{\eapp{x}{\ehole}}}{T}{\elets{D}{x'}} }{ \jasynth{\Gamma}{T}{\elets{D}{x'}} }
  $
  
  By I.H. on the second premise.
  
  \item
  $\footnotesize
  \Rule{AFill-Cons}{ \jatyping{\Gamma}{a_h}{T} \\ \jasynth{\Gamma}{\trefined{\tlist{T}}{\psi'}{\phi'}}{\elets{D}{a_t}} \\ \jstyping{\Gamma}{\fold{\elets{D}{\econs{a_h}{a_t}}}}{\trefined{\tlist{T}}{\psi}{\phi}} \\ \psi' = \subst{\nu+1}{\nu}{\psi} \\ \phi' = \subst{\nu+1}{\nu}{\phi} }{ \jafill{\Gamma}{\econs{a_h}{\ehole}}{\trefined{\tlist{T}}{\psi}{\phi}}{\elets{D}{\econs{a_h}{a_t}}} }
  $
  
  By the third premise.

  \item
  $\footnotesize
  \Rule{AFill-App-SimpAtom}{ \jstyping{\Gamma}{x}{\tarrowm{y}{T_1}{T'}{1}} \\ \tscalar{T_1} \\ \jasynth{\Gamma}{T_1}{\elets{D}{a}} \\ \jstyping{\Gamma}{\fold{\elets{D}{\eapp{x}{a}}}}{T} }{ \jafill{\Gamma,1}{\eapp{x}{\ehole}}{T}{\elets{D; x' \gets \econsume{1}{\eapp{x}{a}} }{x'}} }
  $
    
  Appeal to \propref{insert-tick}.
  
  \item
  $\footnotesize
  \Rule{AFill-App}{ \jstyping{\Gamma}{x}{\tarrowm{\_}{T_1}{T}{1}} \\ T_1~\text{non-scalar} \\ \jsynth{\Gamma}{T_1}{\hat{a}} \\ \jstyping{\Gamma}{\eapp{x}{\hat{a}}}{T} }{ \jafill{\Gamma,1}{\eapp{x}{\ehole}}{T}{\elets{x' \gets \econsume{1}{\eapp{x}{\hat{a}}}}{x'}}}
  $ 
  
  Appeal to \propref{insert-tick}.
  
  \item
  $\footnotesize
  \Rule{AFill-Let}{ \jctxsharing{\Gamma}{\Gamma_1}{\Gamma_2} \\ \jstyping{\Gamma_1}{e_1}{S_1} \\ \jafill{\Gamma_2,\bindvar{x}{S_1}}{\elets{D}{\hole{e_2}}}{T}{\elets{D_2}{a}} }{ \jafill{\Gamma}{\elets{x \gets e_1;D}{\hole{e_2}}}{T}{\elets{ x\gets e_1;D_2}{a}} }
  $
  
  By I.H. on the third premise, we have \[
  \jstyping{\Gamma_2,\bindvar{x}{S_1}}{\fold{\elets{D_2}{a}}}{T}.\]
  
  Since $\fold{\elets{x \gets e_1;D_2}{a}} = \elet{e_1}{x}{\fold{\elets{D_2}{a}}}$, we conclude by \textsc{(T-Let)}.
  
  \item
  $\footnotesize
  \Rule{AFill-Let-Emp}{ \jafill{\Gamma}{\hole{e}}{T}{\elets{D}{a}} }{ \jafill{\Gamma}{\elets{\cdot}{\hole{e}}}{T}{\elets{D}{a}} }
  $
  
  By I.H. on the premise.
  
  \item
  $\footnotesize
  \Rule{AFill-Transfer}{ \jprop{\Gamma}{  \pot{\Gamma} = \pot{\Gamma'} } \\ \jafill{\Gamma'}{\hole{e}}{T}{\elets{D}{a}} }{ \jafill{\Gamma}{\hole{e}}{T}{\elets{D}{a}} }
  $
  
  By I.H. on the second premise, we have $\jstyping{\Gamma'}{\fold{\elets{D}{a}}}{T}$.
  
  Thus we derive $\jstyping{\Gamma}{\fold{\elets{D}{a}}}{T}$ by \textsc{(S-Transfer)}.
  
\end{itemize}
\end{proof}

\begin{lemma}\label{lem:synth}
If \jfill{\Gamma}{\hole{e}}{S}{e}, then \jstyping{\Gamma}{e}{S}.
\end{lemma}
\begin{proof}
By induction on the derivation of \jfill{\Gamma}{\hole{e}}{S}{e}.
\begin{itemize}
  \item
  $\footnotesize
  \Rule{Syn-Imp}{ \jprop{\Gamma}{\bot} }{ \jsynth{\Gamma}{T}{\eimp} }
  $
  
  We derive $\jstyping{\Gamma}{\eimp}{T}$ by \textsc{(T-Imp)}.
  
  \item
  $\footnotesize
  \Rule{Syn-Cond}{ \jasynth{\Gamma}{\tbool}{\elets{D}{x}} \\ \jfill{\Gamma}{\elets{D}{\econd{x}{\ehole}{\ehole}}}{T}{e} }{ \jsynth{\Gamma}{T}{e} }
  $
  
  By I.H. on the second premise.
  
  \item
  $\footnotesize
  \Rule{Syn-MatL}{ \jwftype{\Gamma}{T} \\ \jasynth{\Gamma}{\tlist{T}}{\elets{D}{x}} \\ \jfill{\Gamma}{\elets{D}{\ematl{x}{\ehole}{x_h}{x_t}{\ehole}}}{T}{e} }{ \jsynth{\Gamma}{T}{e} }
  $
  
  By I.H. on the third premise.
  
  \item
  $\footnotesize
  \Rule{Syn-Fix}{ \jsynth{\Gamma,\bindvar{f}{(\tarrow{x}{T_x}{T} )},\bindvar{x}{T_x}}{T}{e} \\ \jctxsharing{\Gamma}{\Gamma}{\Gamma} \\ \jwftype{\Gamma}{\p{\tarrow{x}{T_x}{T} } } }{ \jsynth{\Gamma}{(\tarrow{x}{T_x}{T} )}{\efix{f}{x}{e}} }
  $
    
  By I.H. on the first premise, we have $\jstyping{\Gamma,\bindvar{f}{(\tarrow{x}{T_x}{T} )},\bindvar{x}{T_x}}{e}{T}$.
  
  Thus we derive $\jstyping{\Gamma}{\efix{f}{x}{e}}{(\tarrow{x}{T_x}{T} )}$.
  
  \item
  $\footnotesize
  \Rule{Syn-Abs-Lin}{ \jsynth{\Gamma,\bindvar{x}{T_x}}{T}{e} \\ \jwftype{\Gamma}{T_x} }{ \jsynth{\Gamma}{\p{\tarrowm{x}{T_x}{T}{1} }}{\eabs{x}{e}} }
  $
  
  By I.H. on the first premise, we have $\jstyping{\Gamma,\bindvar{x}{T_x}}{e}{T}$.
  
  Thus we derive $\jstyping{\Gamma}{\eabs{x}{e}}{\tarrowm{x}{T_x}{T}{1}}$ by \textsc{(T-Abs-Lin)}.
  
  \item
  $\footnotesize
  \Rule{Syn-Gen}{ \jsharing{\Gamma,\alpha}{S}{S}{S} \\ \jsynth{\Gamma,\alpha}{S}{e} \\ \jval{e} }{ \jsynth{\Gamma}{\forall\alpha.S}{e} }
  $
  
  By I.H. on the second premise, we have $\jstyping{\Gamma,\alpha}{e}{S}$.
  
  Thus we derive $\jstyping{\Gamma}{e}{\forall\alpha.S}$ by \textsc{(S-Gen)}.
  
  \item
  $\footnotesize
  \Rule{Fill-Cond}{ \jatyping{\Gamma}{x}{\tbool} \\ \jsynth{\Gamma,x}{T}{e_1} \\ \jsynth{\Gamma,\neg x}{T}{e_2} }{ \jfill{\Gamma}{\econd{x}{\ehole}{\ehole}}{T}{\econd{x}{e_1}{e_2}} }
  $
  
  By I.H. on the second premise, we have $\jstyping{\Gamma,x}{e_1}{T}$.
  
  By I.H. on the third premise, we have $\jstyping{\Gamma,\neg x}{e_2}{T}$.
  
  Thus we derive $\jstyping{\Gamma}{\econd{x}{e_1}{e_2}}{T}$ by \textsc{(T-Cond)}.
  
  \item
  $\footnotesize
  \Rule{Fill-MatL}{ \jctxsharing{\Gamma}{\Gamma_1}{\Gamma_2} \\ \jatyping{\Gamma_1}{x}{\tlist{T}} \\ \jsynth{\Gamma_2,x=0}{T}{e_1} \\ \jsynth{\Gamma_2,x_h:T,x_t:\tlist{T},x=x_t+1}{T}{e_2} }{ \jfill{\Gamma}{\ematl{x}{\ehole}{x_h}{x_t}{\ehole}}{T}{\ematl{x}{e_1}{x_h}{x_t}{e_2}} }
  $
  
  By I.H. on the third premise, we have $\jstyping{\Gamma_2,x=0}{e_1}{T}$.
  
  By I.H. on the fourth premise, we have $\jstyping{\Gamma_2,x_h:T,x_t:\tlist{T},x=x_t+1}{e_2}{T}$.
  
  Thus we derive $\jstyping{\Gamma}{\ematl{x}{e_1}{x_h}{x_t}{e_2}}{T}$.
  
  \item
  $\footnotesize
  \Rule{Fill-Let}{ \jctxsharing{\Gamma}{\Gamma_1}{\Gamma_2} \\ \jstyping{\Gamma_1}{e_1}{S_1} \\ \jfill{\Gamma_2,\bindvar{x}{S_1}}{\elets{D}{\hole{e_2}}}{T}{e_2} \\ \jwftype{\Gamma}{T} }{ \jfill{\Gamma}{\elets{D; x \gets e_1}{\hole{e_2}}}{T}{\elet{e_1}{x}{e_2}} }
  $
  
  By I.H. on the third premise, we have $\jstyping{\Gamma_2,\bindvar{x}{S_1}}{e_2}{T}$.
  
  Thus we derive $\jstyping{\Gamma}{\elet{e_1}{x}{e_2}}{T}$ by \textsc{(T-Let)}.
  
  \item
  $\footnotesize
  \Rule{Fill-Let-Emp}{ \jfill{\Gamma}{\hole{e}}{T}{e} }{ \jfill{\Gamma}{\elets{\cdot}{\hole{e}}}{T}{e} }
  $
  
  By I.H. on the premise.
  
  \item
  $\footnotesize
  \Rule{Syn-Atom}{ \jasynth{\Gamma}{T}{\elets{D}{a}} }{ \jsynth{\Gamma}{T}{\fold{\elets{D}{a}}} }
  $
  
  Appeal to \lemref{atomic-synth}.
\end{itemize}
\end{proof}

\begin{theorem}[Soundness of Synthesis]\label{Thm:SynthSound}
If \jsynth{\Gamma}{S}{e}, then \jstyping{\Gamma}{e}{S}.
\end{theorem}
\begin{proof}
By Lemma~\ref{lem:synth}.
\end{proof}

\fi

\end{document}